\documentclass[11pt]{article}
\usepackage[utf8]{inputenc}
\usepackage[margin=1in]{geometry}
\usepackage{rotating}
\usepackage{graphicx}
\usepackage{hyperref}
\usepackage[title,titletoc,toc]{appendix}
\usepackage{amsfonts}
\usepackage[font=small, justification=centering]{caption}
\usepackage{enumerate}
\usepackage{amsmath, amssymb,amsthm}
\usepackage{mathptmx}
\usepackage{nicefrac}
\usepackage{bm}
\usepackage{longtable}
\usepackage{color,soul}
\usepackage{graphicx,wrapfig}
\usepackage[dvipsnames]{xcolor} 
\usepackage{algorithm}
\usepackage{algcompatible}
\usepackage[noend]{algpseudocode}

\usepackage{wrapfig}
\usepackage{MnSymbol}
\usepackage{url}
\usepackage[title]{appendix}
\usepackage{booktabs}
\usepackage{multirow}
\usepackage[flushleft]{threeparttable}
\usepackage{changepage}
\usepackage{setspace}
 
\usepackage{enumitem}
\usepackage{comment}
\usepackage{placeins}
\usepackage{setspace}

\usepackage[parfill]{parskip}

\newcommand{\chs}{\color{black}}
\newcommand{\che}{\color{black}}

\newtheorem{lemma}{Lemma}[section]
\newtheorem{theorem}{Theorem}[section]

\newtheorem{prop}{Proposition}[section]
\newtheorem{cor}{Corollary}[section]

\begin{document}

\title{Dynamic discretization discovery under hard node storage constraints}

\author{Madison Van Dyk$^{a,}$\thanks{\noindent Dept. of Combinatorics and Optimization, University of Waterloo, Waterloo, ON, N2L 3G1, Canada. 
\newline \indent ~~Our work was sponsored by the NSERC Discovery Grant Program, grant number RGPIN-03956-2017.
\newline 
\indent $~^{a}$Corresponding author. Email address:  madison.vandyk@uwaterloo.ca.} \and Jochen Koenemann$^*$}
\date{}
\maketitle

\begin{abstract}
\parskip=.5\baselineskip
\noindent The recently developed dynamic discretization discovery (DDD) is a powerful method  that allows many time-dependent problems to become more tractable. While DDD has been applied to a variety of problems, one particular challenge has been to deal with storage constraints without leading to a weak relaxation in each iteration. Specifically, the current approach to deal with certain hard storage constraints in continuous settings is to remove a subset of the storage constraints completely in each iteration of DDD.

\noindent In this work, we show that for discrete problems, such weak relaxations are not necessary. Specifically, we find bounds on the additional storage that must be permitted in each iteration. We demonstrate our techniques in the case of the classical universal packet routing problem in the presence of bounded node storage, which can currently only be solved via integer programming. We present computational results demonstrating the effectiveness of DDD when solving universal packet routing.

\end{abstract}

\section{Introduction}
\noindent Many practically important applications, in areas spanning delivery planning and communication systems, can be cast as network design \chs and routing \che problems \cite{SkutellaSurvey}. Practical \chs routing \che questions are often complicated by added temporal considerations; decisions often need to account for transit times, arc throughput, and node storage. 

There are two main model types used to solve \chs routing \che problems with additional temporal considerations: continuous formulations and time-indexed formulations. \emph{Continuous} formulations use continuous variables to model timing decisions, whereas \emph{time-indexed} formulations have variables and constraints indexed by each time point a decision could be made \cite{LagosDDD}. A  time-indexed formulation for a \chs (discrete-time) \che temporal problem is obtained by expressing the temporal problem as a static problem in the corresponding \emph{time-expanded network}. Let $D = (N, A)$ be a directed graph with arc transit times $\tau$, and let $T$ be the time horizon under consideration. The corresponding time-expanded network $D_T = (N_T, A_T)$ consists of a copy of each node $v \in N$ for each time point \chs $t \in [T]:=\{0,1, \ldots, T\}$\che, and a copy of each arc for each departure time. 

While continuous formulations are more compact than time-indexed formulations, they require ``big-M" constraints which lead to slow solve times due to large branch-and-bound trees \cite{LagosDDD}. Time-indexed formulations have stronger LP relaxations, but this time-expansion often makes the problems impractical to solve since the network grows linearly in $T$ \cite{FordFulkerson58, FordFulkerson59}. 

Reducing the size of time-expanded network formulations has been an active area of research in both theoretical and applied optimization.  Boland et al. \cite{Boland1, DDD_TSP} recently introduced the framework of \emph{dynamic discretization discovery} (DDD) which solves certain classes of temporal problems by only using a subset of the time-indexed variables and constraints. A DDD algorithm solves a series of integer programs (IPs) defined on networks that include only a subset of the \chs nodes and arcs \che in the full time-expanded network, called \emph{partially time-expanded networks}. The partially time-expanded networks are constructed to ensure that \chs the corresponding formulations \che provide a lower bound on the optimal value of the original problem. If a solution to the partially time-expanded network cannot be converted to a solution to the original instance of equal cost, we refine the partially time-expanded network by adding nodes and arcs. The main advantage of DDD is that it determines which time-points are needed to solve the problem, without \chs possibly \che ever constructing the full time-expanded network. 

In the past few years, DDD has been applied to a number of problems including the \emph{travelling salesman problem with time windows} (TSP-TW)\cite{DDD_TSP} and the \emph{continuous time service network design \chs (SND) \che problem} \cite{Boland1}. For these connectivity problems, there are no storage capacity constraints at nodes. In turn, the proofs for the correctness of the lower bounds rely on this freedom so that flow arriving ``early'' does not lead to infeasibility. This need for \chs unlimited node storage has  previously \che prevented the DDD model from addressing many real-world problems, such as certain dynamic scheduling problems for which only heuristic techniques are known for large instances \cite{amazon}. \chs Removing DDD's reliance on unbounded zero-cost storage was noted as an important direction of future research by Boland and Savelsbergh \cite{DDD_survey}\che. Lagos et al. \cite{LagosDDD} recently addressed this shortcoming and extended the DDD approach to solve the \chs \emph{continuous inventory routing problem} (CIR) \che with out-and-back routes, which has storage capacity constraints. However, in their continuous setting they relax a subset of the storage constraints fully in each iteration. In Appendix \ref{app:improved_bound}, we show that such an approach can lead to DDD requiring $\Omega(T)$ iterations for the discrete setting where time points are in $\{0,1,\ldots, T\}$, whereas with our improved bounds DDD would terminate in a single iteration. 

\chs While DDD is often applied to continuous problems without a prespecified discretization, many practical routing problems are in fact discrete in nature, such as the dynamic scheduling problem considered by Lara et al. \cite{amazon}. Hewitt and Crainic \cite{hewitt_crainic} also state that in many applications, continuous problems are modelled as discrete problems where the time granularity is chosen based on the application at hand. For instance, there is often a minimum discretization of time, such as an hour or 15 minutes, that is useful in practice. With this in mind, we tighten the relaxation of storage constraints of Lagos et al. \cite{LagosDDD} for problems with discrete time discretizations. Specifically, we argue that it is not necessary to completely remove a subset of the node storage constraints in each iteration and we prove bounds on the storage required based on the capacities and structure of the base graph. \che At the heart of the challenge to solving the dynamic scheduling problem considered by Lara et al. in \cite{amazon} lies a packet routing problem. \chs Since the main contribution of this paper is the further development of the DDD method, we focus on one particular problem, the so-called \emph{universal packet routing} problem (UPR). This allows us to clearly present the extension of the DDD method to a problem with bounded node storage without messy complicating constraints already addressed by previous work. Additionally, our bounds apply to any underlying graph structure, in contrast to the bounds presented in \cite{LagosDDD} that only apply to graphs that have an out-and-back network structure (i.e. stars). \che

In this work we extend the DDD model to address UPR, a dynamic scheduling problem with bounded storage which involves routing and scheduling packets through a communication network. Due to the bounded storage assumption, the only currently known approach to solve this problem is via integer programming. In UPR, we are given a directed graph $D = (N,A)$ which we will call the \emph{flat} or \emph{base} network. Each arc $a \in A$ has an associated transit time $\tau_a \in \mathbb{N}$, and a capacity $u_a \in \mathbb{N}$ which denotes the maximum number of packets that can traverse arc $a$ simultaneously. Let $\mathcal{K}$ denote a set of $k$ packets, each with an associated origin \chs $s_k$ and  destination $t_k$\che. We say that a packet is \emph{active} if it is not located at its origin or destination. Additionally, each node $v \in N$ has storage capacity \chs $b_v \in \mathbb{N}$\che, meaning that it can store at most $b_v$ active packets at any time. The makespan of a schedule is the latest arrival time of any packet at its destination. The objective of UPR is to route all packets through the network in order to minimize the makespan of the schedule, while respecting arc capacities and node storage levels at every point in time. This problem is NP-hard, and even the special case where $b_v = 0$ for all $v \in N$ is as hard to approximate as vertex colouring \cite{Busch2004DirectRA}. \chs That is, it is hard to approximate within $\Omega(n^{1-\epsilon})$ for any $\epsilon > 0$, assuming \textsc{NP $\subsetneq$ ZPP}, where $n$ denotes the number of vertices in the graph \cite{Feige}\che. Traditionally, arc capacities and node storage levels are referred to as bandwidth and buffers respectively. \chs We may assume we know an upper bound $T$ on the minimum makespan of the schedule, $T^*$, via running for example a greedy algorithm.\che
\subsection{Our Contributions}
The contributions in this paper are both algorithmic and theoretical. \chs  We show that in many discrete settings, it is not necessary to remove a subset of the node storage constraints in the lower bound model. Instead, we prove upper bounds on the storage that must be permitted at each timed node based on the current time-expanded network and the arc and storage capacities in the underlying network. Additionally, we extend the framework of DDD to address a temporal problem with bounded node storage levels on a general graph rather than a restricted out-and-back structure where only a single node has degree greater than 1. \che

\chs In Appendix \ref{app:improved_continuous}, we show that our bounds for relaxed node capacities can be extended to the continuous setting and demonstrate that our construction generalizes and tightens the relaxation of vehicle waiting constraints presented by Lagos et al. \cite{LagosDDD}. Specifically, we show that in the continuous setting, in each iteration of DDD our construction leads to the removal of a subset of the constraints removed in \cite{LagosDDD}, and in certain settings this subset is strict. Our results generalize the results of Lagos et al. since the graph is no longer restricted to a star. \che The main ingredients of our contributions are as follows:

\begin{enumerate}[noitemsep]
    \item We develop and implement a lower bound model and refinement process for time-indexed problems with bounded node storage;
    \item To prove that the lower bound model is in fact a relaxation, we present arguments relying on structural observations of the map from the fully time-expanded network to the partially time-expanded network;
    \item \chs We prove that with our lower bound and refinement process, the algorithm terminates with an optimal solution in at most $|N|T^*$ iterations, where $N$ is the set of nodes in the base graph and $T^*$ is the minimum makespan;\che
    \item We implement and test our DDD algorithm on two classes of instances: one based on the population centres of the United States, and the other based on social networks. We demonstrate that our DDD algorithm \chs completes in an average of 53\% of the time of solving the full time-indexed formulation when the known upper bound is 2$T^*$ for a class of geographic instances, and 49\% of the time for a class of geometric instances. We also show that the DDD algorithm performs better when the underlying graph is sparse.\che
\end{enumerate}

\chs Traditionally, DDD has addressed network design problems where the task is to route flow through a network and purchase capacity along arcs at specific times to facilitate that flow. In the UPR problem, the network with capacities is given as an input, and the task is to find an optimal routing through this network. Our application of DDD exhibits further evidence of the potential for DDD to allow a variety of temporal problems to become more tractable. \che

\subsection{Other related work}\label{sec:lit_rev}
Temporal network design \chs and routing \che problems were first introduced by Ford and Fulkerson \cite{FordFulkerson58, FordFulkerson59} in the context of network flow theory. Ford and Fulkerson showed that these ``flow over time'' problems can be reframed as static network flow problems in the corresponding time-expanded network. For a general background on temporal flows, we refer the reader to a recent survey by Skutella \cite{SkutellaSurvey}. In the case of multicommodity flows, Hall et al. \cite{Hall2003} provide hardness proofs as well as polytime \chs solvable \che instances. The problem of UPR considered in this paper is a multicommodity flow over time problem with the additional constraint that flow values are integer. The theoretical and algorithmic techniques presented in this paper can be extended to more general fractional variants of UPR. 

For temporal network flows, which permit fractional values, Fleischer et al. \cite{Fleischer2003, FleischerSkutella2007} provide guarantees on the cost increase of the optimal solution when we allow a coarser network and only include vertex copies for every $\Delta$ units of time. In this $\Delta$-condensed approach, each node shares the same discretization $\Delta$, as opposed to a partially time-expanded network in DDD where the discretization for each node is non-uniform. Wang and Regan \cite{WangRegan2} show that iteratively refining a time window discretization for TSP-TW will converge to an optimal solution. Similarly, Dash et al. \cite{Dash} iteratively refine a set of time periods based on a preprocessing scheme in contrast to the dynamic scheme in DDD. 

\subsubsection*{Dynamic discretization discovery}
Initial applications of the DDD framework addressed connectivity problems such as the shortest path problem \cite{DDDpaths} and the travelling salesman problem with time windows \cite{DDD_TSP}. In the case of SND \cite{Boland1} where trucks have capacities, there is no bound on the number of trucks that can travel along a specified arc at any given time. \chs The same assumption is made by Scherr et al. \cite{DDDAuto} and Hewitt \cite{flexible} when applying DDD to variants of SND\che. These assumptions avoid complicating capacity constraints that could be problematic when mapping a solution from the fully time-expanded network to the partially time-expanded network\chs, and vice versa\che. For a complete presentation of the DDD framework, we refer the reader to the survey of Boland and Savelsbergh \cite{DDD_survey}. 

In the scheduling problem with time-dependent durations and resource consumptions constraints considered by Pottel and Goel \cite{Pottel}, the resource constraints at nodes can be encoded using arc capacities.  Lagos et al. \cite{LagosDDD} consider the \emph{continuous inventory routing problem} (CIR), in which a company manages the inventory of its clients, and delivers product from a single facility. \chs Each delivery is restricted to serving a single client, and then the truck must return to the facility. Thus, their results only apply when the graph is a star. The authors also encode two different storage capacity constraints. \che While each client has a storage limit, the authors assume that products that arrive at a client location do not impact the storage level unless a delivery is scheduled. In many problems, including UPR, all stationary flow must count towards the storage level at some node. \chs This aligns more closely with the constraint of Lagos et al. \cite{LagosDDD} that at most one vehicle can visit a fixed client at any point in time. To relax this constraint in their lower bound model, the authors completely remove a subset of the vehicle storage constraints.  \che

Beyond modelling techniques, there is a burgeoning area of research addressing DDD from the standpoint of algorithm engineering \cite{Interval_DDD, DDDAuto, Vu_hewitt}. Marshall et al. \cite{Interval_DDD} introduce the \emph{interval-based dynamic discretization discovery algorithm} (DDDI) which was demonstrated to find solutions to instances of SND orders of magnitude faster than traditional DDD. Scherr et al. \cite{DDDAuto} suggest removing nodes and arcs if they are no longer required for a high quality solution. Hewitt \cite{enhanced} explores speed-up techniques for DDD which include enhancements such as a two-phase implementation of DDD and the addition of valid cuts to strengthen the relaxed model in each iteration.

\subsubsection*{Packet routing}\label{sec:lit_rev_packet}
Packet routing in the literature refers to a broad range of problems, closely related to our definition of UPR in this paper. In \emph{store-and-forward packet routing} (SF-PR), arcs can only accommodate a single packet at any given time, and transit times are unit length. Additionally, each packet has a specified path it must follow in the underlying network. The \emph{congestion} $C$ denotes the maximal number of paths using a single arc in the base graph, and the \emph{dilation} $D$ denotes the maximal length of a path along which a packet must be routed. An $O(C+D)$ approximation for SF-PR was originally developed by Leighton et al. \cite{Leighton1, Leighton2}. Building upon these initial results, various papers have established polytime approximation results for SF-PR with arbitrary arc capacities and transit times \cite{UniversalPR, Scheideler}. However, the constants in these guarantees remain large for general graphs \cite{UniversalPR} \chs($39(C+D)$ for general graphs, and $23.4(C+D)$ when there are unit transit times and arc capacities)\che. SF-PR was proven to be NP-hard by Di Ianni \cite{DiIanni}. Peis et al. \cite{packetcomplexity, UniversalPR} generalized SF-PR to include arbitrary arc capacities and transit times and proved that this problem is APX-hard. We note that our UPR problem is more general since we add node storage constraints. In this paper, we refer to the variant of UPR where each packet is given a designated path in the base network as the problem of \emph{universal packet routing with fixed paths} (UPR-FP).

Current approximation strategies for packet routing problems (with variable paths) involve converting the problem to an instance of UPR-FP by selecting an appropriate path for each packet. Busch et al. \cite{Busch2004DirectRA} consider the \emph{bufferless} packet routing problem, in which once a packet is injected into the network, it cannot be stored at any node. The authors prove that this problem is not only NP-hard, but as hard to approximate as vertex colouring. Vertex colouring is hard to approximate within $\Omega(n^{1-\epsilon})$ for any $\epsilon > 0$, assuming \textsc{NP $\subsetneq$ ZPP}, where $n$ denotes the number of vertices in the graph \cite{Feige}. \chs In order to ensure that a feasible solution to an instance of UPR is indeed contained in some time-indexed formulation, we must allow time-indexed variables for times up to some known upper bound on the min makespan, $T^*$. By our previous discussion, it is not reasonable to expect that we know an upper bound $T$ on $T^*$ where $T \approx T^*$. This motivates the importance of applying DDD to solve this problem, since we find that the upper bound provided has much less of an impact on the runtime of DDD compared to solving on the full time-expanded network. \che

\subsection{Roadmap}
The remainder of this paper is organized as follows. In Section \ref{sec:intro}, we introduce the problem of UPR along with a time-indexed formulation. In this section we also present the general framework of DDD. In Section \ref{sec:LB_model}, we describe the lower bound model and introduce additional properties that are sufficient to ensure the lower bound model indeed gives a lower bound on the optimal value of the UPR instance. In Section \ref{sec:UB_Aug}, we outline the upper bound and augmentation steps necessary to either terminate the DDD algorithm, or produce the subsequent partially time-expanded network. In Section \ref{sec:computations} we demonstrate the effectiveness of the DDD algorithm on a set of randomly generated test cases. We also examine the experimental results and observe trends in the effectiveness of the DDD algorithm given input parameters. 

\chs In Appendix \ref{app:improved_bound} we demonstrate how our tightened relaxation of the storage bounds can significantly impact the number of iterations until DDD terminates. In Appendix \ref{app:improved_continuous} we outline how the relaxed storage bounds presented in Section \ref{sec:LB_model} generalize the relaxation of storage presented by Lagos et al. \cite{LagosDDD}. Finally, in Appendix \ref{app:2_stage} we discuss the performance of the two-phase DDD approach on the geographic instances.\che

\section{Problem statement and background}\label{sec:intro}
Before defining the problem of packet routing, we introduce the concept of time-expanded networks. 

\subsection{Time-expanded networks}
Let $D = (N, A)$ be a directed graph with arcs labelled with transit times $\tau$, and let $T$ be the time horizon under consideration. The corresponding \emph{(fully) time-expanded network} $D_T = (N_T, A_T)$ 
consists of a copy of each node $v \in N$ for each time point \chs $t \in [T]:= \{0, 1, \ldots, T\}$\che, as well as a copy of each arc for each departure time. Specifically, $N_T = \{(v,t): v \in N, t \in [T]\}$ and $A_T = \{((v,t), (w, t+ \tau_{vw})): (v,t) \in N_T, vw \in A, t+ \tau_{vw} \leq T\}$. For example, given the digraph in Figure 1 along with a time horizon $T = 3$, we obtain the time-expanded network provided in Figure 2. When storage is permitted at nodes, we also add a \emph{holdover} arcs, \chs $H_T = \{((v,t), (v, t+1): v \in V, t \in [T-1]\}$, and write $D_T = (N_T, A_T \cup H_T)$\che. We will refer to arcs in $A_T$ as \emph{movement} arcs. \chs A \emph{partially time-expanded network} with respect to $D$ and $T$ is any directed graph $D_S = (N_S, A_S \cup H_S)$ where $N_S \subseteq N_T$, $A_S \subseteq \{((v,t), (w, t')): (v,t) \in N_T, vw \in A, t' \leq t+ \tau_{vw} \leq T\}$, and $H_S$ connects each node copy to its next copy in $N_S$. \che
\begin{figure}[!htb]
    \centering
    \begin{minipage}{.5\textwidth}
        \centering
        \vspace{.6cm}
        \includegraphics[width=.53\textwidth]{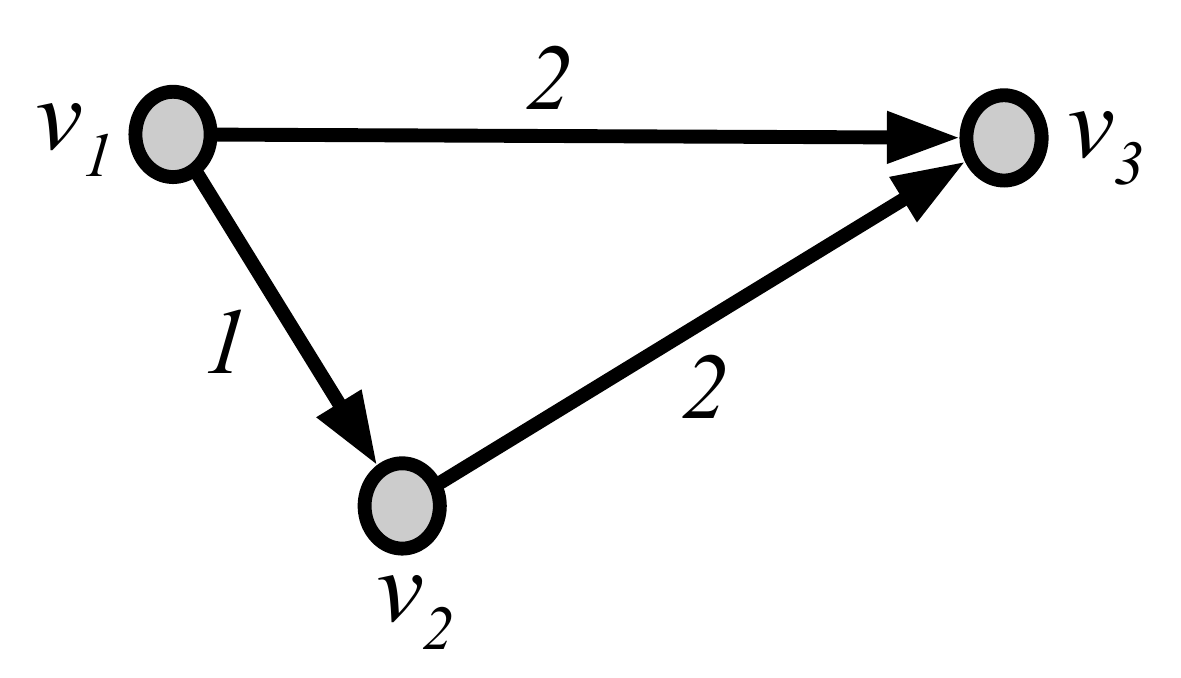}
        \caption{Base directed graph.}
        \label{fig:full}
    \end{minipage}%
    \begin{minipage}{0.5\textwidth}
        \centering
        \includegraphics[width=.4\textwidth]{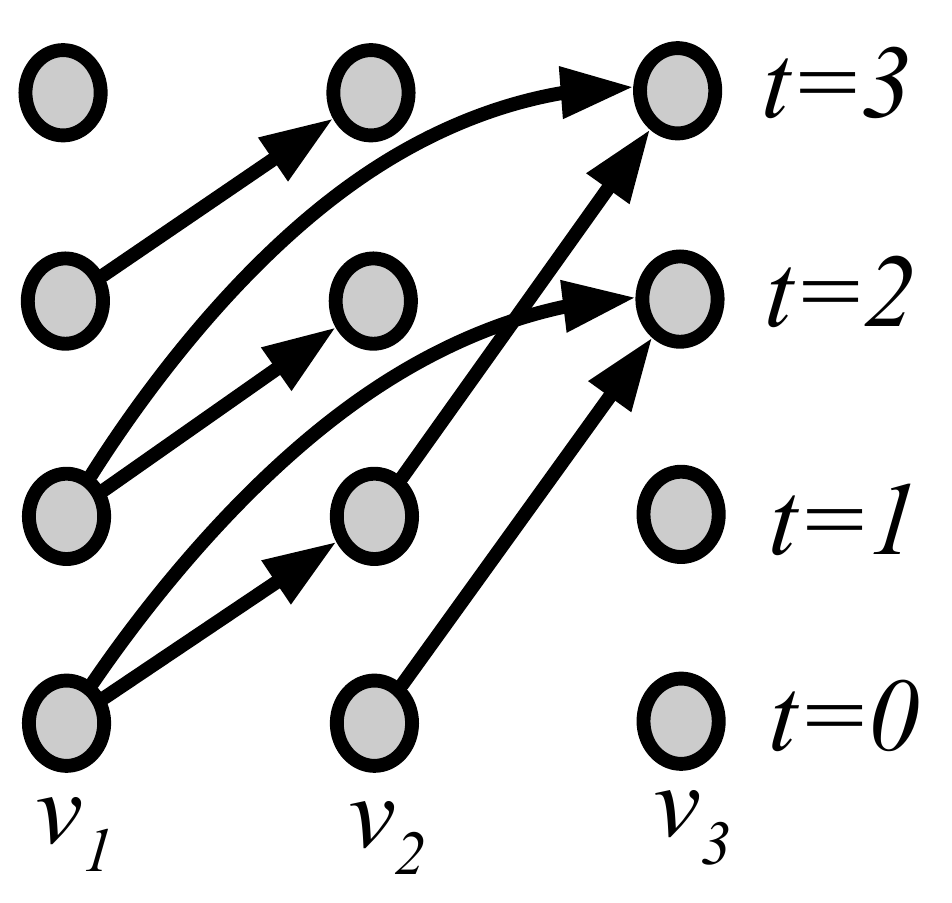}
        \caption{Corresponding time-expanded graph.}
        \label{fig:partial}
    \end{minipage}
\end{figure}

The advantage of modelling a temporal problem with a time-expanded network is that we can recast it as a static problem and apply techniques used to solve static network problems. However, this time-expansion often makes the resulting static problem impractical to solve since the network grows linearly in $T$ \cite{FordFulkerson58, FordFulkerson59}. In DDD, we reduce the size of the time-expanded network \chs by constructing partially time-expanded networks that appropriately \emph{underestimate} the transit times of the arcs in $A_T$\che. In this paper, $D_T = (N_T, A_T \cup H_T)$ denotes the fully time-expanded network with time horizon $T$, and $D_S = (N_S, A_S \cup H_S)$ is any partially time-expanded network. We let $D_{\infty}$ denote the infinite fully time-expanded network where $T = \infty$. Note, the abstract notion of $D_{\infty}$ is only used to define the packet routing problem. For clarity, from now on we will refer to arcs and nodes in a time-expanded network as \chs \emph{timed arcs} and \emph{timed nodes} \che respectively. We will write timed nodes with their associated times as $(v,t)$, and the corresponding node in the base graph is $v$. Similarly, we write timed arcs as $e = ((v,t), (w,t'))$, and the corresponding arc in the base graph is $a=vw$. Additionally, we will refer to paths in the time-expanded network as \emph{trajectories}.

\subsection{Universal packet routing}
Let $D = (N,A)$ be a directed graph which we will call the \emph{flat} or \emph{base} network. Each arc $a \in A$ has an associated transit time $\tau_a \in \mathbb{N}$, and a capacity $u_a \in \mathbb{N}$ which denotes the maximum number of packets that can depart along arc $a$ simultaneously. Let $\mathcal{K}$ denote a set of \chs packets and for each packet $k \in \mathcal{K}$, let $s_k$ and $t_k$ denote its associated origin and destination respectively\che. We say that a packet is \emph{active} if it is not located at its origin or destination. Additionally, each node $v \in N$ has a storage level of \chs $b_v \in \mathbb{N}$\che, meaning that it can store at most $b_v$ active packets at any time. The makespan of a schedule is the latest arrival time of any packet at its destination. The objective of universal packet routing (UPR) is to send each packet along a single trajectory in $D_{\infty}$ that minimizes the makespan of the schedule while respecting arc capacity and node storage. We let $T^*$ denote the minimum makespan.

We note that the techniques presented in this paper extend to the setting where packet sizes are also of arbitrary size, as well as the setting where packets have varying release times. We will occasionally refer to the packets as commodities throughout this paper.

\subsubsection*{Challenge of solving this problem in practice}
This problem could be formulated with continuous variables to model time, however, the need to keep track of when two packets meet at the same arc or node requires ``big-M" constraints. As a result, the LP relaxation of these formulations are weak, and historically continuous formulations for temporal network design problems have been found to perform poorly in practice \cite{DDD_survey}. Instead, the natural strategy is to encode these time-dependent problems as static problems in time-expanded graphs \chs since these formulations have strong relaxations\che.

However, in order to map an instance of UPR to a finite time-expanded network, we need to know some upper bound $T$ on the value of the minimum makespan, $T^*$. Unfortunately, UPR is at least as hard to approximate as vertex colouring \cite{Busch2004DirectRA}. Thus, in practical applications we are forced to use a relatively large value of $T$ as an upper bound to ensure a solution can be found in the corresponding fully time-expanded network $D_T$. We show in Section \ref{sec:computations} that this can greatly increase solving time, partly due to the introduction of additional symmetries in the corresponding MIP as $T$ increases. \chs When solving an instance of UPR that arises from a practical application, we may have solved similar instances in the past and as a result know a good upper bound $T$ on $T^*$, where $T = \alpha T^*$ for a small value of $\alpha$. In contrast, if the instance is unknown to us, the upper bound we can produce would likely require a large value for $\alpha$\che. In our computational experiments, we therefore test the proposed algorithm with upper bounds of $T^*, 1.5 T^*$, and $2 T^*$ \chs to understand the performance of DDD compared to the traditional MIP as the strength of the known upper bound varies\che. 

\subsubsection*{IP for the fully time-expanded network}
For each $v \in N$, let $\mathcal{K}_v$ denote the set of commodities that are \emph{active} at $v$. That is, $\mathcal{K}_v = \{k \in \mathcal{K}: \chs s_k \che \neq v,\chs t_k \che \neq v\}$. Note that the commodities in $\mathcal{K} \setminus \mathcal{K}_v$ do not contribute to storage levels at node $v$, since they are either at their origin or destination. \chs We emphasize that $\mathcal{K}_v$ is \emph{not} time-dependent.\che

Let $T$ be an upper bound on the value of $T^*$, which we will assume is given to us. We would like to determine trajectories in $D_T$ that minimizes makespan while obeying the capacity limitations. \chs First, note that since all input data is integer, the decision times of an optimal solution are in $[T]$\che. For each packet $k \in \mathcal{K}$ and each \chs timed arc \che $e \in A_T \cup H_T$, we have a binary variable $x_e^k$ which is equal to 1 if packet $k$ is scheduled to travel along \chs timed arc \che $e$ in its assigned trajectory in $D_T$. We assign the timed arc capacities in $D_T$ directly from the arc and node capacities in $D$. Specifically, for each timed arc $e = ((v,t),(w,t')) \in A_T$, we define $u_e := u_{vw}$. Similarly, for $e = ((v,t)(v,t')) \in H_T$, we define $b_e := b_v$. 

Let $\delta^+_{D_T} (v,t)$ and $\delta^-_{D_T} (v,t)$ denote the outgoing and incoming \chs timed arcs \che at $(v,t)$ in $D_T$. That is, $\delta^+_{D_T} (v,t) = \{e \in A_T \cup H_T: e = ((v,t),(w,t'))\}$, and $\delta^-_{D_T} (v,t) = \{e \in A_T \cup H_T: e = ((w,t'),(v,t))\}$. The following IP models UPR when we are given the upper bound $T$.
\begin{align}\tag{\chs UPR($D_T$)\che}
    \min~~ & \bar{T} \\
    \vspace{.5cm}
    \mbox{s.t.} ~~
    & t' \cdot x_e^k \leq \bar{T} \quad \forall k \in \mathcal{K}, ~\forall e = ((v,t), (w,t')) \in A_T \label{constr:IPT_timing}\\
    & \sum_{e \in \delta_{D_T}^{+}(v,t)} x_e^k - \sum_{e \in \delta_{D_T}^{-}(v,t)} x_e^k = \begin{cases}
        1 ~(v,t) = (s_k, 0) \\
        -1 ~(v,t) = (t_k, T) \\
        0 ~\mbox{otherwise} \\
    \end{cases} \quad \forall k \in \mathcal{K}, (v,t) \in N_T \label{const:flow}\\
    \vspace{.5cm}
    & \sum_{k \in \mathcal{K}} x_e^k \leq u_e \quad \forall e \in A_T \label{const:IPTarc_cap}\\
    & \sum_{k \in \mathcal{K}_v} x_e^k \leq b_e \quad \forall ~e \in H_T \label{const:IPTstorage}\\
    & x_e^k \in \{0,1\} \quad \forall k \in \mathcal{K}, ~\forall e \in A_T \cup H_T \label{const:integrality}.
\end{align}

Constraint ($\ref{constr:IPT_timing}$) encodes that the time horizon is equal to the latest arrival time among all packets. Constraint ($\ref{const:flow}$) ensures that the set of trajectories satisfies flow conservation. Finally, constraints ($\ref{const:IPTarc_cap}$) and ($\ref{const:IPTstorage}$) ensure that the trajectories satisfy arc capacities and node storage constraints respectively.

\chs In addition, we will add the following constraint to strengthen the LP relaxation. This turns out to be essential when applying the two-phase DDD approach presented in Appendix \ref{app:2_stage}. For each $k \in \mathcal{K}$, let $A_{T, k}^{final} = \{((v,t),(w, t')) \in A_T: w = t_k\}$ be the set of timed movement arcs entering the destination of packet $k$. Since each packet has a single trajectory in a feasible integer solution, we have the following constraint.
\begin{align}
    & \sum_{e=((v,t),(w,t')) \in A_{T, k}^{final}} t' \cdot x_e^k \leq \bar{T} \quad \forall k \in \mathcal{K}. \label{constr:two_stage_IP}
\end{align}
\che
\subsection{DDD framework}\label{sec:DDD_framework}
In general, a DDD algorithm aims to solve a minimization problem $P$ defined on some fully time-expanded network, $D_T$, with a corresponding mixed integer program (MIP), \chs IP$(D_T)$\che. In DDD, partially time-expanded networks, denoted $D_S$, are constructed and refined in such a way that they are sparse relative to the full time-expanded network. Additionally, each partial network is constructed so that an optimal solution to the MIP induced by $D_S$, denoted IP$(D_S)$, provides a lower bound on the optimal value of \chs IP$(D_T)$\che. When this relationship holds, we will refer to $D_S$ as a \emph{relaxation} of $D_T$.

In a partially time-expanded network, arc transit times are \emph{underestimated} so that every trajectory in $D_T$ can be mapped to a trajectory in $D_S$ with the same underlying path, albeit with shortened \chs timed arcs\che. In each iteration we solve $D_S$ and check if the solution to $D_S$ we found can be mapped to a solution to $D_T$ of equal cost.  The fact that arcs in $D_S$ do not all have realistic lengths is one reason why it may not be possible to obtain a corresponding solution in $D_T$. If we cannot obtain a corresponding solution in $D_T$, we refine the partially time-expanded network by lengthening short \chs timed arcs and adding timed nodes and timed arcs \che to $D_S$. Boland et al. \cite{Boland1} introduced the following key properties that ensure that $D_S$ is a relaxation, when the original problem has time horizon $T$, and each commodity $k \in \mathcal{K}$ has release time $r_k$ and deadline $d_k$. 

\subsubsection*{Standard DDD properties:}
\begin{enumerate}[noitemsep]
    \item[($P1$)] For all commodities $k \in \mathcal{K}$, the nodes $(s_k, r_k)$ and $(t_k, d_k)$ are in $N_S$. 
    \item[($P2$)] Every arc $((v,t), (w, \bar{t})) \in A_S$ has $\bar{t} \leq t + \tau_{vw}$. 
    \item[($P3$)] For every arc $a = (v,w) \in A$ in the flat network, and for every node $(v,t)$ in the partially time-expanded network $D_S = (N_S, A_S \cup H_S)$ with $t + \tau_{vw} \leq T$, there is a timed-copy of $a$ in $A_S$ starting at $(v,t)$. 
    \item[($P4$)] If arc $((v,t), (w,t')) \in A_S$, then there is no node $(w,t'')$ in $N_S$ with $t' < t'' \leq t + \tau_{vw}$. 
\end{enumerate}

For universal packet routing, we will assume we are given an upper bound, $T$, on the minimum required makespan, and we set $r_k = 0$ and $d_k = T$ for all $k \in \mathcal{K}$. Note that these properties are not sufficient for $D_S$ to be a lower bound model for $D_T$ in the case of universal packet routing, due to the arc and node capacity constraints. \chs For example, if the timed arcs in $D_S$ were all given the same capacity as their underlying arc in $D$, then $D_S$ would have a smaller total arc capacity than $D_T$\che. We will describe the additional properties necessary in Sections \ref{sec:arc_caps} and \ref{sec:storage}. We remark that the techniques presented in this paper are quite general and can easily be extended to the setting where $r_k$ and $d_k$ vary, as well as the case where commodities have non-unit demands \chs and flow for a single commodity can be sent fractionally along multiple trajectories\che.

\subsubsection*{Map from $D_T$ to $D_S$}
When proving that $D_S$ is a relaxation of $D_T$, we need to map a solution of \chs IP($D_T$) \che to a solution of \chs IP($D_S$) \che with no greater cost. The following map, $\mu: A_T \rightarrow A_S$, is a standard tool in the DDD literature when proving that the lower bound model is a relaxation. For any \chs timed arc \che $e \in A_T$, $\mu$ maps the flow on $e$ to an arc $\mu(e)$ in $A_S$. Specifically, 
\begin{equation}\label{mu_eq}
    e = ((v,\bar{t}), (w,\bar{t}')) \quad \rightarrow  \quad \mu(e) = ((v,\hat{t}), (w, \hat{t}')),
\end{equation}  
where $\hat{t} = \max\{t: t \leq \bar{t}, (v,t) \in N_S\}$, and $\hat{t}' = \max\{t: t \leq \hat{t} + \tau_e, (v,t) \in N_S\}$. Observe that $\hat{t}'$ is dependent on $\hat{t}$ rather than $\bar{t}'$. We note that for UPR, we use \chs UPR($D_T$) \che in place of \chs IP($D_T$) \che and UPR($D_S$) (presented in Section \ref{sec:LB_model}) in place of IP($D_S$). 

\subsection{Outline of the DDD approach for universal packet routing}
As described above, the key components of the DDD approach are the lower bound model, the upper bound/termination step, and the augmentation/refinement step. In Section \ref{sec:LB_model} we present two additional properties to ensure that $D_S$ is a relaxation of $D_T$ when the underlying problem is universal packet routing. \chs Our results allow DDD to be more effective for solving flow problems with storage constraints since our improved relaxation of node storage reduces the number of iterations required until DDD terminates for certain instances (see Appendix \ref{app:improved_bound})\che. In Section \ref{sec:UB_Aug}, we describe the upper bound model and augmentation steps. We then proceed to present computational results in Section \ref{sec:computations}. 

A sketch of the overall DDD algorithm for universal packet routing is as follows. We present the precise algorithm at the end of Section 5. The subroutines along with the definition of UPR$(D_S)$ are presented in Sections 4 and 5.

\begin{algorithm}[h]
\hspace*{\algorithmicindent} \textbf{Input}: Base network $D = (N,A)$, packet set $\mathcal{K}$, an upper bound, $T$, on the optimal makespan
\begin{algorithmic}[1]  
	\State Create initial partially time-expanded network $D_S$ satisfying lower bound properties
	\While{not solved}
	    \State Solve UPR($D_S$)
        \State Determine if the solution can be converted to a solution to UPR($D_T$) without increasing the 
        \Statex \hskip\algorithmicindent time horizon
        \If{the solution to the partially time-expanded network can be converted}
            \State Stop. An optimal solution has been found for UPR($D_T$).
        \EndIf
        \State Augment the current set of timed nodes $N_S$ by correcting at least one arc in the support of the partial 
        \Statex \hskip\algorithmicindent solution that is either too short, or has exceeded arc or storage capacity
    \EndWhile
	\caption{Solve UPR-DDD($D, \mathcal{K}, T$)}
	\label{alg:da}
 \end{algorithmic}
\end{algorithm}	
\vspace{-.75cm}
\FloatBarrier

\section{Lower Bound Model}\label{sec:LB_model}
One of the key components of the DDD iterative approach is the lower bound model. Specifically, given an appropriate subset of \chs timed nodes \che $N_S$, we want to obtain a partially time-expanded network $D_S = (N_S, A_S \cup H_S)$ along with a formulation UPR($D_S$) which has optimal value at most that of \chs UPR$(D_T)$\che. As is standard in DDD, we construct $A_S$ according to (P1) - (P4). The specific selection of arcs is given at the end of this section in Algorithm \ref{alg:gen_A_S}. We now state the lower bound formulation corresponding to a partially time-expanded network $D_S$, and define $u'$ and $b'$ in Sections \ref{sec:arc_caps} and \ref{sec:storage}. 
\begin{align}\tag{UPR$(D_S)$}
    \min~~ & \bar{T} \\
    \vspace{.5cm}
    \mbox{s.t.} ~~
    & (t + \tau_{vw})  \cdot x_e^k \leq \bar{T} \quad \forall k \in \mathcal{K}, ~\forall e = ((v,t), (w,t')) \in A_S \label{constr:IPS_timing}\\
    & \chs \sum_{e=((v,t),(w,t')) \in A_{S, k}^{final}} (t + \tau_{vw}) \cdot x_e^k \leq \bar{T} \quad \forall k \in \mathcal{K}. \label{constr:two_stage_IP_S}\che\\
    & \sum_{e \in \delta_{D_S}^{+}(v,t)} x_e^k - \sum_{e \in \delta_{D_S}^{-}(v,t)} x_e^k = \begin{cases}
        1 ~(v,t) = (s_k, 0) \\
        -1 ~(v,t) = (t_k, T) \\
        0 ~\mbox{otherwise} \\
    \end{cases} \quad \forall k \in \mathcal{K}, (v,t) \in N_S \label{const:D_S_flow}\\
    \vspace{.5cm}
    & \sum_{k \in \mathcal{K}} x_e^k \leq u'_e \quad \forall e \in A_S \label{const:arc_cap}\\
    & \sum_{k \in \mathcal{K}_v} x_e^k \leq b'_e \quad \forall e \in H_S \label{const:DS_storage}\\
    &x_e^k \in \{0,1\} \quad \forall k \in \mathcal{K}, e \in A_S \cup H_S \label{const:D_S_integrality}.
\end{align}

In addition to modifying the arc and node capacities, we replaced \chs constraints (\ref{constr:IPT_timing}) and (\ref{constr:two_stage_IP}) in UPR$(D_T)$ with (\ref{constr:IPS_timing}) and (\ref{constr:two_stage_IP_S})\che. Observe that in a partially time-expanded network, we may have $t' < t + \tau_{vw}$ for some $((v,t), (w,t')) \in A_S$. Therefore, \chs constraints (\ref{constr:IPS_timing}) and (\ref{constr:two_stage_IP_S}) are tighter than constraints (\ref{constr:IPT_timing}) and (\ref{constr:two_stage_IP})\che. In Section \ref{sec:augment}, we will prove that constraint (\ref{constr:IPS_timing}) ensures that throughout DDD, so long as $\tau_a \geq 1$ for all $a \in A$, we never add a \chs timed node \che $(v,t)$ with $t > T^*$.

We need to assign the arc capacities and node storage levels in $D_S$ to ensure that  UPR($D_S$) is a relaxation of UPR$(D_T)$. As is standard, to prove that the optimal value of UPR($D_S$) is at most that of \chs UPR$(D_T)$\che, we use the map $\mu$, as defined in Section \ref{sec:DDD_framework}, to map feasible solutions of \chs UPR$(D_T)$ \che with makespan $\bar{T}$ to those of UPR($D_S$) with makespan at most $\bar{T}$.

Recall from (\ref{mu_eq}), for any \chs timed arc \che $e = ((v,\bar{t}), (w,\bar{t}')) \in A_T$, $\mu(e) = ((v,\hat{t}), (w, \hat{t}'))$ where $\hat{t} = \max\{t: t \leq \bar{t}, (v,t) \in N_S\}$ and $\hat{t}' = \max\{t: t \leq \hat{t} + \tau_{vw}, (v,t) \in N_S\}$. With this map $\mu$ in mind, we will show how to define $u'$, and $b'$ so that UPR$(D_S)$ is a relaxation of UPR$(D_T)$. Let $\bar{x}$ be a feasible solution to UPR$(D_T)$. We define $\hat{x}$ as the binary vector such that for all $e \in A_S$ and $k \in \mathcal{K}$, 
\[\hat{x}_e^k = \max\{\bar{x}_f^k: \mu(f) = e\}.\]   
That is, for all $e \in A_S$ and $k \in \mathcal{K}$, $\hat{x}_e^k = 1$ if $\mu(f) = e$ for any \chs timed arc \che $f \in A_T$ with $\bar{x}_f^k = 1$. By abuse of notation, $\hat{x}$ obtained from $\bar{x}$ in this manner is denoted by $\mu(\bar{x})$ in this paper.

In the seminal work of Boland et al. first introducing the DDD method, the authors prove the following lemma for problems with flow conservation constraints (Theorem 2 in \cite{Boland1}). 
\begin{lemma}\label{lemma:flow}
If $\bar{x}$ is a vector that satisfies the flow and integrality constraints  in \chs UPR$(D_T)$  (constraints (\ref{const:flow}) and (\ref{const:integrality}))\che, then $\hat{x} = \mu (\bar{x})$ satisfies the analogous constraints in \chs UPR($D_S$) ((\ref{const:D_S_flow}) and (\ref{const:D_S_integrality}))\che. 
\end{lemma}

\subsection{Arc capacities}\label{sec:arc_caps}
Techniques to incorporate arc capacities into the DDD framework were presented in $\cite{LagosDDD}$ and $\cite{Pottel}$. We present our work in full in this section for clarity, and for use in the novel work in Section \ref{sec:storage}. 

In the example provided in Figure \ref{mu_arc_map}, we assume there are unit arc capacities. Observe that $\mu$ will map each of the two $u \rightarrow v \rightarrow w$ trajectories (represented with dashed blue lines) to the same trajectory in $D_S$, which exceeds the original unit capacities. Thus, it is necessary to add to properties (P1)-(P4) in order to provide a lower bound for problems with arc capacities.

\begin{figure}[h]
    \centering
    \includegraphics[width=.5\textwidth]{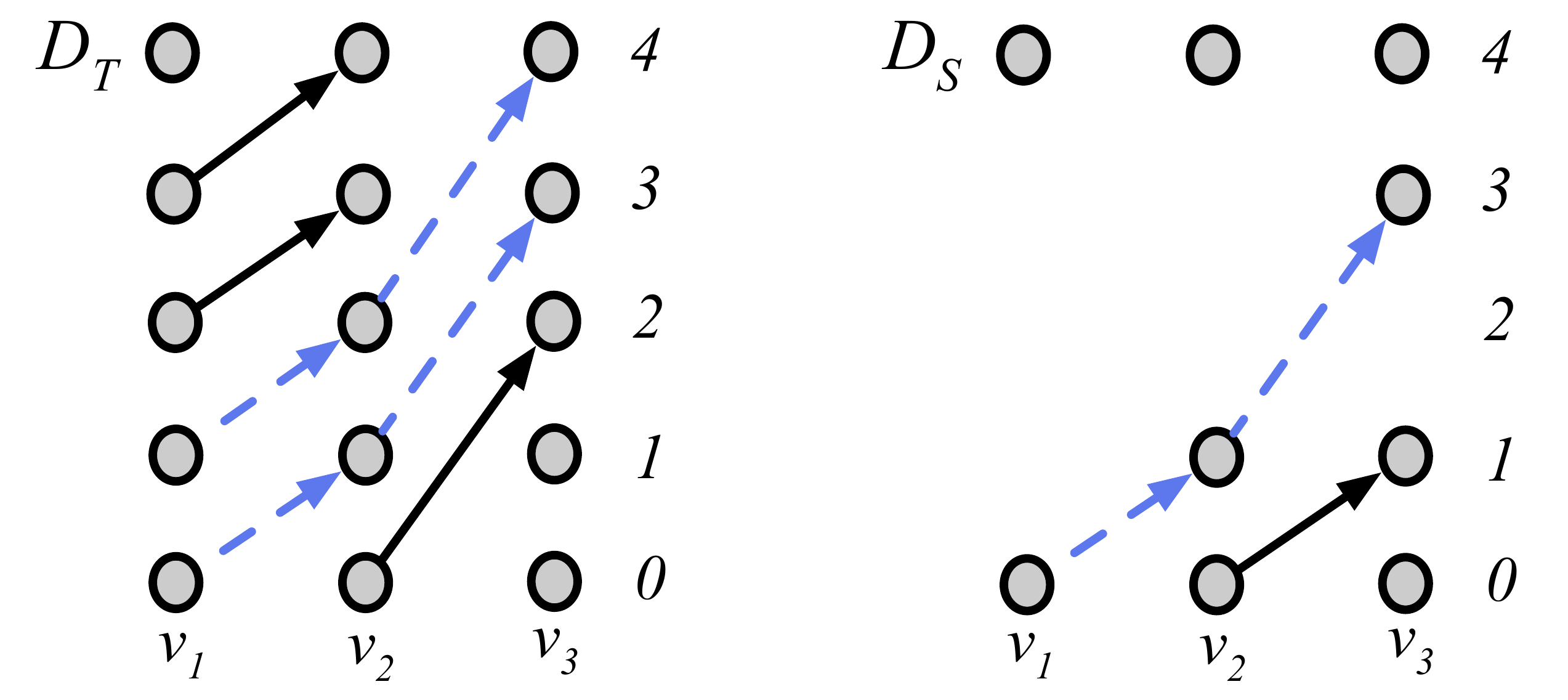}
    \caption{Two trajectories in $D_T$ are mapped to the same trajectory in $D_S$ via $\mu$.}
    \label{mu_arc_map}
\end{figure}

For each timed node $(v,t) \in N_S$, let $\mathtt{n_S}(v,t)$ be the time of the first appearance of $v$ after $t$ in $N_S$. That is,
\[\mathtt{n_S}(v,t) = \min \{t': t' > t, (v, t') \in N_S\}.\]

Let $e = ((v,t), (w, t')) \in A_S$. Then $\mu$ maps $f = ((v,t_1),(w,t_2)) \in A_T$ to $e$ when $t_1 \in \{t, t+1, \ldots, \mathtt{n_S}(v,t) - 1\}$. To capture the length of this interval, we define
\[\mathtt{m_S}(v,t) := \mathtt{n_S}(v,t) - t,\] which is the number of time units until the next appearance of $v$ in $N_S$. Observe that for any $(v,t) \in N_T$, $\mathtt{n_T}(v,t) = t+1$ and $\mathtt{m_T}(v,t) = 1$. Note, we will later use the fact that $\mathtt{m_S}(v,t)$ is well-defined even if $(v,t) \notin N_S$. For any \chs timed arc \che $((v, t), (w, t')) \in A_S$ there are $\mathtt{m_S}(v,t)$  \chs timed arcs \che in $A_T$ that are mapped to $((v, t), (w, t'))$ according to the map $\mu$. This proves that the following property, ($P^{\mathtt{arcs}}$), is sufficient in order for $\hat{x}$ to satisfy the arc capacity constraints of \chs UPR($D_S$)\che. 
\[ \mbox{($P^{\mathtt{arcs}}$) For any arc $\chs e \che = ((v,t), (w, t')) \in A_S$, $u'_{\chs e\che} = u_{\chs e \che} \cdot \mathtt{m_S}(v,t)$ }\]
Specifically, we have proven the following Lemma. \chs We include a brief formal proof for completeness.\che

\begin{lemma}\label{lemma:arcs}
Let $D_S$ be a partial network that satisfies properties $(P1)-(P4)$ and $(P^{\mathtt{arcs}})$, and let $\bar{x}$ be a solution to \chs UPR$(D_T)$\che. Then $\hat{x} = \mu(\bar{x})$ satisfies constraint \chs (11)\che. 
\end{lemma}

\chs
\begin{proof}
Let $e = ((v,t), (w, t')) \in A_S$, and consider $\sum_{k \in \mathcal{K}} \hat{x}_e^k$. If $\hat{x}_e^k = 1$ for some commodity $k \in \mathcal{K}$, then $\bar{x}_a^k = 1$ for a timed arc $a = ((v,\bar{t}), (w, \bar{t} + \tau_{vw})) \in A_T$ with $\mu(a) = e$. By definition of the map $\mu$, we know that $\bar{t} \in \{t, t + 1, \cdots, \mathtt{n_S}(v,t) - 1\}$. Thus, the set $\{a \in A_T: \mu(a) = e\}$ has size $\mathtt{m_S}(v,t)$. Therefore, 
\[ \sum_{k \in \mathcal{K}} \hat{x}_e^k \leq 
    \sum_{\substack{a \in A_T: \\ \mu(a) = e}} ~ \sum_{k \in \mathcal{K}} \bar{x}_a^k
\leq u_a \cdot \mathtt{m_S}(v,t),\]
where the final inequality holds since $\bar{x}$ was a feasible solution for UPR$(D_T)$. The result follows since $u_e = u_a$.
\end{proof}
\che

\subsection{Storage limits}\label{sec:storage}

First, we describe the difference between the storage constraints in universal packet routing \chs and \che the storage constraints dealt with in the work of Lagos et al. \cite{LagosDDD} when solving the continuous time inventory routing problem (CIR). In the CIR problem, a company manages the inventory of its clients, and delivers product from a single facility. Lagos et al. \chs add the restriction that only a single vehicle can be at a given client location at any point in time. In essence, this is a hard storage capacity at the parking lot for the client. Lagos et al. deal solely with graphs that are stars (``out-and-back'' routes), which simplifies the solution space so that each client node has only a single incoming arc in the flat network. In order to obtain a lower bound on the CIR instance given a partially time-expanded network, they ``relax'' a subset of the storage constraints. For each relaxed storage constraint at some timed node $(v,t)$, the authors remove the storage constraint (vehicle limit), effectively allowing unlimited storage at $(v,t)$ in that iteration.\che 

\chs However, in our work we show that in the discrete setting we can be more conservative in the relaxation of storage capacities. In Appendix \ref{app:improved_bound}, we show that the approach of Lagos et al. applied to a discrete problem would result in $\Omega(T)$ iterations, whereas our tighter relaxation would allow DDD to terminate in a single iteration. In addition, our work deals with more complex flat networks that allow arbitrary routes rather than just out-and-back routes. This is of non-trivial importance, since it is not clear that the approach of Lagos et al. would perform well outside of the out-and-back framework.\che

Similar to the case of arc capacities, we cannot simply assign the node capacities from the base graph to the nodes in $N_S$. In this section, we will show how to assign holdover arc storage $b'_e$ to each timed arc $e \in H_S$ to ensure that UPR($D_S$) is a relaxation of UPR($D_T$).

Recall the definition of $\mathtt{m_S}(v,t)$, which is the number of time units until the first appearance of $v$ after time $t$ in $N_S$. Note, $\mathtt{m_S}(v,t)$ is well-defined even if $(v,t) \notin N_S$. For the following discussion, we look at the neighbours of a timed node in $D_S$. For $(v,t) \in N_S$, let $N_S^-(v,t)$ be the incoming neighbours of $(v,t)$ in $D_S$. Specifically, 
$$ N_S^-(v,t) = \{(w,t'): \exists ~ e = ((w,t'), (v,t)) \in A_S\}.
$$ 
Let $\bar{x}$ be a feasible solution to UPR($D_T$), and let $\hat{x}$ be the vector we obtain via the map $\mu: A_T \rightarrow A_S$ as stated at the beginning of this section. We would like to understand how the map $\mu$ could impact the storage required at some \chs timed node \che $(v, t) \in N_S$. 

For all \chs $k \in \mathcal{K}$ \che, let \chs $Q_k$ \che be the trajectory in $D_T$ that packet \chs $k$ \che travels along according to $\bar{x}$. \chs $Q_k$ consists of an ordered set of movement timed arcs, $\{e_1^{\chs k\che}, e_2^{\chs k\che}, \ldots, e^{\chs k\che}_{l_{\chs k\che}}\}$, along with additional holdover arcs. Let $P_k$ be the corresponding path in the underlying graph $D$\che. We define $\bar{t}_u^{\chs k \che, out}$ and $\bar{t}_v^{\chs k \che, in}$ so that $((u, \bar{t}^{k, out}_u), (v, \bar{t}_v^{k, in})) = e^k_j$, and we define $\hat{t}_u^{\chs k \che, out}$ and $\hat{t}_v^{\chs k \che, in}$ analogously. That is, 
$$ 
    e = ((u, \bar{t}_u^{\chs k \che, out}), (v, \bar{t}_v^{\chs k \che, in})) \quad \rightarrow \quad \mu(e) = ((u, \hat{t}_u^{\chs k \che, out}), (v, \hat{t}_v^{\chs k \che, in}))
$$
Consider two consecutive movement timed arcs in $D_T$, $e_{uv} = ((u,\bar{t}_u^{\chs k \che, out}), (v, \bar{t}_v^{\chs k \che, in}))$ and $e_{vw} = ((v,\bar{t}_v^{k, out}),(w,\bar{t}_w^{k, in}))$ along the trajectory \chs $Q_k$\che. The flow on $e_{uv}$ and $e_{vw}$ is mapped to $\mu(e_{uv}) = ((u, \hat{t}_u^{\chs k \che, out}), (v, \hat{t}_v^{\chs k \che, in}))$ and $\mu(e_{vw}) = ((v,\hat{t}_v^{\chs k \che, out}),(w,\hat{t}_w^{k, in}))$ respectively.

Let $(v,t)$ be a timed copy of $v$ in $N_S$. The following straightforward facts will be used to  understand how the map $\mu$ impacts the storage of packet $\chs k \che$ at $(v,t)$. 
\begin{enumerate}[noitemsep]
    \item[(F1)] $\hat{t}_l^{\chs k \che, out} = \max\{t : t \leq \bar{t}_l^{\chs k \che, out}, (l, t) \in N_S\}$ for $l \in \{u,v\}$;
    \item[(F2)] $\hat{t}_v^{\chs k \che, in} = \max \{t: t \leq \hat{t}_u^{\chs k \che, out} + \tau_{uv}, (v, t) \in N_S \}$;
    \item[(F3)] $\hat{t}_v^{\chs k \che, in} \leq \bar{t}_v^{\chs k \che, in}$ and $\hat{t}_v^{\chs k \che, out} \leq \bar{t}_v^{\chs k \che, out}$.
\end{enumerate}
(F1) follows by definition of $\mu$ (equation (\ref{mu_eq})), and (F2) follows from (F1) along with the fact that there is a $uv$ timed arc departing $(u, \hat{t}_u^{k, out})$ in $A_S$ that is as long as possible (properties (P2) and (P4)). (F3) follows from (F1) and (F2), along with the fact that $\bar{t}_v^{k, in} = \bar{t}_u^{k, out} + \tau_{uv}$.

If packet $\chs k\che$ was previously stored at $(v,t)$ according to $\bar{x}$ ($\bar{t}_v^{\chs k \che,in} \leq t < \bar{t}_v^{\chs k \che,out}$), then $\mu$ cannot introduce \emph{additional} storage of packet $k$. So suppose packet $k$ is not stored at $(v,t)$ in $\bar{x}$. If  $\bar{t}_v^{\chs k \che,in} \leq t$, then since the packet is not stored at $(v,t)$, we also have that $\bar{t}_v^{\chs k \che, out} \leq t$. In this case, by fact (F3) it follows that $\hat{t}_v^{\chs k \che, out} \leq \bar{t}_v^{\chs k \che, out} \leq t$, and so packet $\chs k \che$ is not stored at $(v,t)$ in $\hat{x}$. However, the storage of packet $\chs k\che$ could increase at $(v, t)$ if $\bar{t}_v^{\chs k \che, in} > t$, and packet $\chs k \che$ arrives at $v$ earlier according to $\hat{x}$ than it is scheduled to arrive according to $\bar{x}$. That is, $\hat{t}_v^{\chs k \che, in} < \bar{t}_v^{\chs k \che, in}$. This can happen if either:

\begin{enumerate}[noitemsep]
    \item Flow departs $u$ at the \chs same \che time ($\hat{t}_u^{\chs k \che, out} = \bar{t}_u^{\chs k \che, out}$), but $(v, \bar{t}_v^{\chs k \che, in}) \notin N_S$. For example, in Figure \ref{storage_D_S1} \chs shows a partially time-expanded network where \che $(u,1) \in N_S$, but $(v,2)\notin N_S$, so $\mu((u,1),(v,2)) = ((u,1), (v,1))$;
    \item Flow is forced to depart early from the preceding node ($\hat{t}_u^{\chs k \che, out} < \bar{t}_u^{\chs k \che, out}$). For example, in Figure \ref{storage_D_S2} shows a \chs partially time-expanded network where \che $(u,1) \notin N_S$, so flow must depart $u$ early and $\mu((u,1),(v,2)) = ((u,0), (v,1))$.
\end{enumerate}

\begin{figure}[!htb]
    \centering
    \begin{minipage}{.33\textwidth}
        \centering
    \includegraphics[width=0.6\textwidth]{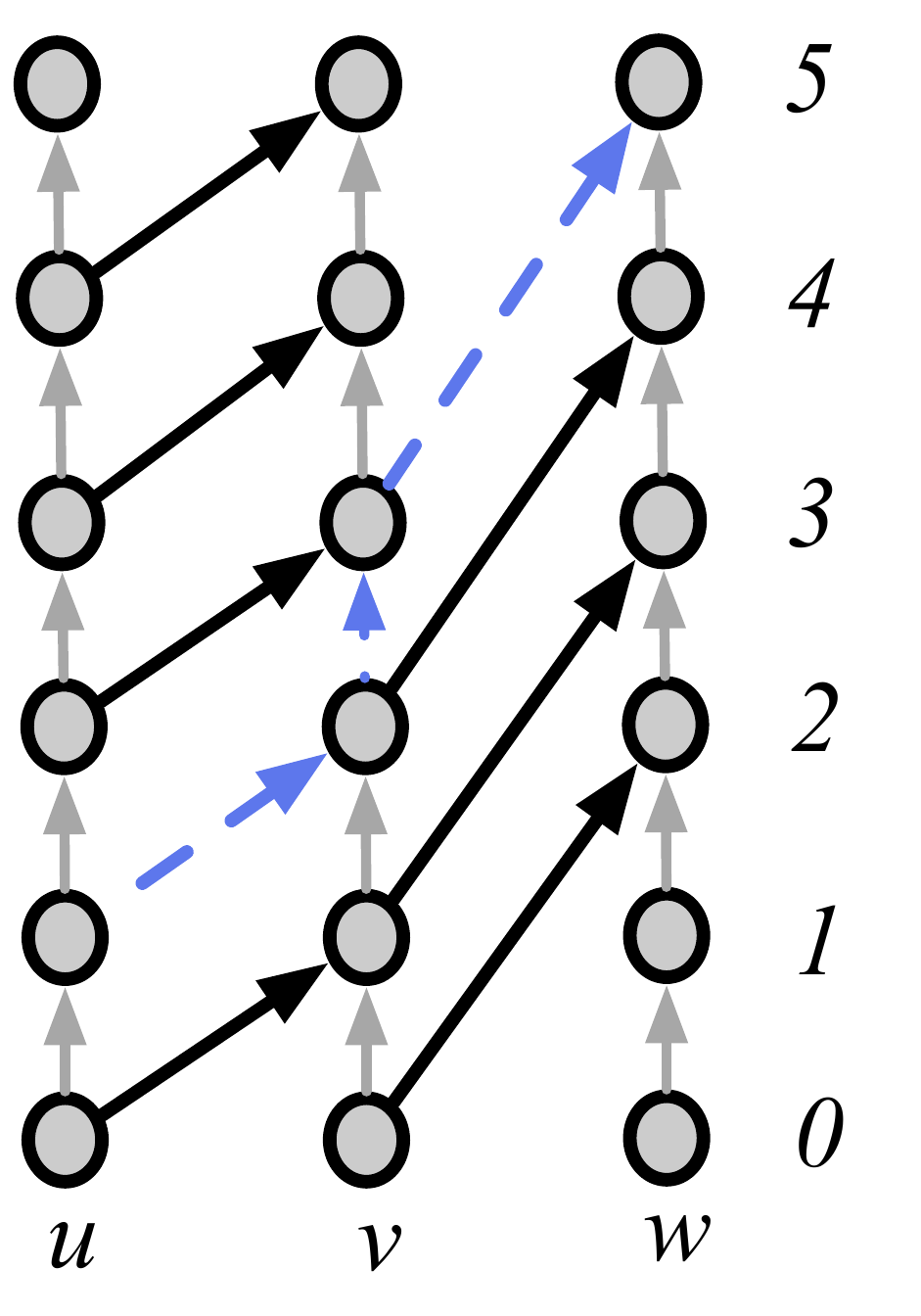}
    \caption{Original trajectory}
    \label{storage_D_T}
    \end{minipage}%
    \begin{minipage}{.33\textwidth}
    \centering
    \includegraphics[width=0.6\textwidth]{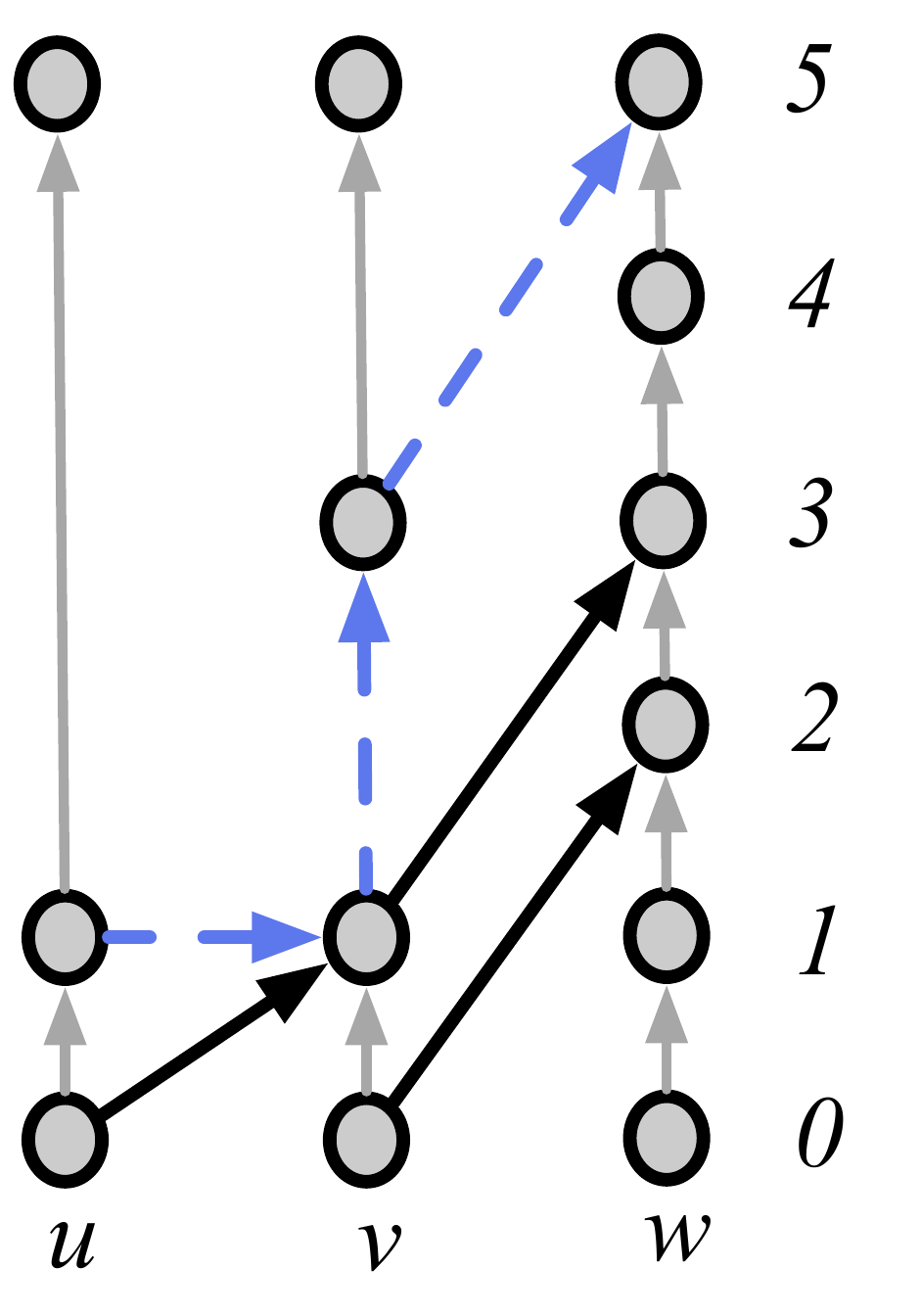}
    \caption{Scenario 1}
    \label{storage_D_S1}
    \end{minipage}%
    \begin{minipage}{0.33\textwidth}
        \centering
        \includegraphics[width=0.6\textwidth]{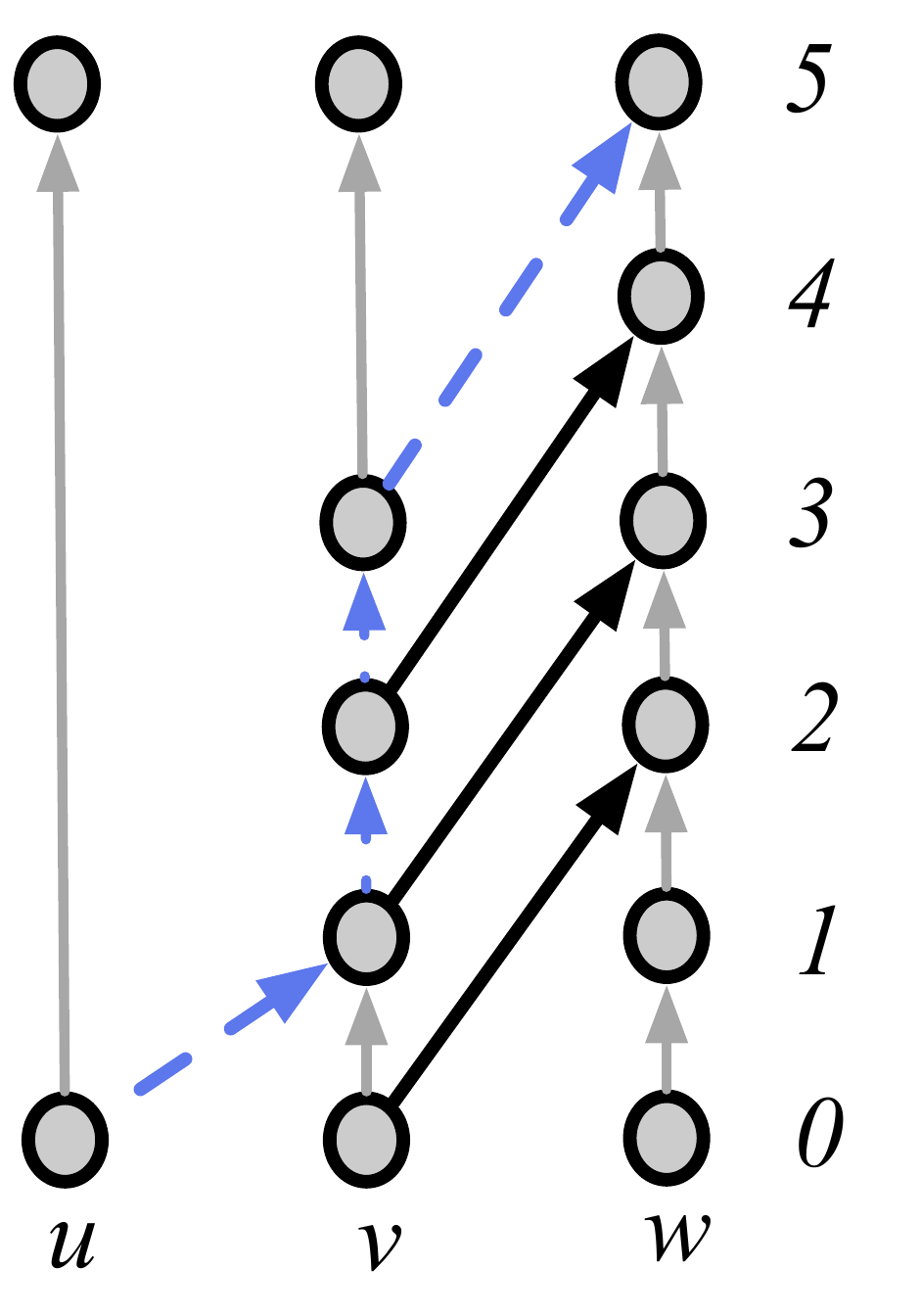}
    \caption{Scenario 2}
    \label{storage_D_S2}
    \end{minipage}
\end{figure}
\FloatBarrier

Let $\mathcal{K}(v)$ denote the set of packets that travel along a trajectory that includes a timed copy of $v$ according to $\bar{x}$. Let $\mathcal{K}^1_{\hat{x}}(v)$ denote the set of packets in $\mathcal{K}(v)$ that depart the node preceding $v$ at the same time in $\hat{x}$ and $\bar{x}$. Similarly, let $\mathcal{K}^2_{\hat{x}}(v)$ denote the set of packets in $\mathcal{K}(v)$ that depart the preceding node earlier in $\hat{x}$ than in $\bar{x}$. Note, $\mathcal{K}(v) = \mathcal{K}^1_{\hat{x}}(v) \cup \mathcal{K}^2_{\hat{x}}(v)$. Let $e = ((v,t), (v,t'))$ be the timed arc in $H_S$ departing $(v,t)$, and let $f = ((v,t), (v,t+1))$ be the timed arc in $H_T$ departing $(v,t)$. 

\chs
\begin{lemma}\label{lemma:x1-a}
If $D_S$ satisfies properties $(P1) - (P4)$, then any packet $ k \in \mathcal{K}^1_{\hat{x}}(v)$ that was not stored at $v$ in $\bar{x}$ is not stored at $v$ in $\hat{x} = \mu(\bar{x})$.
\end{lemma}

\begin{proof}
Suppose packet $k \in \mathcal{K}^1_{\hat{x}}(v)$ is not stored at $v$ in $\bar{x}$. That is, $\bar{t}^{k, in}_v = \bar{t}^{k, out}_v$. Let $u$ be the preceding node along the path $P_{k}$ in $\bar{x}$. Since $k \in \mathcal{K}^1_{\hat{x}}(v)$, we know that $\hat{t}^{k, out}_u = \bar{t}^{k, out}_u$. Thus, in $D_S$, packet $k$ travels along some arc $((u, \bar{t}^{k, out}_u),(v, \hat{t}^{k, in}_v))$. By fact (F2), 
$~\hat{t}^{k, in}_v = \max \{t: t \leq \hat{t}^{k, out}_u + \tau_{uv}, (v, t) \in N_S \}.$
Similarly by fact (F1), 
$~\hat{t}^{k, out}_v = \max\{t : t \leq \bar{t}^{k, out}_v, (v, t) \in N_S\}.$
Since $\bar{t}^{k, out}_v = \bar{t}^{k, out}_u + \tau_{uv}$ and $\hat{t}^{k, out}_u = \bar{t}^{k, out}_u$ we see that $\hat{t}^{k, in}_v = \hat{t}^{k, out}_v$. Thus, no additional storage of packet $k$ was introduced at $v$ and $\hat{x}_e^{k} = \bar{x}_f^{k} = 0$.
\end{proof}
\che

\chs We first establish a simple bound on the additional storage needed at $(v,t)$ to accommodate packets in $\mathcal{K}^1_{\hat{x}}(v)$ in Lemma \ref{lemma:x1-b}. We then tighten this argument in Lemma \ref{lemma:x1-c}. In Appendix \ref{app:improved_bound} we show that this tightening can prove essential to the effectiveness of DDD for certain problem instances.\che

\chs
\begin{lemma}\label{lemma:x1-b}
If $D_S$ satisfies properties $(P1) - (P4)$, then 
$$\sum_{k \in \mathcal{K}^1_{\hat{x}}(v)} \hat{x}_e^k \leq \sum_{k \in \mathcal{K}^1_{\hat{x}}(v)} \bar{x}_f^k +  (\mathtt{m_S}(v,t)-1) b_v.$$
\end{lemma}
\che
\begin{proof}
Suppose packet $\chs k \che$ was stored at $v$ in $\bar{x}$ ($\bar{t}_v^{\chs k \che, in} < \bar{t}_v^{\chs k \che, out}$). We know that as in Figure \ref{storage_D_S2},  storage could be introduced at $(v,t)$ if $\hat{t}_v^{\chs k \che, in} \leq t < \bar{t}_v^{\chs k \che, in}$. Since \chs packet $k$ \che departs the preceding node at the same time in $\bar{x}$ and $\hat{x}$, it follows that $\bar{t}_v^{\chs k \che, in} \in \{t+1, t+2, \cdots, \mathtt{n}_S(v,t)-1\}$. \chs Thus each packet in $\mathcal{K}^1_{\hat{x}}(v)$ that introduces storage at $(v,t)$ must have previously been stored at $v$ at one of the times in this interval, which has length $\mathtt{m_S}(v,t) - 1$. Therefore, $\sum_{k \in \mathcal{K}^1_{\hat{x}}(v)} \hat{x}_e^k \leq \sum_{k \in \mathcal{K}^1_{\hat{x}}(v)} \bar{x}_f^k +  (\mathtt{m_S}(v,t)-1) b_v$. \che
\end{proof}

\chs
In the following Lemma we observe that this bound can be significantly tightened. This tightening relies on the argument that if two packets $j$ and $k$ in $\mathcal{K}^1_{\hat{x}}(v)$ where not stored at $v$ at the same time in $\bar{x}$, then the same is true in $\hat{x}$. 
\che

\chs
\begin{lemma}\label{lemma:x1-c}
If $D_S$ satisfies properties $(P1) - (P4)$, then 
$$\sum_{k \in \mathcal{K}^1_{\hat{x}}(v)} \hat{x}_e^k \leq b_v.$$ 
If in addition, $\mathtt{n}_S(v,t) = t+1$, then 
$$\sum_{k \in \mathcal{K}^1_{\hat{x}}(v)} \hat{x}_e^k \leq \sum_{k \in \mathcal{K}^1_{\hat{x}}(v)} \bar{x}_f^k.$$
\end{lemma}

\noindent \emph{Proof.} 
Suppose packet $k$ was stored at $v$ in $\bar{x}$ ($\bar{t}^{k, in}_v < \bar{t}^{k, out}_v$). Again,  storage could be introduced if $\hat{t}^{k, in}_v < \bar{t}^{k, in}_v$. However, we will show that if commodities $k$ and $j$ in $\mathcal{K}^1_{\hat{x}}(v)$ were not stored at the same time at $v$ according to $\bar{x}$, then the same holds for $\hat{x}$. This would prove that $\sum_{k \in \mathcal{K}^1_{\hat{x}}(v)} \hat{x}_e^k \leq b_v$ since $\bar{x}$ was feasible. 
\che

\begin{wrapfigure}{r}{0.25\textwidth}
\centering
\vspace{-1cm}
\includegraphics[width=0.16\textwidth]{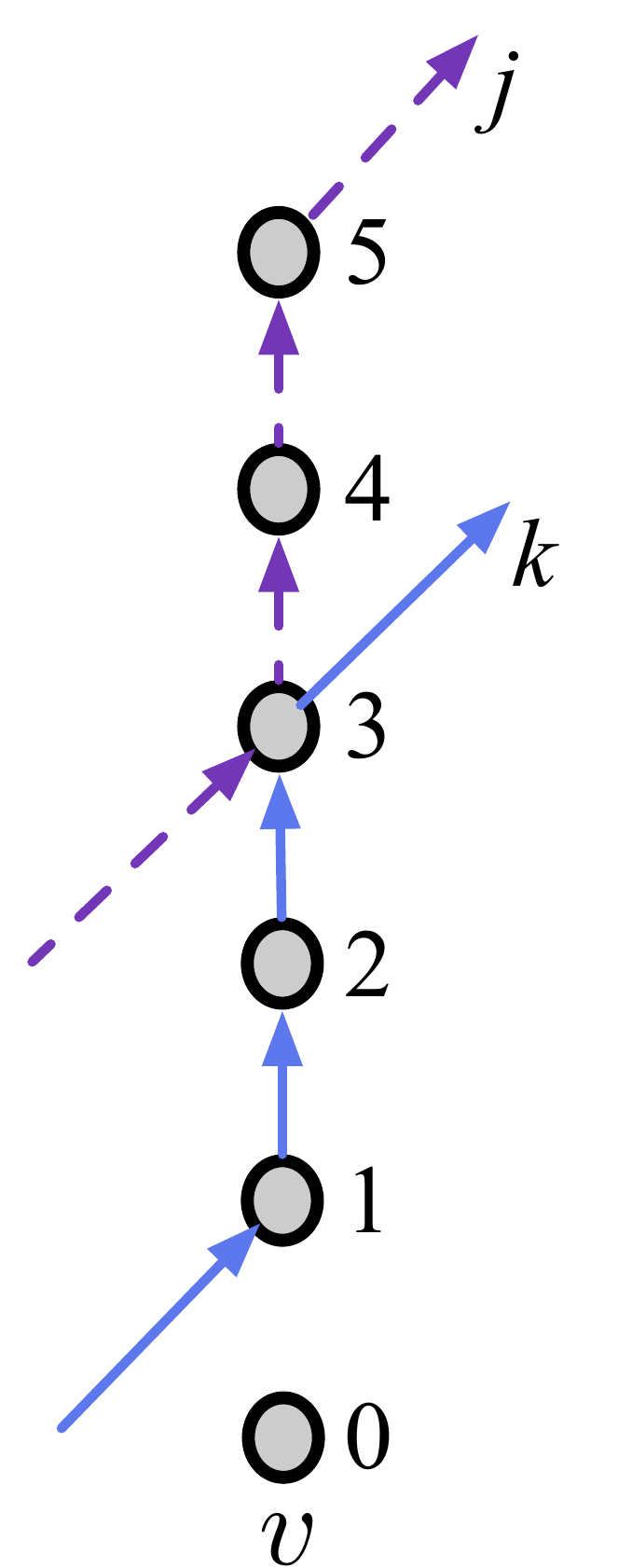}
\caption{Disjoint storage}
\label{fig:disjoint_storage}
\vspace{-.5cm}
\end{wrapfigure} 

Without loss of generality, we may assume $\bar{t}_v^{\chs k\che, out} \leq \bar{t}_v^{j, in}$, as in Figure \ref{fig:disjoint_storage} ($\bar{t}_v^{\chs k\che, out} = 3$ and $\bar{t}_v^{j, in} = 3$), since one packet must have departed node $v$ no later than the arrival time of the other. We claim that $\hat{t}_v^{\chs k\che, out} \leq \hat{t}_v^{j, in}$.

Suppose $\hat{t}_v^{j, in} \geq \bar{t}_v^{\chs k\che, out}$. Then since $ \hat{t}_v^{\chs k\che, out} \leq \bar{t}^{k, out}_v$ by (F3), it follows that $\hat{t}_v^{\chs k\che, out} \leq \hat{t}_v^{j, in}$. Alternatively, suppose $\hat{t}_v^{j, in} < \bar{t}_v^{\chs k\che, out}$. By (F1),
$$\hat{t}_v^{\chs k\che, out} = \max\{t : t\leq \bar{t}_v^{\chs k\che, out}, (v, t) \in N_S\}.$$
\chs Let $u$ be the node that packet $j$ visits before $v$ according to $\bar{x}$\che. By (F2) and since $\hat{t}^{j, out}_u = \bar{t}_u^{j, out}$ (S1), 
$$\hat{t}^{j, in}_v = \max \{t: t \leq \hat{t}_v^{j, out} + \tau_{uv} = \bar{t}_v^{j, in}, (v, t) \in N_S \}.$$
Since $\hat{t}_v^{j, in} < \bar{t}_v^{\chs k\che, out}$ and $\bar{t}_v^{\chs k\che, out} \leq \bar{t}_v^{j, in}$, it follows that $\hat{t}_v^{\chs k\che, out} = \hat{t}_v^{j, in}$ $= \max\{t : t\leq \bar{t}_v^{\chs k\che, out}, (v, t) \in N_S\}$ as required. This proves that $\hat{x}^1$ satisfies storage constraints at all $(v,t) \in N_S$.

Finally, when $(v, t+1) \in N_S$, then the set of packets in $\mathcal{K}^1_{\hat{x}}(v)$ stored at $(v,t)$ in $\hat{x}$ are precisely those that are stored at $(v,t)$ in $\bar{x}$. Thus, if $\mathtt{n}_S(v,t) = t+1$, then $\sum_{k \in \mathcal{K}^1_{\hat{x}}(v)} \hat{x}_e^k \leq \sum_{k \in \mathcal{K}^1_{\hat{x}}(v)} \bar{x}_f^k$. 

\vspace{-.3cm}
\hspace{15.5cm} \qedsymbol
\vspace{.3cm}

\chs
We now bound the storage of packets in $\mathcal{K}^2_{\hat{x}}(v)$. In order to be stored at $v$ at time $t$, a packet must arrive at $v$ by time $t$. This gives an upper bound on the time the packet could have departed the previous node in $D_S$, which implies a corresponding upper bound on the time the packet could have departed the previous node in $\bar{x}$. Similarly, we identify a lower bound on the time a packet in $\mathcal{K}^2_{\hat{x}}(v)$ could depart the previous node in $D_T$ if it was not stored at $(v,t)$ in $\bar{x}$, but is stored at $(v,t)$ in $\hat{x}$. 

For ease of notation, we introduce an additional definition. Given partially and fully time-expanded networks $D_S$ and $D_T$ respectively, for each 
$e = ((v,t), (v,t')) \in H_S$ we define 
$$ U_e(D_S, D_T) := \sum\limits_{(w,t') \in N^-_T(v,t) \cup N^-_S(v,t)} u_{wv} \cdot (\mathtt{m_S}(w,t') - 1). $$
We write $U_e$ when $D_S$ and $D_T$ are self-evident.
\che

\begin{lemma}\label{lemma:x2}
If $D_S$ satisfies properties $(P1) - (P4)$, then 
$$\sum_{k \in \mathcal{K}^2_{\hat{x}}(v)} \hat{x}_e^k \leq \sum_{k \in \mathcal{K}^2_{\hat{x}}(v)} \bar{x}_f^k \chs + U_e. \che$$
\end{lemma}

\begin{wrapfigure}{r}{0.3\textwidth}
\centering
\vspace{-.25cm}
\includegraphics[width=0.3\textwidth]{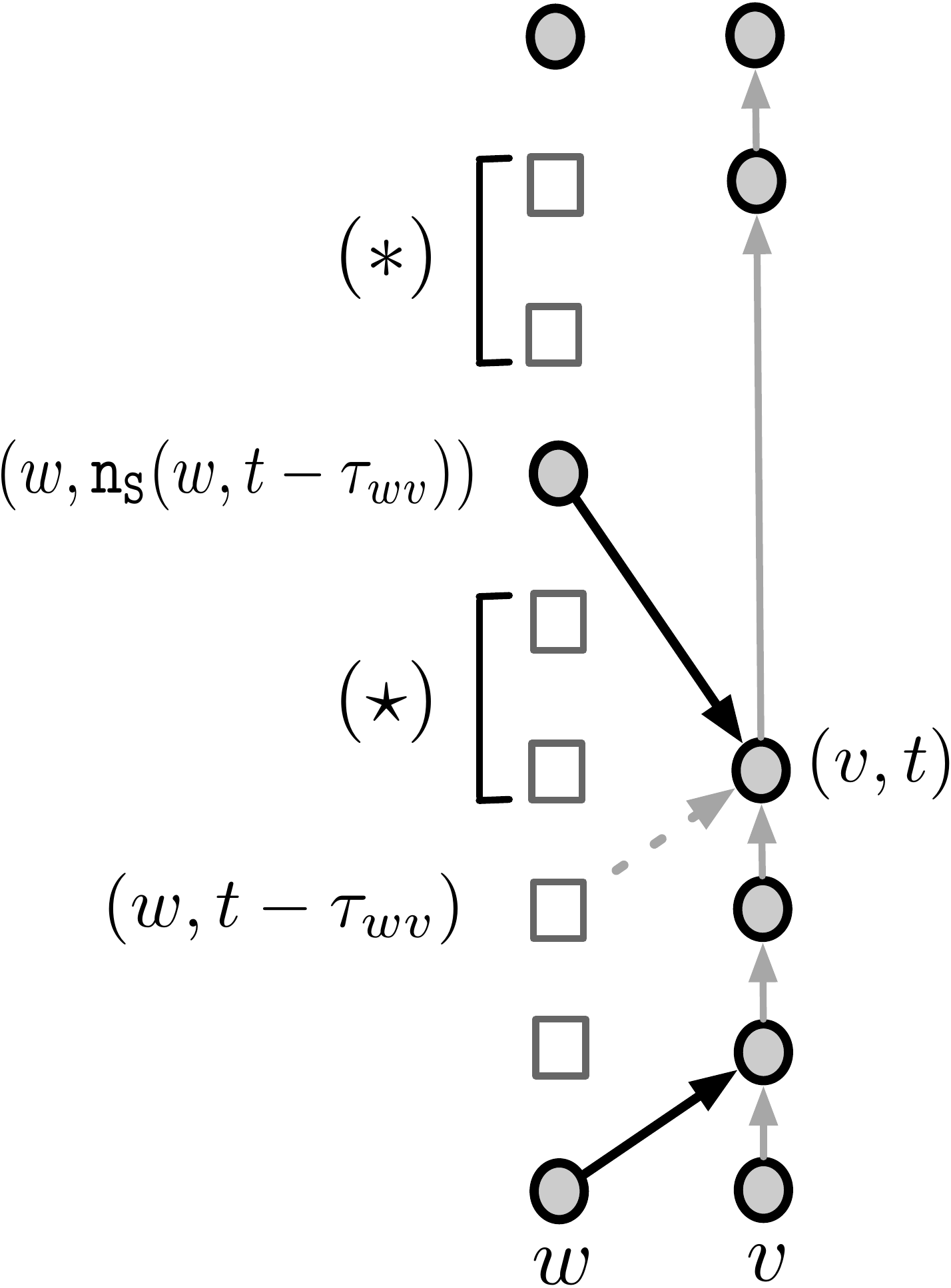}
\caption{}
\vspace{-.5cm}
\label{fig:group2_bound}
\end{wrapfigure} 

\noindent \emph{Proof.} 
Let $w \in N^{-}(v)$. We know that multiple timed arcs of $(w,v)$ in $A_T$ are mapped to the same timed arc $((w,t_1), (v,t_2))$ in $A_S$ as discussed in Section \ref{sec:arc_caps} which could introduce additional storage at $v$. \chs Throughout this proof we will consider Figure \ref{fig:group2_bound}. The missing nodes are marked with white squares, and the dashed gray arc shows the arc we would obtain if $(w,t- \tau_{wv})$ was in the current partially time-expanded network (which it may or may not be).\che

Suppose $\chs k\che$ is a packet in $\mathcal{K}^2_{\hat{x}}(v)$ where $w$ is the node it visits immediately before $v$ according to the solution $\bar{x}$. That is, packet $\chs k\che$ travels along a timed arc $((w, \bar{t}_w^{k,out}), (v, \bar{t}_v^{k,in}))$ in $D_T$. First, consider the case where $\bar{t}_v^{k,in} \leq t$. If packet $\chs k\che$ was originally stored at $(v,t)$ in $\bar{x}$, then the storage needed at $(v,t)$ for packet $\chs k\che$ in $\hat{x}$ cannot exceed the level in $\bar{x}$. If instead packet $k$ was not stored at $v$, then $\bar{t}_v^{k,out} \leq t$ and so $\hat{t}_v^{k,out} \leq t$ as well. As a result, packet $\chs k\che$ would not stored at $(v,t)$ in $\hat{x}$. 

Thus, it is only possible to introduce storage at $(v,t)$ in $\hat{x}$ for packet $\chs k\che$ when $\bar{t}_v^{\chs k\che,in} > t$ and $\hat{t}_v^{\chs k\che,in} \leq t$.
The first condition is equivalent to $\bar{t}_w^{\chs k\che,out} > t - \tau_{wv}$. 
\chs
We now consider the second condition, $\hat{t}_v^{k,in} \leq t$. First observe that if $\bar{t}_w^{k,out} < \mathtt{n_S}(w, t - \tau_{wv})$, then $\hat{t}_v^{k,in} \leq t$. Thus, it could be the case that  $$\bar{t}^{k,out}_w \in \{t - \tau_{wv} + 1, t - \tau_{wv} + 2, \cdots, \mathtt{n_S}(w, t - \tau_{wv}) - 1\}.$$

This interval has length $\mathtt{m_S}(w,t - \tau_{wv}) - 1$, and we mark these departure times with $(\star)$ in Figure \ref{fig:group2_bound}. Note that $(w,t-\tau_{wv}) \in N_T^-(v,t)$.

Alternatively we could have $\bar{t}_w^{k,out} \geq \mathtt{n_S}(w, t - \tau_{wv})$, and if $\hat{t}_v^{k,in} \leq t$, property (P4) implies that $\hat{t}_v^{k,in} = t$. That is, $\mu(((w, \bar{t}_w^{k,out}), (v, \bar{t}_v^{k,in})))$ is an incoming timed arc at $(v,t)$ in $D_S$. Thus, $(w, \hat{t}_w^{k,out}) \in N_S^-(v,t)$. This set is marked with $(\ast)$ in Figure \ref{fig:group2_bound}. 
\che
Thus, the additional storage needed at $(v,t)$ to accommodate packets in $\mathcal{K}^2_{\hat{x}}(v)$ is at most
\chs
$$\sum_{(w,t') \in N^-_T(v,t) \cup N^-_S(v,t)} u_{wv} \cdot (\mathtt{m_S}(w,t') - 1) = U_e.$$
\che

\vspace{-.5cm}
\hspace{15.5cm} \qedsymbol

\noindent Thus, we introduce the following property. The proof of Lemma \ref{lemma:storage} follows from \chs Lemmas \ref{lemma:x1-a}, \ref{lemma:x1-c}, and \ref{lemma:x2} \che along with the fact that for each $e \in H_S$, $\sum_{k \in \mathcal{K}_v} \hat{x}_e^k = \sum_{k \in \mathcal{K}^1_{\hat{x}}(v)} \hat{x}_e^k + \sum_{k \in \mathcal{K}^2_{\hat{x}}(v)} \hat{x}_e^k$. 

\chs
\noindent ($P^{\mathtt{storage}}$): For any $e = ((v,t), (v,t')) \in H_S$, 
$$b'_e \geq \begin{cases}
        b_v + U_e & \quad \mbox {if} ~(v,t+1) \in N_S\\
        2b_v + U_e & \quad \mbox {if} ~(v,t+1) \notin N_S \\
        \end{cases}$$
\che
\begin{lemma}\label{lemma:storage}
Let $D_S$ be a partially time-expanded network that satisfies properties $(P1)-(P4)$ and $(P^{\mathtt{storage}})$, and let $\bar{x}$ be a solution to \chs UPR($D_T$)\che. Then $\hat{x} = \mu(\bar{x})$ satisfies constraint \chs (12)\che.
\end{lemma}
\noindent Finally, we prove the following theorem that UPR($D_S$) is indeed a lower bound. 

\begin{theorem}
If $D_S$ satisfies properties $(P1) - (P4)$, $(P^{\mathtt{arcs}})$, and $(P^{\mathtt{storage}})$,
then the objective value of an optimal solution to UPR($D_S$) is at most the objective value of an optimal solution to \chs UPR($D_T$)\che.
\end{theorem}
\begin{proof}
Let $\bar{x}$ be a solution to \chs UPR($D_T$)\che, and let $\hat{x} = \mu(\bar{x})$. By Lemma \ref{lemma:flow}, $\hat{x}$ satisfies the flow and integrality constraints of UPR($D_S$). By Lemmas \ref{lemma:arcs} and \ref{lemma:storage}, $\hat{x}$ satisfies the arc capacity and storage capacity constraints of UPR($D_S$). Finally, since $\mu$ maps trajectories in $D_T$ to trajectories in $D_S$ that underestimate the original length, the objective value of $\hat{x}$ in UPR($D_S$) is at most the objective value of $\bar{x}$. 
\end{proof}

It is important to note that when $N_S = N_T$, UPR($D_S$) is equivalent to \chs UPR$(D_T)$\che. Furthermore, the objective value of UPR($D_S$) is non-decreasing as we add timed nodes to $N_S$. However, we would ultimately like to solve \chs UPR($D_T$) \che without having to use $N_S = N_T$. In the following sections, we will describe how to detect when a solution to UPR($D_S$) can be converted to a solution of \chs UPR($D_T$) \che of equal makespan, and if not, how we select the timed nodes to add to $N_S$. 

The minimal set of \chs timed nodes \che satisfying $(P1)-(P4)$, $(P^{\mathtt{arcs}})$, and $(P^{\mathtt{storage}})$, is the set of all nodes at times 0 and $T$. We begin the algorithm with this set of timed nodes.

\begin{algorithm}[h]
\hspace*{\algorithmicindent} \textbf{Input}: Base network $D = (N,A)$, packet set $\mathcal{K}$, and upper bound, $T$, on the optimal makespan
\begin{algorithmic}[1]  
	\ForAll{$v \in N$}
	    \State Add \chs timed \che nodes $(v,0)$ and $(v,T)$ to $N_S$
    	\EndFor
    	\State \Return $N_S$
\caption{Generate-Initial-$N_S$($D = (N,A), \mathcal{K}, T$)}
\label{alg:initial_N_S}
\end{algorithmic}
\end{algorithm}	
\FloatBarrier

\chs In Algorithm \ref{alg:gen_A_S} we take as input the current set of timed nodes and generate the timed arcs $A_S$ and $H_S$ along with capacities $u'$ and $b'$ so that $D_S = (N_S, A_S \cup H_S), u', b'$ satisfies $(P1) - (P4)$, $(P^{\mathtt{arcs}})$, and $(P^{\mathtt{storage}})$. We define the capacities so that they satisfy $(P^{\mathtt{arcs}})$ and $(P^{\mathtt{storage}})$ with equality. The construction of the timed arcs is standard in the DDD literature.\che

\begin{algorithm}[h]
\hspace*{\algorithmicindent} \textbf{Input}: Base network $D = (N,A)$, and a set of timed nodes, $N_S$
\begin{algorithmic}[1]  
	\ForAll{$(v,t) \in N_S$}
		\State $e \leftarrow ((v,t), (v, \mathtt{n_S}(v,t)))$
        		\State $U_e(D_S,D_T) \leftarrow \sum_{(w,t') \in N^-_T(v,t)\cup  N^-_S(v,t)} u_{wv} \cdot (\mathtt{m_S}				(w,t') - 1)$
        		\State $b'_e := \begin{cases}
	     			b_v + U_e(D_S,D_T) \quad \mbox {if} ~(v,t+1) \in N_S \\
         			2 b_v + U_e(D_S,D_T)\quad \mbox {if} ~(v,t+1) \notin N_S \\
        		\end{cases}$
       		 \State If $t < T$, add timed arc $e$ to $H_S$ with storage capacity $b'_e$
        		\ForAll{$vw \in A$}
	        		\State Add timed arc $f = ((v,t), (w,t'))$ with capacity $u'_f = u_{vw} \cdot \mathtt{m_S}(v,t)$ to 					$A_S$ where $t'$ is the largest 
            		\Statex \hskip\algorithmicindent \hskip\algorithmicindent value such that $(w,t') \in N_S$ and $t' \leq t + 					\tau_{vw}$
       		 \EndFor
   	 \EndFor
    	 \State \Return $A_S, H_S$
\caption{Generate-$A_S\cup H_S$($N_S$, $D, T$)}
\label{alg:gen_A_S}
\end{algorithmic}
\end{algorithm}	
\FloatBarrier

\vspace{-0.5cm}
\section{Upper bound model and augmentation}\label{sec:UB_Aug}
\subsection{Upper bound model}
Given a solution to the lower bound model, we want to determine if it can be converted to an optimal solution of \chs UPR($D_T$)\che. Suppose we are working with a partially time-expanded network $D_S = (N_S, A_S \cup H_S)$ that satisfies (P1) - (P4), $(P^{\mathtt{arcs}})$, and $(P^{\mathtt{storage}})$. Let $\hat{x}$ be an optimal solution to UPR($D_S$) with value $\hat{T}$. We would like to know if $\hat{x}$ can be converted to a solution to \chs UPR($D_T$) \che with the same value (i.e. makespan). 

Observe that $\hat{x}$ specifies a trajectory in $D_S$ for each packet, each of which corresponds to a path in the underlying static graph $D$. Additionally, we are given a candidate makespan $\hat{T}$. Thus, we can generate an upper bound for \chs UPR($D_T$) \che  if we solve UPR in $D_T$ with the added the restriction that packets follow the underlying paths in $D$ specified by $\hat{x}$. Specifically, we have an instance of UPR-FP where for each $k \in \mathcal{K}$, $P_k$ is the path in $D$ induced by $\hat{x}$ for packet $k$. Let $A_T^k$, and $H_T^k$ denote the set of timed arcs that could be used for a trajectory with underlying path $P_k$. That is, for each $k \in \mathcal{K}$, 
\[ A_T^k = \{ ((v, t), (w, t')) \in A_T: vw \in A(P_k)\} \quad \quad \mbox{and} \quad \quad H_T^k = \{((v, t), (v, t')) \in \chs H_T\che : v \in N(P_k)\}.\]
While this instance of UPR-FP can still be solved more quickly than the original UPR instance with the same time horizon $T$, it is still NP-hard \cite{packetcomplexity}. Thus, we forfeit the ability to obtain an upper bound in each iteration, and instead we will restrict the time horizon to be \chs $T' = \lceil (1 + \alpha) \hat{T}\rceil$ \che for some $\alpha \geq 0$ (we used \chs $\alpha = 0.01$ \che for our computations). This restriction of the time horizon allows us to detect if $\hat{x}$ can be converted to a solution to \chs UPR($D_T$) \che with value at most $T'$.  Let $D_{T'} = (N_{T'}, A_{T'} \cup H_{T'})$ be the fully time-expanded network with time horizon $T'$. We now present the following upper bound formulation. Note that the arc capacities and node storage levels match the values given in the original \chs UPR($D_T$) \che instance.
\begin{align}\tag{UPR-FP$(\{D_{T'}^k\}_{k \in \mathcal{K}}, D_{T'})$}
    \min~~ & \bar{T} \\
    \vspace{.5cm}
    \mbox{s.t.} ~~
    & t' \cdot x_e^k \leq \bar{T} \quad \forall k \in \mathcal{K}, ~ \forall e = ((v,t), (w,t')) \in A^k_{T'} \\
    & \chs \sum_{e=((v,t),(w,t')) \in A_{T', k}^{k, final}} t' \cdot x_e^k \leq \bar{T} \quad \forall k \in \mathcal{K} \che \\
    & \sum_{e \in \delta^{+}_{D_{T'}^k}(v,t)} x_e^k - \sum_{e \in \delta^{-}_{D_{T'}^k}(v,t)} x_e^k = \begin{cases}
        1 ~(v,t) = (s_k, 0) \\
        -1 ~(v,t) = (t_k, T') \\
        0 ~\mbox{otherwise}\\
    \end{cases} \quad \forall k \in \mathcal{K}, \chs (v,t) \in N_{T'}\che \\
    \vspace{.5cm}
    & \sum_{k \in \mathcal{K}} x_e^k \leq u_e \quad \forall e \in A_{T'}\\ 
    & \sum_{k \in \mathcal{K}_v} x_e^k \leq b_e \quad \forall e \in H_{T'} \label{const:storage}\\
    & x_e^k \in \{0,1\} \quad \forall k \in \mathcal{K}, e \in A^k_{T'} \cup H^k_{T'}.
\end{align}

This gives the following upper-bound procedure. \chs Algorithm \ref{alg:UB} takes as input an optimal solution $\hat{x}$ to the current partially time-expanded network with makespan $\hat{T}$. We also take as input the optimality factor tolerance $\alpha \geq 0$ and the current best-known upper bound on $T^*$, denoted $\mathtt{UB}$. $\hat{x}$ defines a trajectory $\hat{Q}_k$ for each $k \in \mathcal{K}$, and projecting this down to the base graph defines a path $P_k$ for each $k \in \mathcal{K}$. We then solve the UPR problem with fixed paths with an upper bound of $T' = \lceil (1+\alpha) \hat{T}\rceil $, where each $k \in \mathcal{K}$ must follow a trajectory with underlying path $P_k$. If the problem is feasible, the we check if the value $V$ is less than our current best upper bound and output the current feasible solution. Otherwise we return the original upper bound.\che

\begin{algorithm}[ht]
\hspace*{\algorithmicindent} \textbf{Input}: Base network $D = (N,A)$, an optimal partial network solution $\hat{x}$ with value $\hat{T}$, and parameter  \\
\hspace*{\algorithmicindent} $\alpha \geq 0$, and a current upper bound on the value of $T^*$, denoted $\mathtt{UB}$
\begin{algorithmic}[1]  
    	\State $T' = \lceil (1 + \alpha) \hat{T} \rceil$, and $\bar{x} = \emptyset$
    	\State $D_{T'} = (N_{T'}, A_{T'} \cup H_{T'})$
	\State Let $\hat{Q} = \{\hat{Q}_k\}_{k \in \mathcal{K}}$ denote the set of trajectories given by $\hat{x}$
	\State For each $k \in \mathcal{K}$ let $P_k$ denote the underlying path in $D$ of trajectory $\hat{Q}_k$
	\State $A^k_{T'} = \{ ((v, t), (w, t')) \in A_{T'}: vw \in A(P_k)\}$, and  $H^k_{T'} = \{ ((v, t), (v, t')) \in H_{T'}: v \in N(P_k)\}$
	\State $D^k_{T'} = (N_{T'}, A^k_{T'} \cup H^k_{T'})$
	\State Solve UPR-FP($\{D^k_{T'}\}_{k \in \mathcal{K}}, D_{T'}$)
    	\If{UPR-FP($\{D^k_{T'}\}_{k \in \mathcal{K}}, D_{T'}$) is feasible and has optimal value $\bar{T}'$}
		\State Let $\bar{x}$ be an optimal solution to UPR-FP($\{D^k_{T'}\}_{k \in \mathcal{K}}, D_{T'}$)
		\State $\mathtt{UB} = \min\{\mathtt{UB}, \bar{T}'\}$
	\Else
    		\State $\mathtt{UB} = \mathtt{UB}$
   	\EndIf
    	\State \Return $\mathtt{UB}, ~\bar{x}$
    \caption{Compute-UB($D, \hat{x}, \hat{T}, \alpha, \mathtt{UB}$)}
	\label{alg:UB}
 \end{algorithmic}
\end{algorithm}	
\vspace{-.5cm}
\FloatBarrier

\subsection{Augmentation step}\label{sec:augment}

We now consider the case where the partial solution $\hat{x}$ cannot be converted to a solution of UPR($D_T$) of equal cost using the upper bound model (line \chs 12 was executed \che in Algorithm \ref{alg:UB}). In this section, we will detail how to augment the set $N_S$. 

Due to our relaxation procedure, we know that $\hat{x}$ may not be convertible to a solution to UPR($D_T$) with equal makespan due to shortened arcs in $D_S$, relaxed arc capacities, and relaxed node storage levels. 

Let $e = ((v,t), (w, t')) \in A_S \cup H_S$ be a timed arc in the support of $\hat{x}$. There is a well-established method to correct short arcs in $A_S$ \cite{Boland1}: 
\begin{center}
if $t' < t + \tau_{vw}$, we add the timed node $(w, t + \tau_{vw})$.
\end{center} 
We now proceed to deal with arcs exceeding arc and storage capacities. For each timed arc $e = ((v,t), (w, t')) \in A_S \cup H_S$, let $\hat{x}_e$ be the total active flow assigned to arc $e$ according to $\hat{x}$. That is,  
\[\hat{x}_e = \begin{cases}
    \sum_{k \in \mathcal{K}} \hat{x}_e^k \quad  e \in A_S \\
    \sum_{k \in \mathcal{K}_v} \hat{x}_e^k \quad e \in H_S
\end{cases}\]
If $e \in A_S$ and $\hat{x}_e > u_e$, then by construction of $u'$, $(v,t+1)\notin N_S$. Thus, we will add  $(v,t+1)$ to $N_S$. 

If instead $e \in H_S$ and $\hat{x}_e > b_e$, by definition of $b'_e$ and $\mathtt{m}_S(v, t)$, it follows that for some $(z, \bar{t}) \in N^-_T(v,t)$ or $(z, \bar{t}) \in N^-_S(v,t)$, we have $\mathtt{m}_S(z, \bar{t}) > 1$. For each $(z, \bar{t}) \in N^-_S(v,t)$ with $\mathtt{m}_S(z, \bar{t}) > 1$, we add $(z, \bar{t} + 1)$ to $N_S$. We also add $(v, t+1)$ to $N_S$ if it is not yet in the set, which then ensures $N^-_S(v,t) \subseteq N^-_T(v,t)$ in the next iteration. 

Algorithm 5 on the following page restates each of these procedures.

\begin{algorithm}[h]
\hspace*{\algorithmicindent} \textbf{Input}: Current partially time-expanded network $D_S = (N_S, A_S \cup H_S)$, base graph $D = (N,A)$, packet set  \\ 
\hspace*{\algorithmicindent} $\mathcal{K}$, and an optimal solution $\hat{x}$ to UPR($D_S$)
\begin{algorithmic}[1]  
    	\State $N'_S \leftarrow N_S$
	\For{$e = ((v,t), (w,t')) \in \mbox{supp}(\hat{x}):= \{e \in A_S \cup H_S: \hat{x}_e > 0\}$}
	    	\If{$e \in A_S$}
	        		\State  Compute the flow assigned to timed arc $e$ according to $\hat{x}$, $\hat{x}_e := \sum_{k \in \mathcal{K}} \hat{x}^k_e$.	
	        		\If{$t' < t + \tau_{vw}$}
	        			\State $N'_S \leftarrow N'_S \cup \{(w, t + \tau_{vw})\}$
	        		\EndIf
	        		\If{$\hat{x}_e > u_e$}
	        			\State $N'_S \leftarrow N'_S \cup \{(v, t + 1)\}$.
            		\EndIf
        		\EndIf
        		\If{$e \in H_S$ }
            		\State Compute the relevant flow assigned to timed arc $e$ according to $\hat{x}$, $\hat{x}_e := \sum_{k \in \mathcal{K}_v} \hat{x}^k_e$.
            		\If{$\hat{x}_e > b_v$}
            			\ForAll{$(z, \bar{t} = t - \tau_{z,v}) \in N^-_T(v,t): \mathtt{m}_S(z, \bar{t}) > 1$}
            				\State $N'_S \leftarrow N'_S \cup \{(z, \bar{t} + 1)\}$.
           			 \EndFor
           			 \State $N'_S \leftarrow N'_S \cup \{(v, t + 1)\}$.
        			\EndIf
    			\State \Return $N'_S$
		\EndIf
	\EndFor
	\caption{Augment-$N_S$($D_S, D, \mathcal{K}, \hat{x}$)}
	\label{alg:aug_DDD}
 \end{algorithmic}
\end{algorithm}	
\vspace{-.5cm}
\FloatBarrier

\begin{prop}\label{prop:augment}
Given an instance of UPR with minimum makespan $T^*$, Algorithm 5 only adds timed nodes $(v,t)$ to $N_S$ with $t \leq T^* + 1$.
\end{prop}
\begin{proof}
In any iteration, the optimal solution $\hat{T}$ of UPR($D_S$) is at most $T^*$ since UPR($D_S$) is a relaxation of UPR($D_T$). When correcting an arc $e = ((v,t), (v, t')) \in H_S$ due to exceeded storage capacity, we know that $t' \leq T^*$, since $v$ is not the destination for the commodities contributing to $\hat{x}_e$, and by constraint (\ref{constr:IPS_timing}). Furthermore, we add nodes $(w,t)$ with $t \leq t' + 1$ for this correction since $\tau \geq 0$. 

Now consider the correction of an arc $((v,t), (w,t'))$ in $A_S$. If the arc exceeds capacity $u$, then we add node $(v,t+1)$ to $N_S$. Since $t \leq T^*$, clearly $t+1 \leq T^* +1$. Finally, if the arc is too short, then we add the node $(w,t + \tau_{vw})$ to $N_S$. Due to our replacement of constraint (\ref{constr:IPT_timing}) with constraint (\ref{constr:IPS_timing}), we see that $t + \tau_{vw} \leq \hat{T} \leq T^*$. 
\end{proof}

Following along the lines of the proof, we easily obtain Corollary \ref{cor:added_nodes}. 

\begin{cor}\label{cor:added_nodes}
Given an instance of UPR with minimum makespan $T^*$ with $\tau_a > 0$ for all $a \in A$, Algorithm 5 only adds nodes $(v,t)$ to $N_S$ with $t \leq T^*$.
\end{cor}

Proposition \ref{prop:augment} points to the strength of the DDD approach over solving UPR($D_T$) when the upper bound $T$ given ends up being much larger than $T^*$. The DDD approach will maintain a much smaller time-expanded network throughout the algorithm. 

In the original application of DDD to \chs SND \che \cite{Boland1}, the solution to the upper bound model dictated which timed arcs were to be corrected in the augmentation step. However, in our model, we correct every timed arc in the support of the optimal solution to UPR$(D_S)$ that is too short, or has exceeded the original arc and storage capacities. As a result, it is not necessary to run the upper bound procedure in each iteration, and instead it may save time to only run the procedure when the makespan reported by two consecutive iterations is similar. While we solved the upper bound model in each iteration in our experiments, it would be worthwhile testing this alternative approach.

\begin{algorithm}[h]
\hspace*{\algorithmicindent} \textbf{Input}: Base network $D = (N,A)$, commodity set $\mathcal{K}$, an upper bound, $T$, on the optimal makespan, and  \\
\hspace*{\algorithmicindent} an optimality parameter $\alpha \geq 0$
\begin{algorithmic}[1]  
    	\State $N_S \leftarrow $  Generate-Initial-$N_S$($D, \mathcal{K}, T)$
	\State $\bar{x} = \emptyset$
	\State $\mathtt{UB} \leftarrow T$
	\State $\mathtt{LB} \leftarrow 0$
	\State gap = $(\mathtt{UB} - \mathtt{LB})/\mathtt{UB}$
	\While{gap $ > \alpha$ or $\bar{x} \neq \emptyset$}
	    	\State $A_S, H_S \leftarrow $ Generate-$A_S\cup H_S(N_S, D, T)$
        		\State $D_S\leftarrow(N_S, A_S \cup H_S)$ 
	    	\State Solve UPR($D_S$), and let $\hat{x}$ be an optimal solution, with value $\hat{T}$
	    	\State $\mathtt{LB} \leftarrow \max\{\mathtt{LB}, \hat{T}\}$
	    	\State $\mathtt{UB}, \bar{x} \leftarrow $ Compute-UB($D, \hat{x}, \hat{T},\alpha, \mathtt{UB}$)
	    	\State gap = $(\mathtt{UB} - \mathtt{LB})/\mathtt{UB}$
        		\If{gap $ \leq \alpha$ and $\bar{x} \neq \emptyset$}
            		\State Stop. An solution within $\alpha$ of optimal has been found for UPR($D_T$).
            		\State \Return $\bar{x}, \mathtt{UB}$
        		\Else
        			\State $N_S \leftarrow $ Augment-$N_S$($D_S, D, \mathcal{K}, \hat{x}$)
        		\EndIf
	\EndWhile
	\caption{Solve UPR-DDD($D, \mathcal{K}, T, \alpha$)}
	\label{alg:DDD_total}
 \end{algorithmic}
\end{algorithm}	
\FloatBarrier

\chs
We now prove correctness of our algorithm as well as bound the number of iterations. It is important to note that we can bound the number of iterations in terms of $T^*$ and not just $T$.

\begin{theorem}\label{thm:termination}
The algorithm $\mbox{UPR-DDD}(D, \mathcal{K}, T,\alpha)$ terminates with solution that has makespan at most an $(1 + \alpha)T^*$in at most $|N|T^*$ iterations.
\end{theorem}
\begin{proof}
First recall that since all input data is integral, and we are given that $T^* \leq T$, the decision times of an optimal solution are in $[T]$. 

Consider an iteration of the algorithm where the partially time-expanded network is $D_S = (N_S, A_S \cup H_S)$ and the relaxed capacities are given by $u'$ and $b'$. Let $\hat{Q}$ be the set of trajectories in $D_S$ that gives a min makespan routing. For each $k \in \mathcal{K}$, let $\hat{Q}_k$ denote the trajectory for packet $k$ in $\hat{Q}$. Let $\hat{T}$ denote the makespan of $\hat{Q}$ in $D_S$. Let $\bar{Q}$ be the set of trajectories we obtain by solving the corresponding upper bound, and suppose the factor gap between the two makespans is greater than $\alpha$. 

It follows that we could not obtain trajectories in $D_T$ with the same underlying paths as $\hat{Q}$ while satisfying the original capacities $u$ and $b$, given a time horizon of $\alpha \hat{T}$. Specifically, it must have been infeasible to simply assign each packet the trajectory $\hat{Q}_k$ in $D_T$. Thus, it must be that some timed arc in $\hat{Q}$ was too short, or exceeded the arc capacity or node storage level. As we argued in Section \ref{sec:augment}, in each scenario there must have been a timed node $(v,t) \in N_T\setminus N_S$ that we can add to $N_S$. 

Furthermore, we proved in Proposition \ref{prop:augment} that our algorithm only adds timed nodes $(v,t)$ to $N_S$ with $t \leq T^* + 1$. Since in each iteration we add at least one timed node and in the first iteration we have at least one copy of each node, the DDD algorithm terminates within $|N| T^*$ iterations.
\end{proof}

\begin{cor}\label{cor:termination}
The algorithm $UPR-DDD(D, \mathcal{K}, T,\alpha=0)$ terminates with an optimal solution in at most $|N|T^*$ iterations.
\end{cor}
\che

\vspace{-0.5cm}
\section{Computational results}\label{sec:computations}
To demonstrate the effectiveness of our DDD algorithm, we compare the runtime of the DDD algorithm and the original full integer program UPR($D_T$) when applied to geographic and geometric instances. For the geographic instances, we base the node and arc selection on the population centres in the United States. Our geometric instances are constructed to model social networks. 

For each problem instance, we initially solve the DDD instance with a sufficiently large time horizon $T$ so that $T^* < T$. This initial solve gives us the value of $T^*$. Then to compare the solve time for DDD and the full IP, we run each algorithm with the time horizon upper bound of $T^*, 1.5 T^*$, and $2T^*$ for up to two hours. Thus, in total we solve each instance seven times. Note, we use upper bounds $T$ as factors of $T^*$ only for analysis purposes. In practice, we would select a value of $T$ that is sufficiently large so that all packets could be routed within time $T$. Each algorithm was coded in \chs Python 3.6.9 with Gurobi 8.1.1 \cite{gurobi} \che as the optimization solver. The running time limit was set to 7200 seconds (two hours) using the deterministic option of the solver and the instances are solved to within 1\% of optimality. The instances were run in a 64 cores 2.6GHz Xeon Gold 6142 Processor with 256GB RAM, running a Linux operating system. Each instance was run with a limit of 5 cores. The generated instances can be found at \textsf{https://github.com/madisonvandyk/UPRlib}.

\subsection{Geographic instances}\label{subsec:geographic}

\subsubsection*{Dataset}
For the base graph, we use the locations of the top $n$ most populated cities in the USA. We randomly select $m$ arcs to form $A$, and set $\tau_a$ to be the distance in hundreds of miles, rounded up to the nearest integer. \chs We compute the shortest directed path between each pair of vertices. We then select $k$ random origin-destination pairs from the digraph $D$ such that there is dipath from the origin to the destination, and the shortest path has at least $\delta$ arcs and length at most a factor $\gamma$ times the max shortest path length of any pair. We construct the origin-destination pairs in this way to ensure that min makespan is not simply the max length of the shortest path\che. For arc and node capacities, we follow a discrete version of the approach of Crainic et al. \cite{Crainic2, Crainic1} that was developed as a rigorous test set for \chs SND\che. This dataset construction has since been modified to analyze the performance of DDD algorithms \cite{Boland1, LagosDDD}. We select capacities from a discrete uniform distribution with endpoints $[\alpha_1, \alpha_2]$ and $[\beta_1, \beta_2]$ for node storage. Crainic et al. \cite{Crainic2, Crainic1} introduced the \emph{capacity ratio} $C = |A| k /\sum_{e \in A} u_e$. As $C$ approaches 1, the network is lightly capacitated, and the congestion level increases as $C$ increases. Crainic studied scenarios with $C \in \{1,2,8\}$ when solving \chs SND\che. However, since in UPR we are not incentivized to consolidate packet flow as is the case of \chs SND\che, we need much more restrictive congestion to generate problems of interest (minimal congestion would allow all packets to be routed along shortest paths, using no node storage).  We now list our set of parameters.

\subsubsection*{Parameters}
\begin{itemize}[noitemsep]
    \item $n = 20$ -- number of nodes;
    \item \chs $m \in \{30, 45, 60\}$\che -- number of arcs;
    \item \chs $k \in \{200, 250, 300\}$\che -- number of packets;
    \item $(\alpha_1, \alpha_2) \in \chs \{(1, \lceil 0.01 k \rceil), (1, \lceil 0.0175 k \rceil), (1, \lceil 0.025 k \rceil)\}$\che -- bounds for arc capacity;
    \item $(\beta_1, \beta_2) \in \chs \{(0, \lceil 0.01 k \rceil), (0, \lceil 0.0175 k \rceil), (0, \lceil 0.025 k \rceil) \rceil\}$ \che -- bounds storage capacity;
    \item $\delta =3, \gamma = 0.90$.
\end{itemize}
The choices for $\alpha_1, \alpha_2$, and $\beta_1, \beta_2$ allow for a range of congestion levels, while ensuring that the resulting instance is always feasible. The three values of $m$ are selected to ensure we examine the DDD algorithm on networks of varying connectivity, \chs while maintaining that the instances are capacitated -- overly dense graphs would allow for many vertex pairs to be connected via few arcs, decreasing the congestion\che. Similarly, the range of values of $k$ allows us to analyze the effectiveness of our DDD algorithm on various densities. 

We offer a quick overview of how the above parameters impact the optimal time horizon $T^*$ as well as the overall solve time for UPR($D_T$), holding all other parameters constant. As $m$ increases, $T^*$ decreases since packets can travel via shorter direct paths, and fewer packets are forced to overlap. While larger $m$ would imply that the UPR($D_T$) takes longer to generate, since $T^*$ decreases significantly the overall solve time decreases in our experiments.  Naturally, as $k$ increases, $T^*$ increases. As expected, as $\alpha$ and $\beta$ increase, $T^*$ decreases. 

\subsubsection*{Results}
We first present the average runtime (in seconds) among all settings of $\alpha$, $\beta$, when $T$, $m$, and $k$ are fixed. We note that the ``ratio'' column denotes the average ratio of the runtimes, rather than the ratio of the average runtimes. Averages marked with $^*$ indicate that there was at least one instance that did not terminate within the time limit. 

\begin{table}[h]
\makebox[\textwidth][c]{
    \begin{tabular}{c|rrr|rrr|rrr}
\toprule
      & \multicolumn{1}{c}{} & \multicolumn{1}{c}{\chs $k = 200$} & \multicolumn{5}{c}{\chs $k = 250$} & \multicolumn{1}{c}{\chs $k = 300$} \\
\midrule
\chs UB &  \chs UPR($D_T$) & \chs DDD  & \chs ratio & \chs UPR($D_T$) & \chs DDD  & \chs ratio & \chs UPR($D_T$) & \chs DDD  & \chs ratio \\ \hline
\chs $T^*$ &  \chs 267  & \chs 878  & \chs 2.82 & \chs 288  & \chs 693  & \chs 2.02 & \chs 674  & \chs 1,463$^*$ & \chs 1.98 \\
\chs $1.5T^*$ & \chs 709  & \chs 1,293 & \chs 1.85 & \chs 1,582 & \chs 1,135 & \chs 0.67 & \chs 1,761 & \chs 1,602$^*$ & \chs 0.88 \\
\chs $2T^*$ & \chs 1,891$^*$ & \chs 959$^*$  & \chs 0.85$^*$ & \chs 2,714$^*$ & \chs 1,368$^*$ & \chs 0.42$^*$ & \chs 2,879$^*$ & \chs 1,838$^*$ & \chs 0.51$^*$ \\
\bottomrule
\end{tabular}%
}
\caption{$m = 30$.}
\label{tab:addlabel}%
\end{table}%
\FloatBarrier

\begin{table}[h]
\makebox[\textwidth][c]{
    \begin{tabular}{c|rrr|rrr|rrr}
\toprule
      & \multicolumn{1}{c}{} & \multicolumn{1}{c}{\chs $k = 200$} & \multicolumn{5}{c}{\chs $k = 250$} & \multicolumn{1}{c}{\chs $k = 300$} \\
\midrule
\chs UB & \chs UPR($D_T$) & \chs DDD  & \chs ratio & \chs UPR($D_T$) & \chs DDD  & \chs ratio & \chs UPR($D_T$) & \chs DDD  & \chs ratio \\ \hline
\chs $T^*$ & \chs 144  & \chs 303 & \chs 2.08 & \chs 210  & \chs 338 & \chs 1.73 & \chs 219  & \chs 345  & \chs 1.68 \\
\chs $1.5T^*$ & \chs 406  & \chs 467 & \chs \chs 1.25 & \chs 882  & \chs 734 & \chs 0.93 & \chs 975  & \chs 814  & \chs 0.79 \\
\chs $2T^*$ & \chs 1,491 & \chs 580 & \chs \chs 0.40 & \chs 2,232 & \chs 652 & \chs 0.34 & \chs 3,100$^*$ & \chs 1,024 & \chs 0.42$^*$ \\
\bottomrule
\end{tabular}%
}
\caption{$m = 45$.}
\label{tab:addlabel}%
\end{table}%
\FloatBarrier


\begin{table}[h]
\makebox[\textwidth][c]{
    \begin{tabular}{c|rrr|rrr|rrr}
\toprule
      & \multicolumn{1}{c}{} & \multicolumn{1}{c}{\chs$k = 200$} & \multicolumn{5}{c}{\chs $k = 250$} & \multicolumn{1}{c}{\chs$k = 300$} \\
\midrule
\chs UB & \chs UPR($D_T$) & \chs DDD  & \chs ratio & \chs UPR($D_T$) & \chs DDD  & \chs ratio & \chs UPR($D_T$) & \chs DDD  & \chs ratio \\ \hline
\chs$T^*$ & \chs135  & \chs315 & \chs2.25 & \chs152  & \chs320 & \chs2.13 & \chs189  & \chs467  & \chs2.44 \\
\chs$1.5T^*$ & \chs390  & \chs517 & \chs1.41 & \chs485  & \chs652 & \chs1.27 & \chs800  & \chs1,958 & \chs2.04 \\
\chs$2T^*$ & \chs1,374 & \chs532 & \chs0.55 & \chs1,513 & \chs726 & \chs0.55 & \chs2,814$^*$ & \chs2,162$^*$ & \chs0.76$^*$ \\
\bottomrule
\end{tabular}%
}
\caption{$m = 60$.}
\label{tab:addlabel}%
\end{table}%
\FloatBarrier

Across all scenarios we see a clear trend that as the upper bound $T$ increases relative to $T^*$, the increase to the runtime to UPR($D_T$) is much greater than the increase to the runtime of DDD. This result is intuitive, since as $T$ increases, UPR($D_T$) becomes larger and additional symmetries are introduced in the network -- most packets will have an increasing number of possible trajectories in an optimal solution. Since the partially time-expanded networks are sparse, many of these additional symmetries are avoided. 

Figures 8 - 13 report the runtime of the experiments when $T = 1.5T^*$ and $T=2T^*$. Each plot presents the number of instances solved within a given time limit. 
These figures demonstrate the same findings as the tables above. The performance of DDD becomes increasingly advantageous over the full IP as $m$ decreases and $T$ increases. 

\begin{figure}[htb!]
    \centering
    \begin{minipage}{.5\textwidth}
        \centering
        \includegraphics[width=\textwidth]{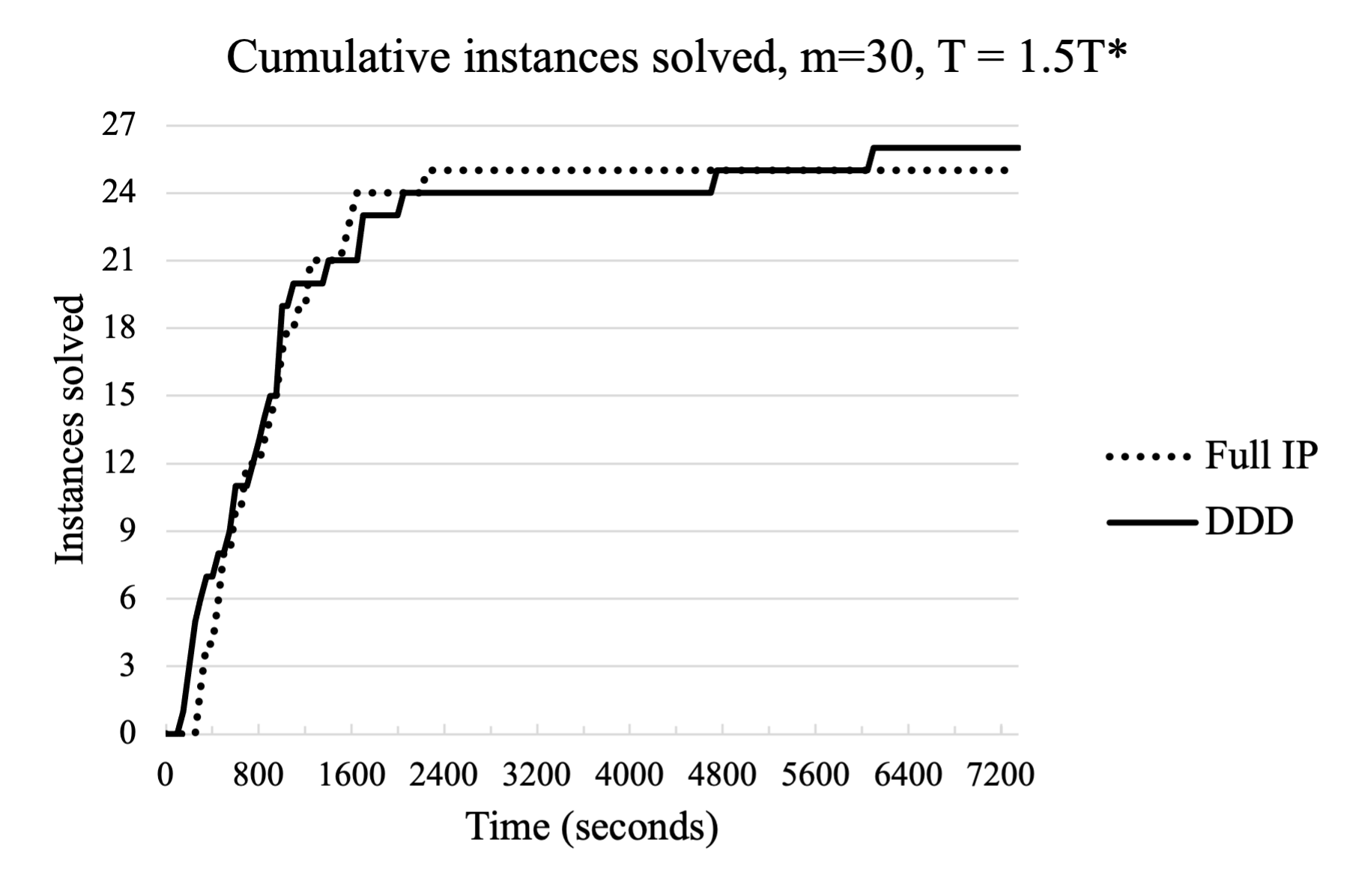}
        \caption{$m=30, T=1.5T^*$}
        \label{}
    \end{minipage}%
    \begin{minipage}{0.5\textwidth}
        \centering
        \includegraphics[width=0.94\textwidth]{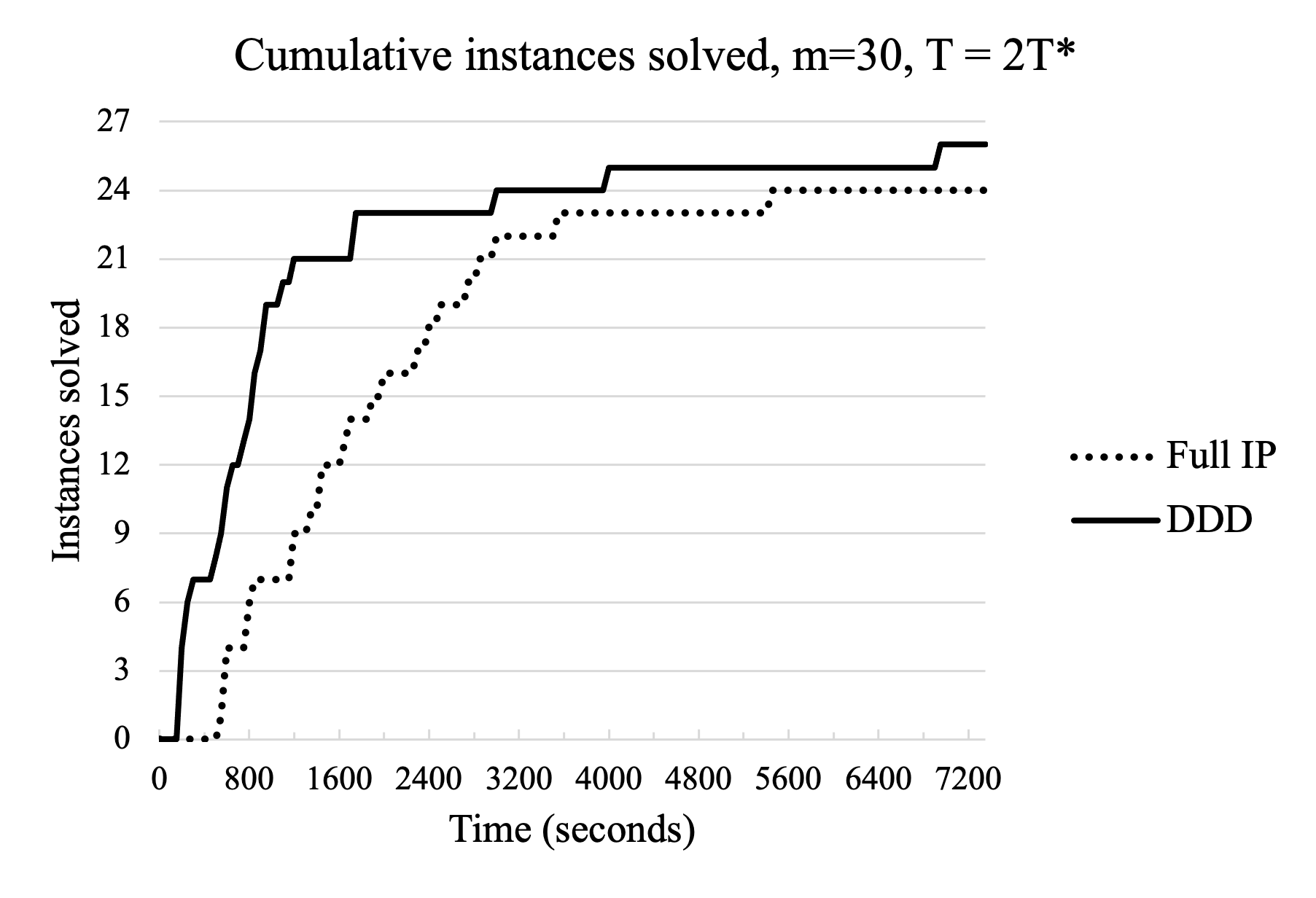}
        \caption{$m=30, T=2T^*$}
        \label{}
    \end{minipage}
\end{figure}

\begin{figure}[!htb]
    \centering
    \begin{minipage}{.5\textwidth}
        \centering
        \includegraphics[width=\textwidth]{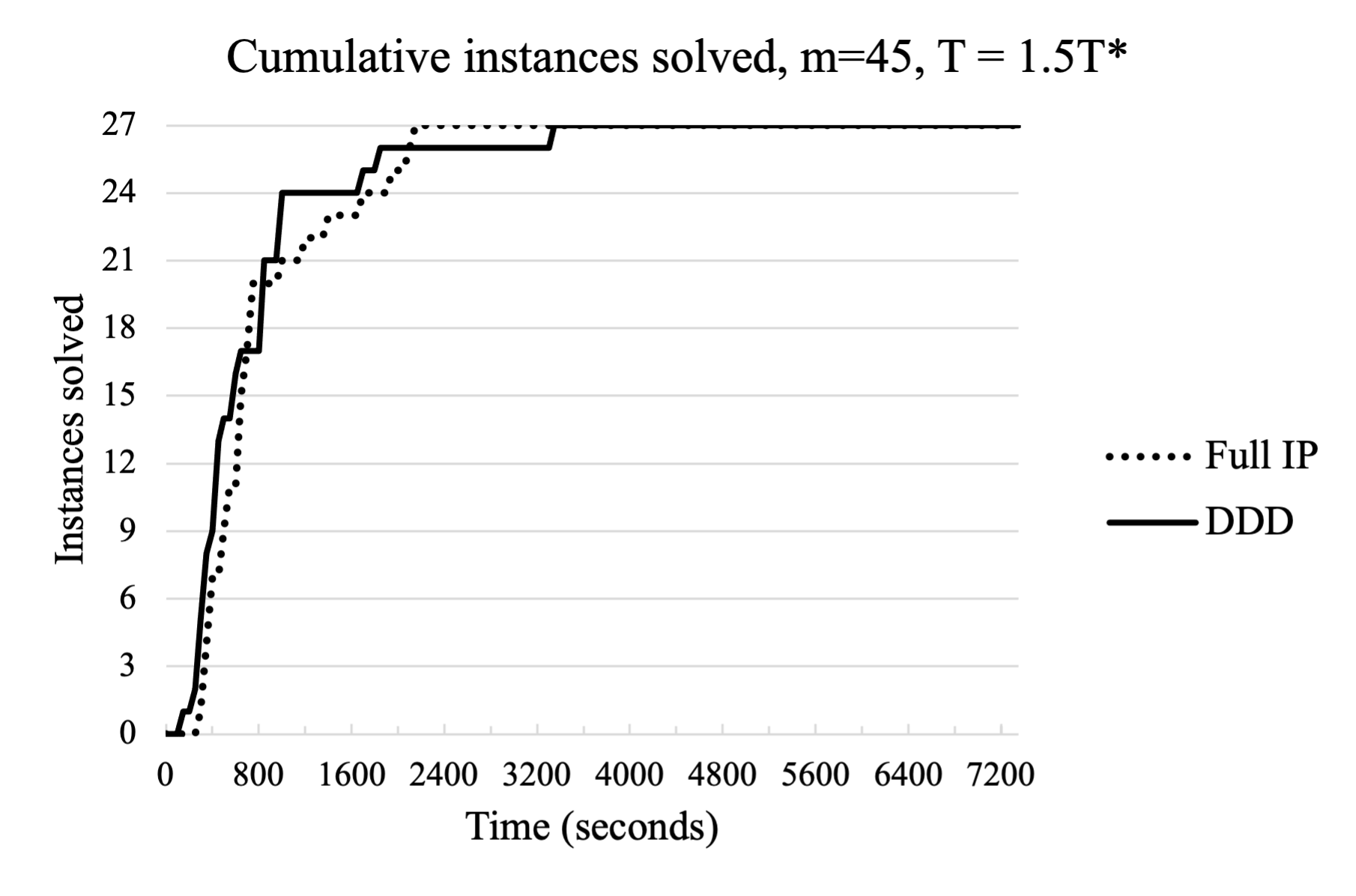}
        \caption{$m=45, T=1.5T^*$}
        \label{}
    \end{minipage}%
    \begin{minipage}{0.5\textwidth}
        \centering
        \includegraphics[width=0.94\textwidth]{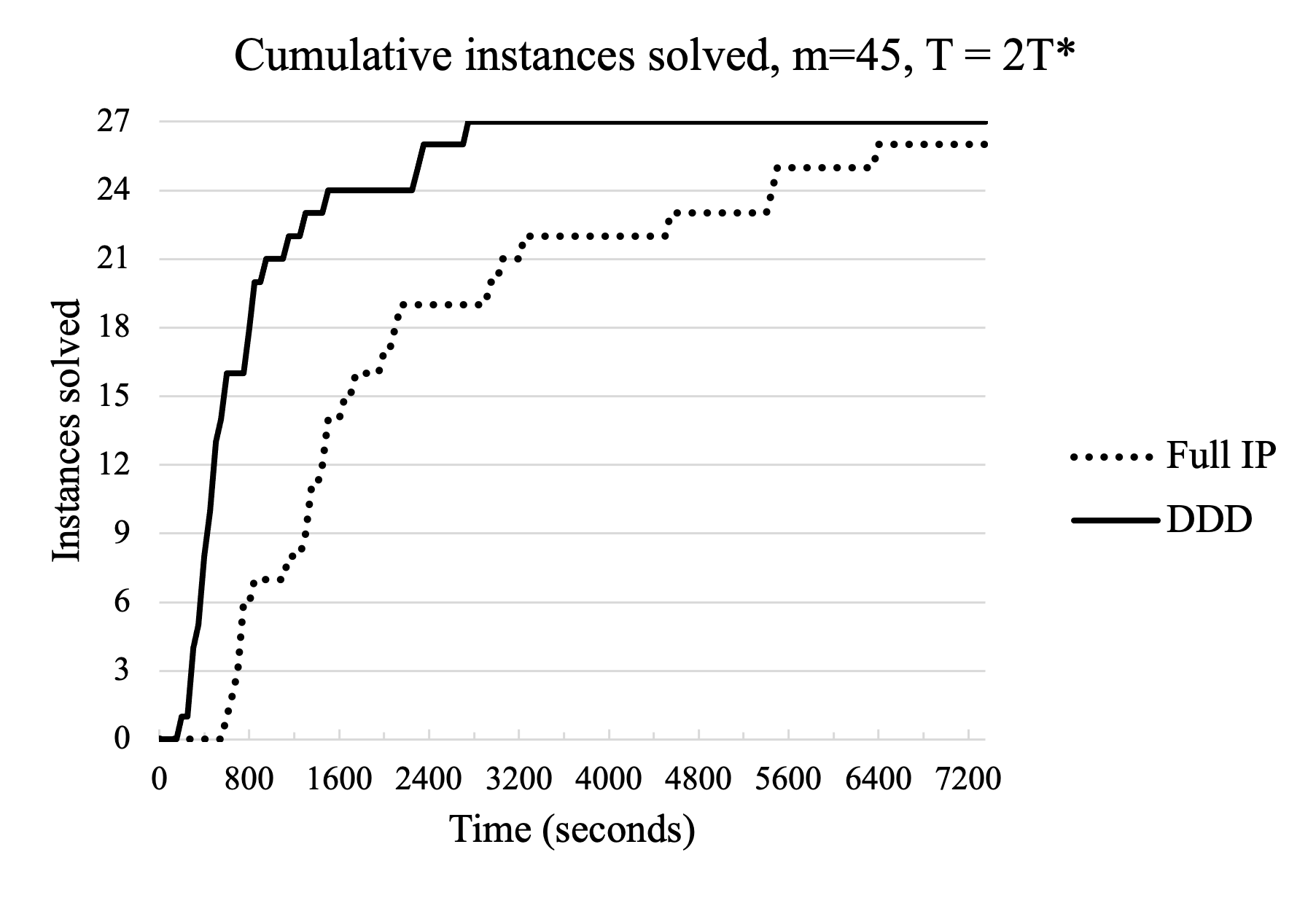}
        \caption{$m=45, T=2T^*$}
        \label{}
    \end{minipage}
\end{figure}

\begin{figure}[!htb]
    \centering
    \begin{minipage}{.5\textwidth}
        \centering
        \includegraphics[width=\textwidth]{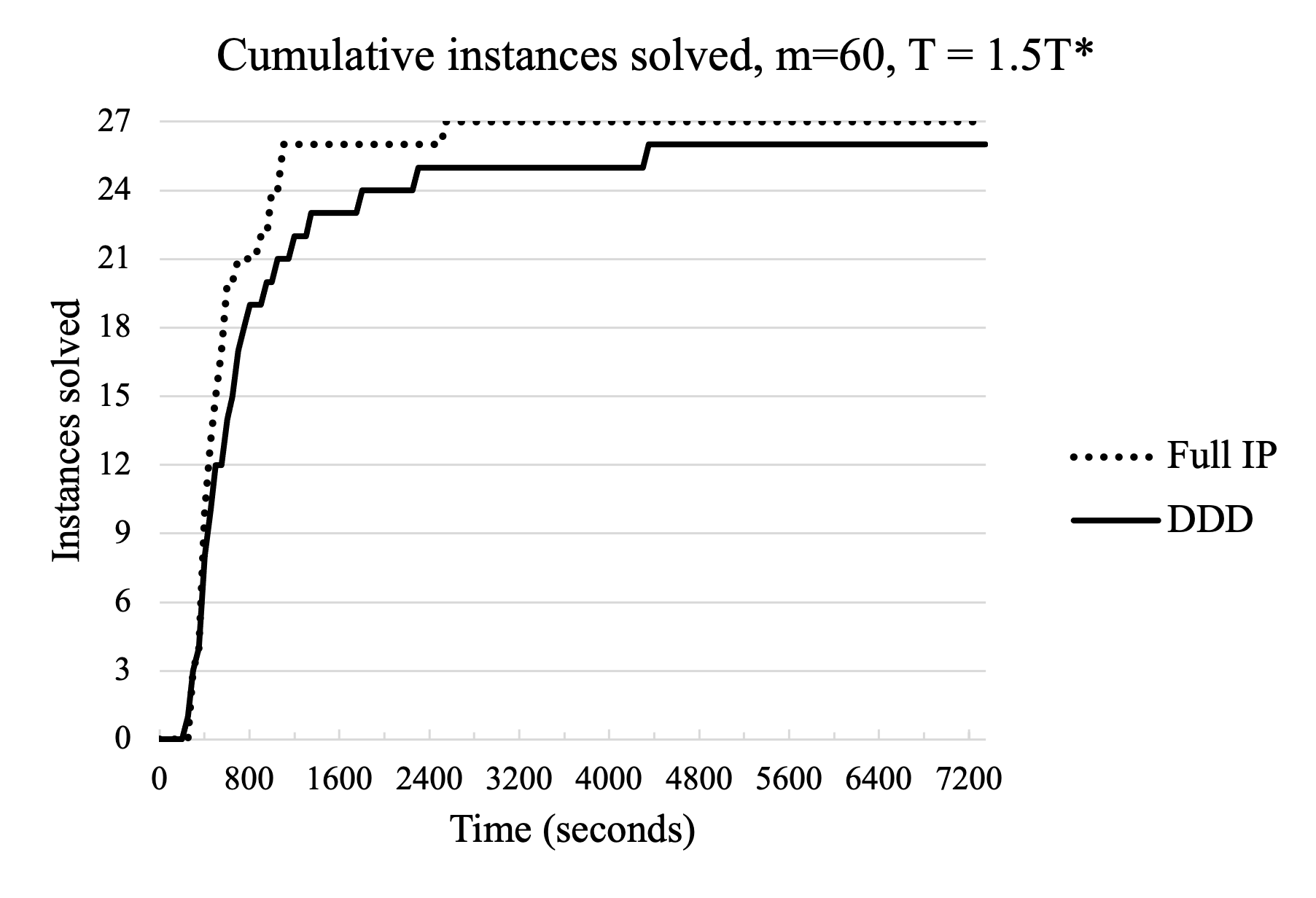}
        \caption{$m=60, T=1.5T^*$}
        \label{}
    \end{minipage}%
    \begin{minipage}{0.5\textwidth}
        \centering
        \includegraphics[width=0.94\textwidth]{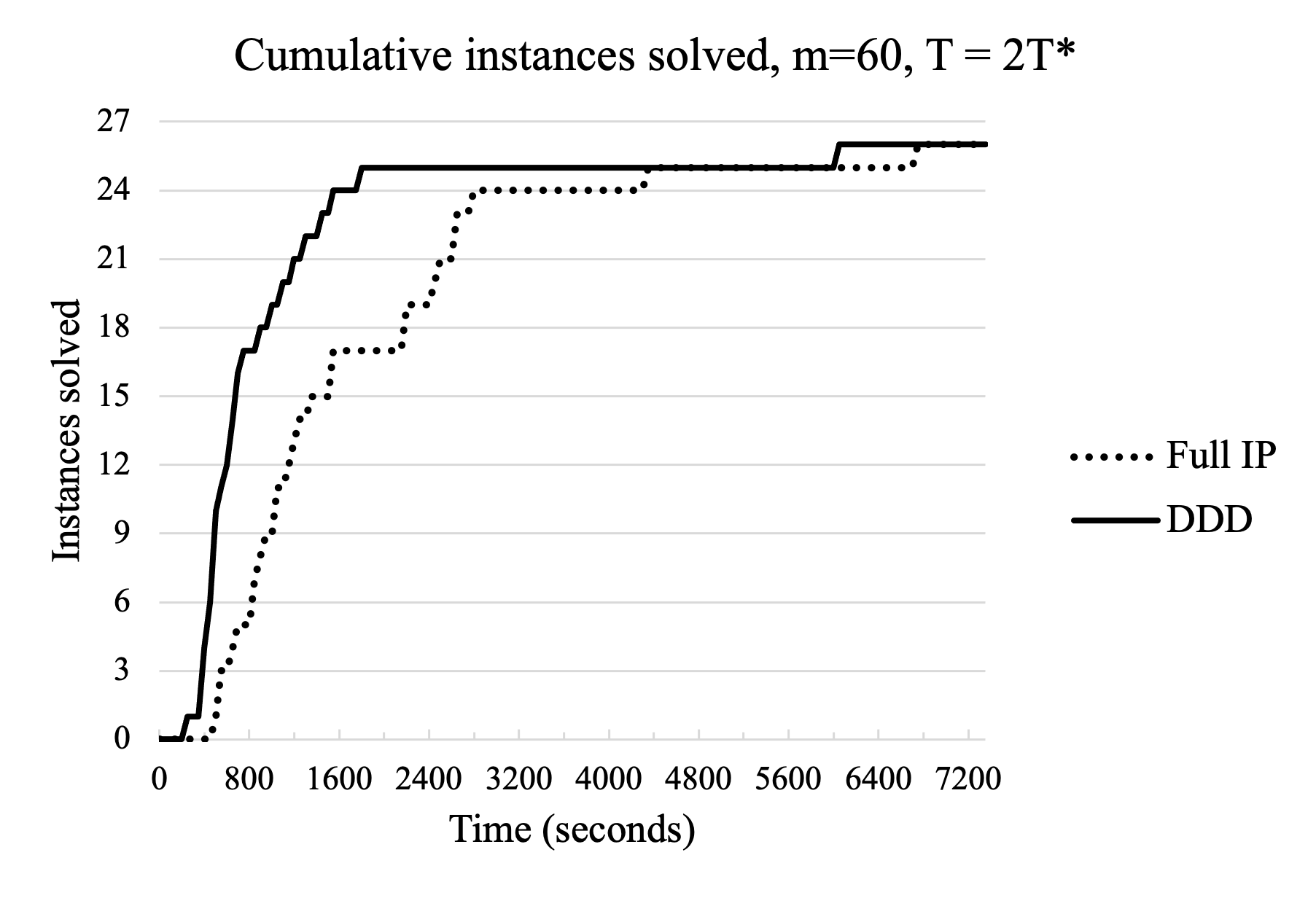}
        \caption{$m=60, T=2T^*$}
        \label{}
    \end{minipage}
\end{figure}
\FloatBarrier

\chs
\subsubsection*{Iteration sizes}
When $T = T^*$, UPR($D_T$) performs better than DDD. This is not surprising since while DDD solves smaller IPs than UPR($D_T$), we still require reasonably dense partially time-expanded networks in order for the algorithm to terminate. The following table presents the average number of iterations required until DDD terminates, as well as the average size of the final timed node set compared to the full timed node set. 
\begin{table}[h!]
\makebox[\textwidth][c]{
    \begin{tabular}{c|cc|cc|cc}
\toprule
      & \multicolumn{2}{c}{\chs $m = 30$} & \multicolumn{2}{c}{\chs $m = 45$} & \multicolumn{2}{c}{\chs $m = 60$} \\
\midrule
\chs factor & \chs iterations & \chs $|N_S^{\mathtt{final}}|/|N_T|$ & \chs iterations & \chs $|N_S^{\mathtt{final}}|/|N_T|$ & \chs iterations & \chs $|N_S^{\mathtt{final}}|/|N_T|$ \\ \hline
\chs $T^*$ & \chs 9.67 & \chs 0.43 & \chs 7.00 & \chs 0.60 & \chs 6.93 & \chs 0.72 \\
\chs $1.5T^*$ & \chs 9.85 & \chs 0.31 & \chs 6.93 & \chs 0.40 & \chs 6.78 & \chs 0.47 \\
\chs $2T^*$ & \chs 9.89 & \chs 0.23 & \chs 7.07 & \chs 0.31 & \chs 7.07 & \chs 0.37\\
\bottomrule
\end{tabular}
}
\caption{Average number of iterations, and relative size of $N_S^{\mathtt{final}}$.}
\label{tab:addlabel}%
\end{table}%
\FloatBarrier

The average number of iterations decreases as the number of arcs, $m$, increases. This is not surprising since a larger number of arcs allows for packet trajectories with fewer timed arcs, and thus DDD generates feasible trajectories in fewer iterations. At the same time, the average ratio of $|N_S^{\mathtt{final}}|/|N_T|$ increases with $m$, which explains why we do not see better performance for DDD for higher values of $m$. One reason this ratio is higher is due to the refinement process for storage capacity, since when correcting exceeded storage at a node $v$, the number of timed nodes added partly depends on the degree of $v$. We explore the impact of sparsity on the performance of DDD further in Section \ref{sec:computations_geometric}.

\subsubsection*{Refinement trends}
In each iteration we solve the integer program defined on the partially time-expanded network and obtain a solution $\hat{x}$. If $\hat{x}$ cannot be converted to an optimal solution in $D_T$, there must be timed arcs in the support of $\hat{x}$ that are either too short, or exceed the original throughput or storage levels. In the following figures, we examine how each of these violated constraint types influences the refinement process in each iteration.  

In Figure \ref{fig:geographic_iter_infeasible_arcs} we present the average proportion of the infeasible arcs in the support of $\hat{x}$ that are short, have exceeded throughput, or exceeded storage in each iteration. In Figure \ref{fig:geographic_iter_added_nodes}, we present the average proportion of the timed nodes added to correct each violated constraint type in each iteration. We included the test instances where DDD took at least 7 iterations before terminating, and only looked at the first 7 iterations. We see below that this turns out to be sufficient to observe clear trends.

\begin{figure}[!htb]
    \centering
    \begin{minipage}{.5\textwidth}
        \centering
        \includegraphics[width=\textwidth]{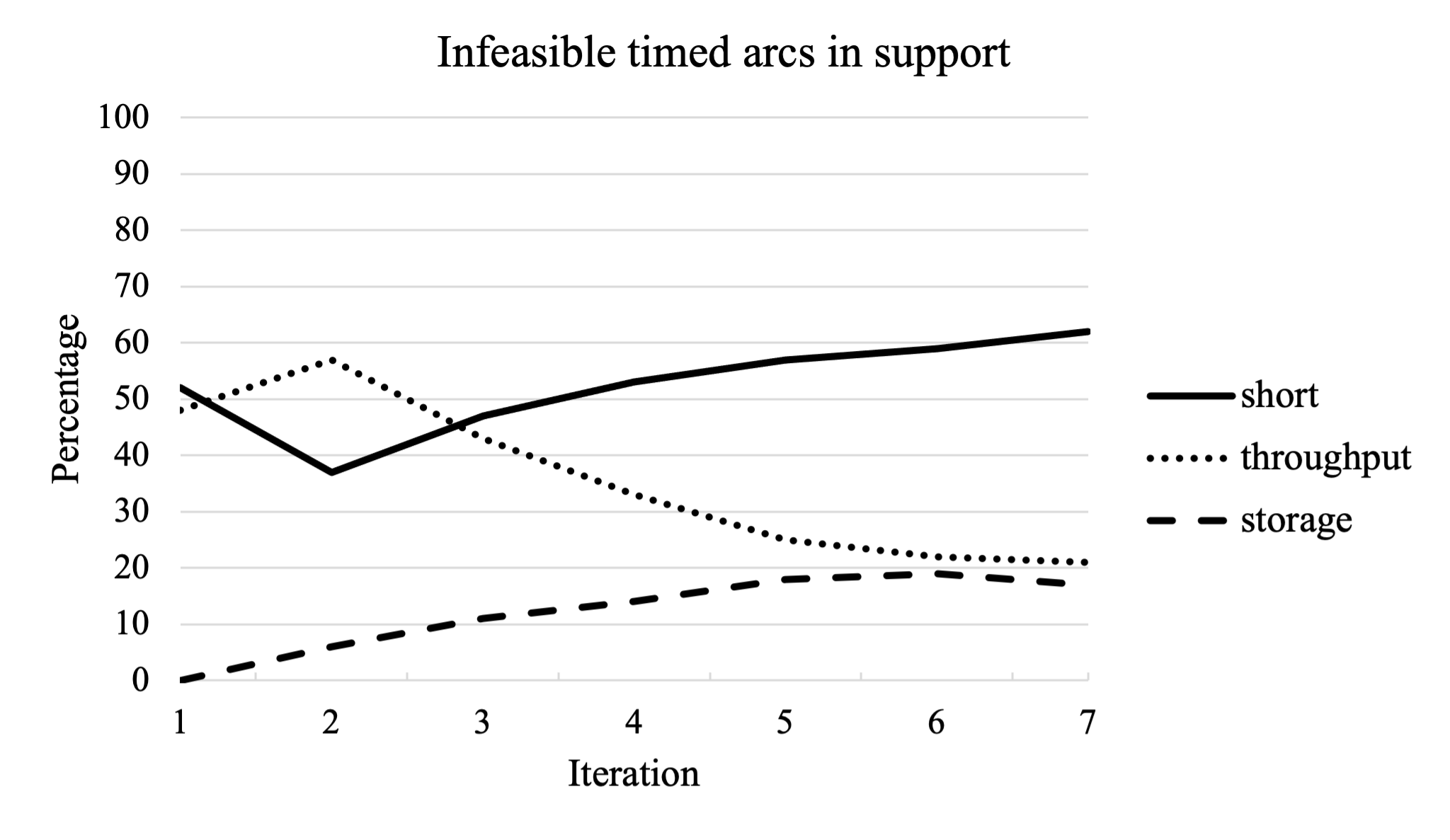}
        \caption{Proportion of infeasible timed arcs.}
        \label{fig:geographic_iter_infeasible_arcs}
    \end{minipage}%
    \begin{minipage}{0.5\textwidth}
        \centering
        \includegraphics[width=\textwidth]{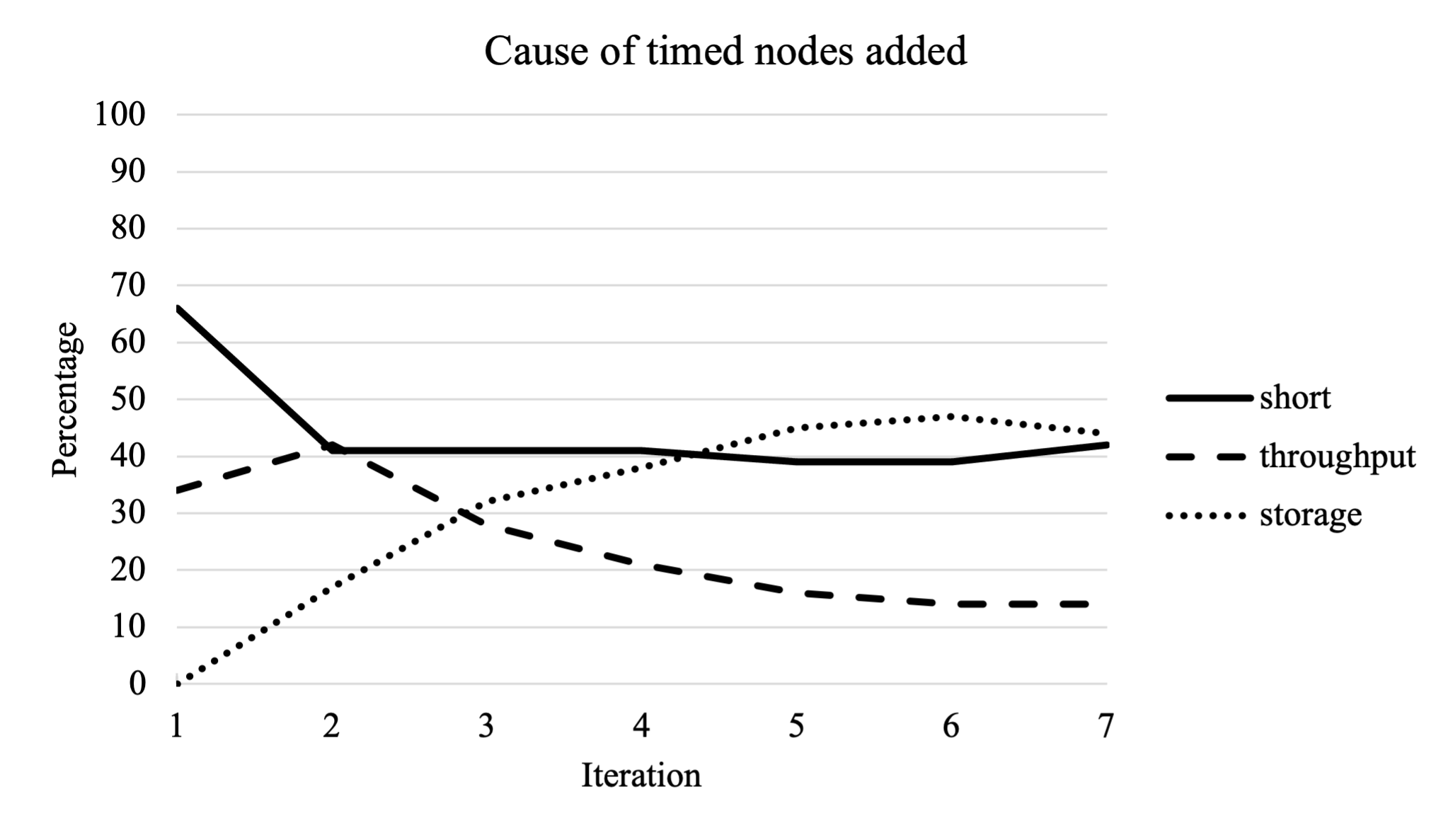}
        \caption{Cause of adding timed nodes}
        \label{fig:geographic_iter_added_nodes}
    \end{minipage}
\end{figure}
Observe that there are no timed arcs in the support of $\hat{x}$ exceeding storage capacity in the first iteration. This is not surprising since no storage is necessary if there are no (or at least no restrictive) throughput capacities.  For the same reason, it is natural that there would be more timed arcs with exceeded throughput than timed holdover arcs with exceeded storage in each iteration. As the iterations progress, the support of $\hat{x}$ increases in size and spreads the flow of packets along a larger number of timed arcs. Thus, it is natural that the proportion of infeasible timed arcs with exceeded throughput decreases. In Figure \ref{fig:iter_added_nodes}, we see that the proportion of timed nodes added due to exceeded storage overtakes those added due to exceeded throughput. This is due to the fact that when correcting exceeded throughput, we add a single timed node, whereas when correcting exceeded storage at a node $v$, we may add up to $\deg_D(v) + 1$ timed nodes. 
\che
\subsection{Geometric instances}\label{sec:computations_geometric}
\subsubsection*{Dataset}
Our instances consist of random geometric graphs which have been used widely to model human social networks \cite{Kleinberg, Milgram}. We begin by selecting $n$ nodes randomly from an $l \times l$ grid. We then connect each ordered pair of nodes with an arc if their L1-norm distance is at most $p$. This is a discrete version of a random geometric graph. Since social networks have low diameter, we augment our arc set according to the popular construction of Kleinberg \cite{Kleinberg}. That is, for each of the nodes $v \in N$, we add $q$ ``long-range'' arcs $(v,w)$ chosen independently at random, where the $i$th directed arc from $v$ has endpoint $w$ with probability proportional to $||v - w||_1^{-r}$. For each of the arcs generated to form $A$, we assign the transit time to be equal to the L1-norm distance between the endpoints. 

We then select $k$ random origin-destination pairs from the digraph $D$ according to the same process as in the geographic instances. As was the case of our geographic instances, we select capacities from a discrete uniform distribution with endpoints $(\alpha_1, \alpha_2)$ and $(\beta_1, \beta_2)$ for arc capacity and node storage respectively. We now list our set of parameters. 

\subsubsection*{Parameters}
\begin{itemize}[noitemsep]
    \item \chs$l = 25$ \che-- grid length and width;
    \item \chs $n = 20$\che -- number of nodes;
    \item \chs $k \in \{200, 225, 250\}$\che -- number of packets;
    \item \chs$p \in \{3, 4\}$\che -- radius for local connections;
    \item \chs$q \in \{1, \{1, 2\}\}$ \che-- number of long-range connections for each node;
    \item $r = 0.5$ -- scaling factor to select long-range connections;
    \item $(\alpha_1, \alpha_2) \in \chs \{(1, \lceil 0.01 k \rceil), (1, \lceil 0.02  k \rceil)\}$ \che -- bounds for arc capacity;
    \item $(\beta_1, \beta_2) \in \chs \{(0, \lceil 0.01 k \rceil), (0, \lceil 0.02  k \rceil) \rceil\}$\che -- bounds storage capacity.
\end{itemize}
\chs When $q = \{1,2\}$, for each node we select 1 or 2 long distance arcs with equal probability\che. We select $p$ and $q$ to be sufficiently small so that the instance is capacitated while still allowing for differing levels of local and global connectivity. In Figure \ref{fig:sparse} we see a sparse network obtained using $n$ and $l$ as stated, with $r = 0.5$, \chs $p = 3$,  and $q = 1$\che. In contrast, Figure \ref{fig:dense} shows a dense network obtained with parameters $r = 0.5$, \chs $p = 4$, and $q = \{1,2\}$\che. In each figure the placement of the nodes corresponds to the location in the $l \times l$ grid. 

\begin{figure}[!htb]
    \centering
    \begin{minipage}{.5\textwidth}
        \centering
        \includegraphics[width=.8\textwidth]{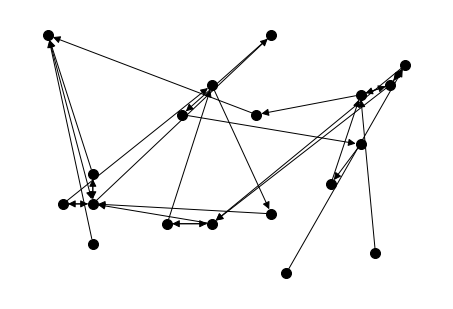}
        \caption{locally and globally sparse}
        \label{fig:sparse}
    \end{minipage}%
    \begin{minipage}{0.5\textwidth}
        \centering
        \includegraphics[width=.8\textwidth]{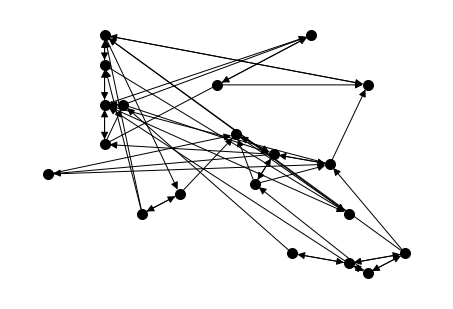}
        \caption{locally and globally dense}
        \label{fig:dense}
    \end{minipage}
\end{figure}

When $r = 0$, the long-range connections are chosen uniformly at random, and when $r = 1$, the connections are chosen with probability proportional to the inverse distance. Since we would like some level of global connectivity, but long-distance connections are less common than short distance connections in social networks, we use $r = 0.5$. As stated, we require higher congestion in UPR than in SND to generate interesting problems. The choices for $\alpha_1, \alpha_2$, and $\beta_1, \beta_2$ allow for a range of congestion levels, while ensuring that the resulting instance is feasible. 

We offer a quick overview of how the above parameters impact the optimal time-horizon $T^*$ as well as the overall solve time for UPR($D_T$), holding all other parameters constant. As $p$ and $q$ increase, $T^*$ decreases since packets can travel via shorter direct paths, and fewer packets are forced to overlap. As the arc and storage capacities increase, $T^*$ decreases since the network can accommodate a higher number of active packets. 

\subsubsection*{Results}
Overall, we see that DDD is faster than solving the full IP when $T = 2T^*$, and faster for slow sparse instances when $T= 1.5T^*$. On average, we find that when $T = 2T^*$, DDD completes in \chs 49\% \che of the time it takes to run UPR($D_T$). Note, this is an \emph{overestimate} since when running UPR($D_T$), some of the instances did not complete within the allowed time.

\begin{figure}[!htb]
    \centering
    \begin{minipage}{.5\textwidth}
        \centering
        \includegraphics[width=0.95\textwidth]{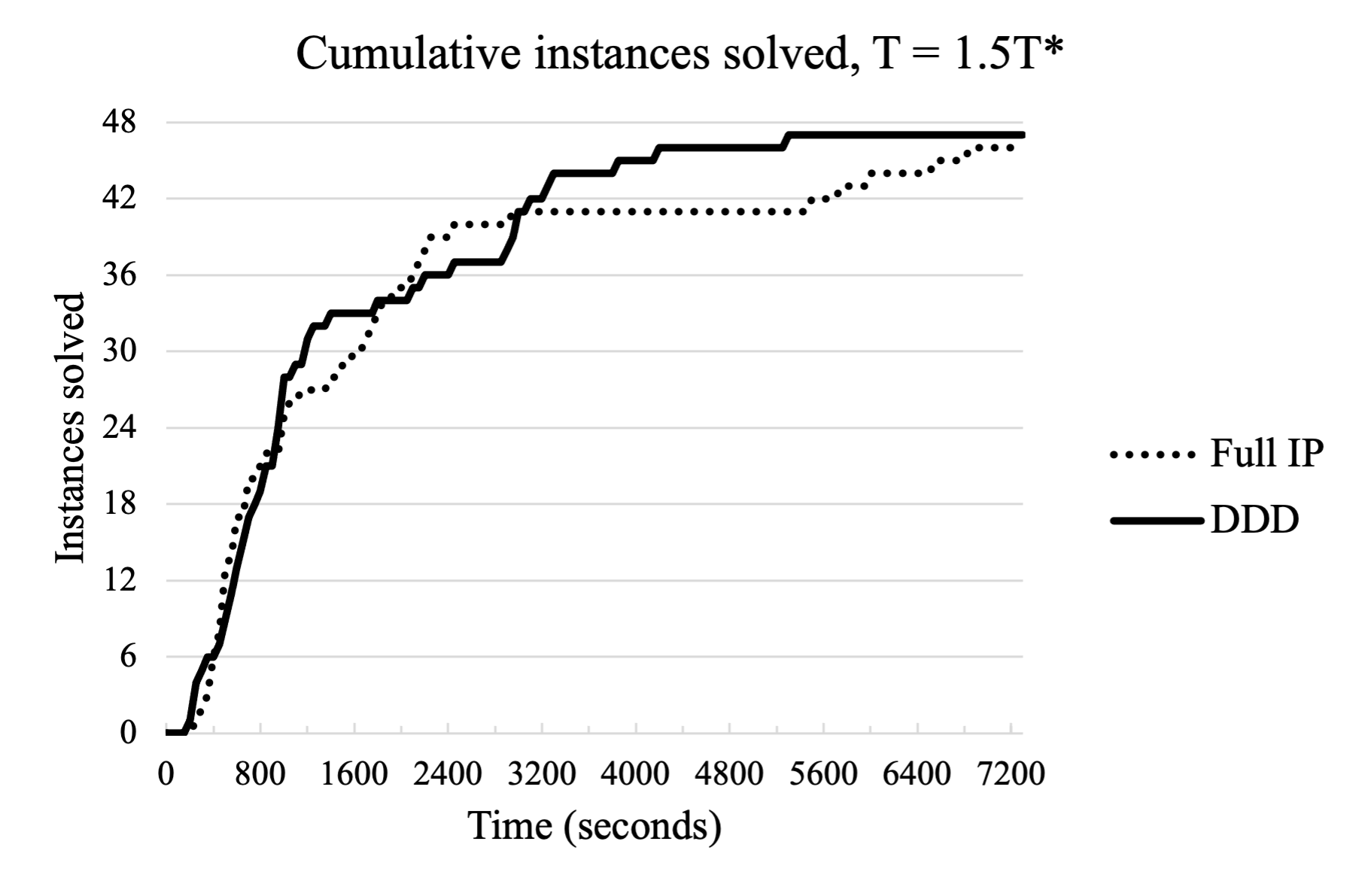}
        \vspace{-.25cm}
        \caption{cumulative instances solved, $T = 1.5T^*$}
        \label{}
    \end{minipage}%
    \begin{minipage}{0.5\textwidth}
        \centering
        \includegraphics[width=0.95\textwidth]{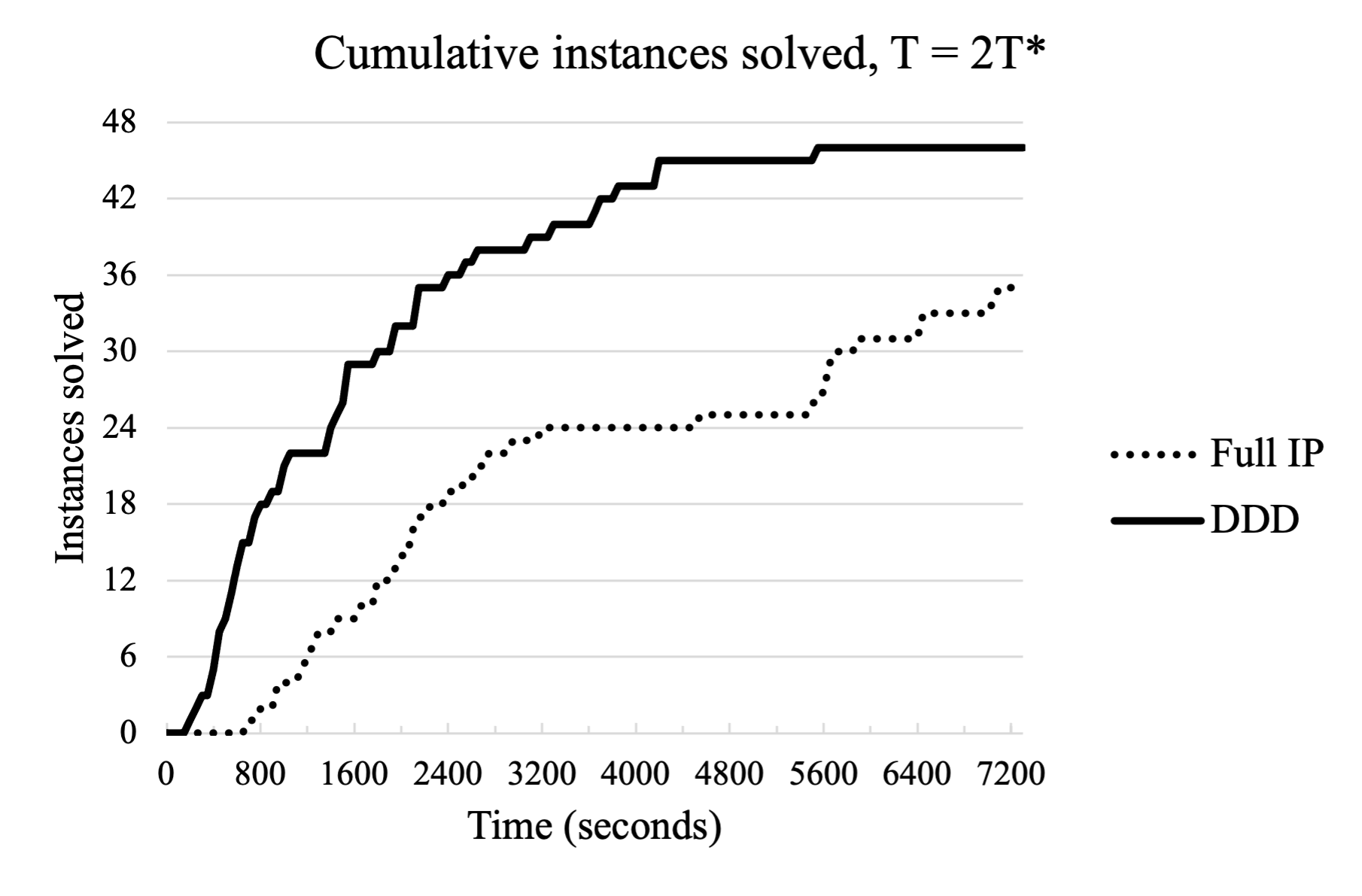}
        \vspace{-.25cm}
        \caption{cumulative instances solved, $T = 2T^*$}
        \label{}
    \end{minipage}
\end{figure}
\FloatBarrier

To examine the results more closely, we divide the instances into groups based on the level of local and global connectivity of the underlying network. Again, the ``ratio'' column given is the average of the ratio between the runtime of UPR($D_T$) and DDD. Again, averages marked with $^*$ indicate that there was at least one instance that did not terminate within the time limit. As a result, the ratio in the same row could be lower if all instances were run until termination. 

\newpage
\subsection*{Low local and low global connectivity}
\begin{table}[h!]
\makebox[\textwidth][c]{
    \begin{tabular}{c|rrr|rrr|rrr}
\toprule
      & \multicolumn{1}{c}{} & \multicolumn{1}{c}{\chs$k = 200$} & \multicolumn{5}{c}{\chs$k = 225$} & \multicolumn{1}{c}{\chs$k = 250$} \\
\midrule
\chs UB &\chs UPR($D_T$) &\chs DDD  &\chs ratio &\chs UPR($D_T$) &\chs DDD  &\chs ratio &\chs UPR($D_T$) &\chs DDD  & \chs ratio \\ \hline
\chs$T^*$ & \chs363   & \chs1,946 & \chs6.17 & \chs477 & \chs1,069 & \chs2.17 & \chs651 & \chs1,921 & \chs4.41 \\
\chs$1.5T^*$ & \chs2,705 & \chs3,220 & \chs1.18 & \chs5,066$^*$ & \chs2,702 & \chs1.15$^*$ & \chs3,644 & \chs1,523 & \chs0.47 \\
\chs$2T^*$ & \chs6,780$^*$ & \chs3,084 & \chs0.44$^*$ & \chs6,191$^*$ & \chs3,188 & \chs0.79$^*$ & \chs5,791$^*$ & \chs1,663 & \chs0.27$^*$ \\
\bottomrule
\end{tabular}%
}
\vspace{-.25cm}
\caption{$(p = 3, q = 1)$.}
\label{tab:addlabel}%
\end{table}%
\FloatBarrier
\vspace{-0.75cm}

\subsection*{Low local and high global connectivity}
\begin{table}[h!]
\makebox[\textwidth][c]{
    \begin{tabular}{c|rrr|rrr|rrr}
\toprule
      & \multicolumn{1}{c}{} & \multicolumn{1}{c}{\chs$k = 200$} & \multicolumn{5}{c}{\chs$k = 225$} & \multicolumn{1}{c}{\chs$k = 250$} \\
\midrule
\chs UB & \chs UPR($D_T$) & \chs DDD  & \chs ratio & \chs UPR($D_T$) & \chs DDD  & \chs ratio & \chs UPR($D_T$) & \chs DDD  & \chs ratio \\ \hline
\chs $T^*$ & \chs 234   & \chs 821   & \chs 3.20 & \chs 285   & \chs 550   & \chs 2.97 & \chs 245   & \chs 893   & \chs 4.09 \\
\chs $1.5T^*$ & \chs 1,205 & \chs 768   & \chs 0.87 & \chs 2,465$^*$ & \chs 1,235 & \chs 2.13$^*$ & \chs 1,029 & \chs 2,672 & \chs 3.61 \\
\chs $2T^*$ & \chs 2,803 & \chs 2,904$^*$ & \chs 0.75 & \chs 3,591$^*$ & \chs 1,391 & \chs 0.59$^*$ & \chs 6,284$^*$ & \chs 1,584 & \chs 0.37$^*$ \\
\bottomrule
\end{tabular}
}
\vspace{-.25cm}
\caption{$(p = 3, q = \{1,2\})$.}
\label{tab:addlabel}%
\end{table}%
\FloatBarrier
\vspace{-0.75cm}

\subsection*{High local and low global connectivity}
\begin{table}[htbp!]
\makebox[\textwidth][c]{
    \begin{tabular}{c|rrr|rrr|rrr}
\toprule
      & \multicolumn{1}{c}{} & \multicolumn{1}{c}{\chs $k = 200$} & \multicolumn{5}{c}{\chs $k = 225$} & \multicolumn{1}{c}{\chs $k = 250$} \\
\midrule
\chs UB & \chs UPR($D_T$) & \chs DDD  & \chs ratio & \chs UPR($D_T$) & \chs DDD  & \chs ratio & \chs UPR($D_T$) & \chs DDD  & \chs ratio \\ \hline
\chs $T^*$ & \chs 186   & \chs 478   & \chs 2.48 & \chs 248   & \chs 734   & \chs 2.76 & \chs 281   & \chs 791   & \chs 2.70 \\
\chs $1.5T^*$ & \chs 1,343 & \chs 1,581 & \chs 1.48 & \chs 1,092 & \chs 1,263 & \chs 1.23 & \chs 1,728 & \chs 1,687 & \chs 1.05 \\
\chs $2T^*$ & \chs 4,726$^*$ & \chs 2,538$^*$ & \chs 0.40$^*$ & \chs 3,776$^*$ & \chs 1,867 & \chs 0.49$^*$ & \chs 5,926$^*$ & \chs 1,880 & \chs 0.40$^*$ \\
\bottomrule
\end{tabular}%
}
\vspace{-.25cm}
\caption{$(p = 4, q = 1)$.}
\label{tab:addlabel}%
\end{table}%
\FloatBarrier
\vspace{-0.75cm}

\subsection*{High local and high global connectivity}
\begin{table}[htbp!]
\makebox[\textwidth][c]{
    \begin{tabular}{c|rrr|rrr|rrr}
\toprule
      & \multicolumn{1}{c}{} & \multicolumn{1}{c}{\chs $k = 200$} & \multicolumn{5}{c}{\chs $k = 225$} & \multicolumn{1}{c}{\chs $k = 250$} \\
\midrule
\chs UB & \chs UPR($D_T$) & \chs DDD  & \chs ratio & \chs UPR($D_T$) & \chs DDD  & \chs ratio & \chs UPR($D_T$) & \chs DDD  & \chs ratio \\ \hline
\chs $T^*$ & \chs 152   & \chs 318   & \chs 2.23 & \chs 187   & \chs 382 & \chs 2.12 & \chs 202   & \chs 438 & \chs 2.20 \\
\chs $1.5T^*$ & \chs 868   & \chs 533   & \chs 0.82 & \chs 796   & \chs 539 & \chs 0.82 & \chs 590   & \chs 834 & \chs 1.41 \\
\chs $2T^*$ & \chs 2,502 & \chs 1,093 & \chs 0.58 & \chs 3,925$^*$ & \chs 624 & \chs 0.31$^*$ & \chs 1,803 & \chs 896 & \chs 0.52 \\
\bottomrule
\end{tabular}%
}
\vspace{-.25cm}
\caption{$(p = 4, q \in \{1,2\})$.}
\label{tab:addlabel}%
\end{table}%
\FloatBarrier

Overall, we draw the same conclusions as in Section \ref{subsec:geographic}. As the upper bound increases relative to $T^*$, the performance of DDD improves over the the full integer program. For sparse instances, DDD outperforms UPR($D_T$) even when $T = 1.5 T^*$ \chs (DDD terminates in an average of 93\% of the time of the full IP). In the following table we see that when the underlying graph is sparser, this increases the average number of iterations until DDD terminates (``iter.'' column). Despite this increase in average iteration count, the average ratio of 
$|N_S^{\mathtt{final}}|/|N_T|$ decreases and so DDD still performs comparatively better on sparser instances. A future direction of research would be to examine how local and global connectivity impacts the performance of the DDD algorithm in general. \che

\begin{table}[h!]
\makebox[\textwidth][c]{
    \begin{tabular}{c|cc|cc|cc|cc}
\toprule
      & \multicolumn{2}{c}{\chs $p = 3$, $q=1$} & \multicolumn{2}{c}{\chs $p = 3$, $q=\{1,2\}$} & \multicolumn{2}{c}{\chs $p = 4$, $q=1$} & \multicolumn{2}{c}{\chs $p = 4$, $q=\{1, 2\}$} \\
\midrule
\chs factor & \chs iter. & \chs $|N_S^{\mathtt{final}}|/|N_T|$ & \chs iter. & \chs $|N_S^{\mathtt{final}}|/|N_T|$ & \chs iter. & \chs $|N_S^{\mathtt{final}}|/|N_T|$ & \chs iter. & \chs $|N_S^{\mathtt{final}}|/|N_T|$\\ \hline
\chs $T^*$ & \chs 9.75 & \chs 0.57 & \chs 8.42 &\chs  0.64 & \chs 7.33 & \chs 0.64 & \chs 6.17 & \chs 0.66 \\
\chs $1.5T^*$ & \chs 8.42 & \chs 0.37 & \chs 8.58 & \chs 0.45 & \chs 7.50 & \chs 0.44 & \chs 6.25 & \chs 0.46 \\
\chs $2T^*$ & \chs 8.58 &\chs  0.28 & \chs 8.25 & \chs 0.34 & \chs 7.67 & \chs 0.34 & \chs 6.58 & \chs 0.35\\
\bottomrule
\end{tabular}
}
\caption{Average number of iterations, and average value of $|N_S^{\mathtt{final}}|/|N_T|$.}
\label{tab:addlabel}%
\end{table}%
\FloatBarrier

\chs
\subsubsection*{Refinement trends}
We observe the same trends as in Section \ref{subsec:geographic} for the refinement step in DDD. As the iterations progress, correcting violated storage constraints increases in its impact on the refinement process, whereas there are fewer violated throughput constraints. Correcting short arcs continues to dominate the reason why most timed arcs are infeasible, while the cause of added timed nodes is split between correcting short timed arcs and correcting exceeded storage arcs.
\begin{figure}[!ht]
    \centering
    \begin{minipage}{.5\textwidth}
        \centering
        \includegraphics[width=\textwidth]{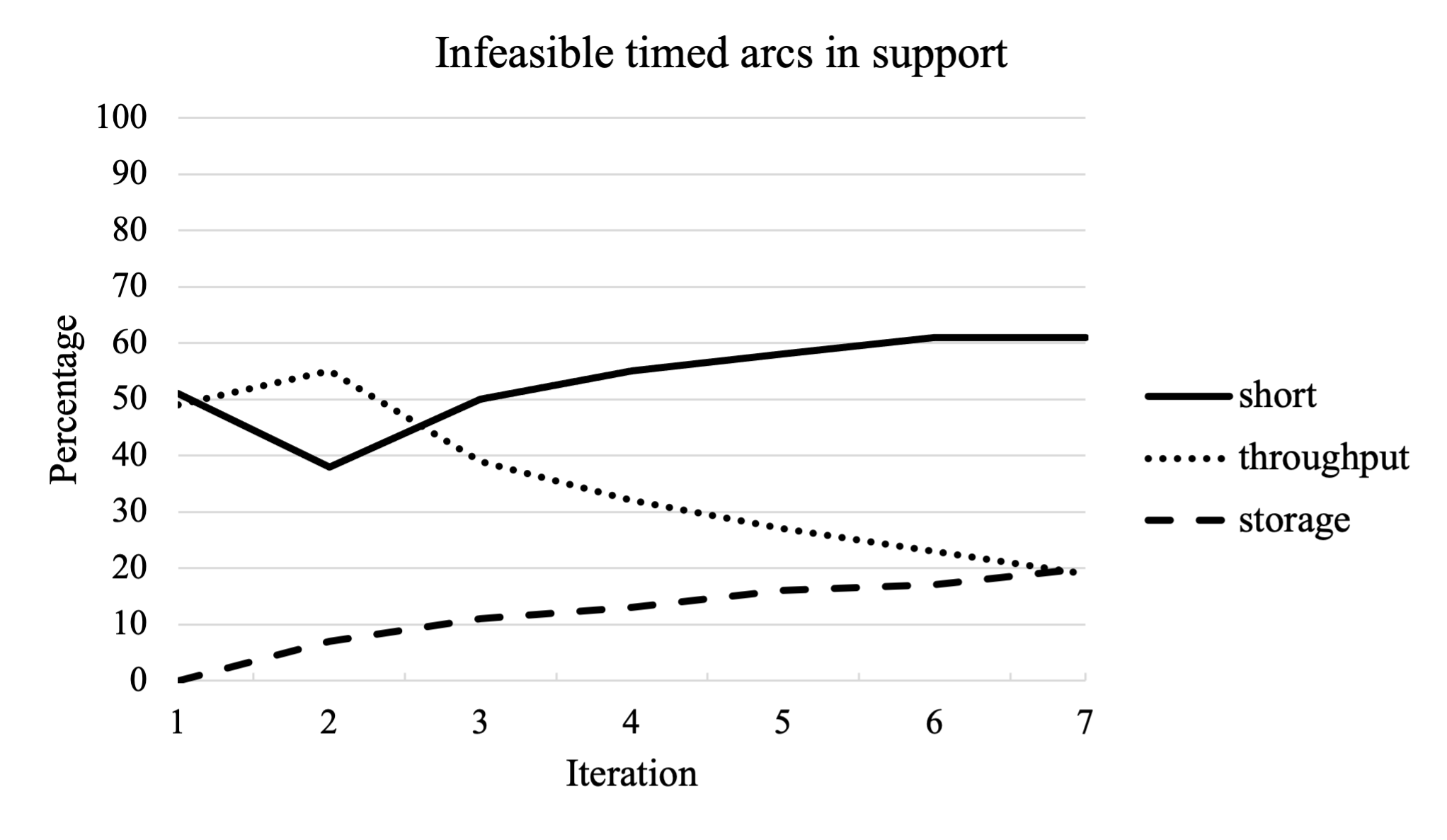}
        \caption{Proportion of infeasible timed arcs.}
        \label{fig:iter_infeasible_arcs}
    \end{minipage}%
    \begin{minipage}{0.48\textwidth}
        \centering
        \includegraphics[width=\textwidth]{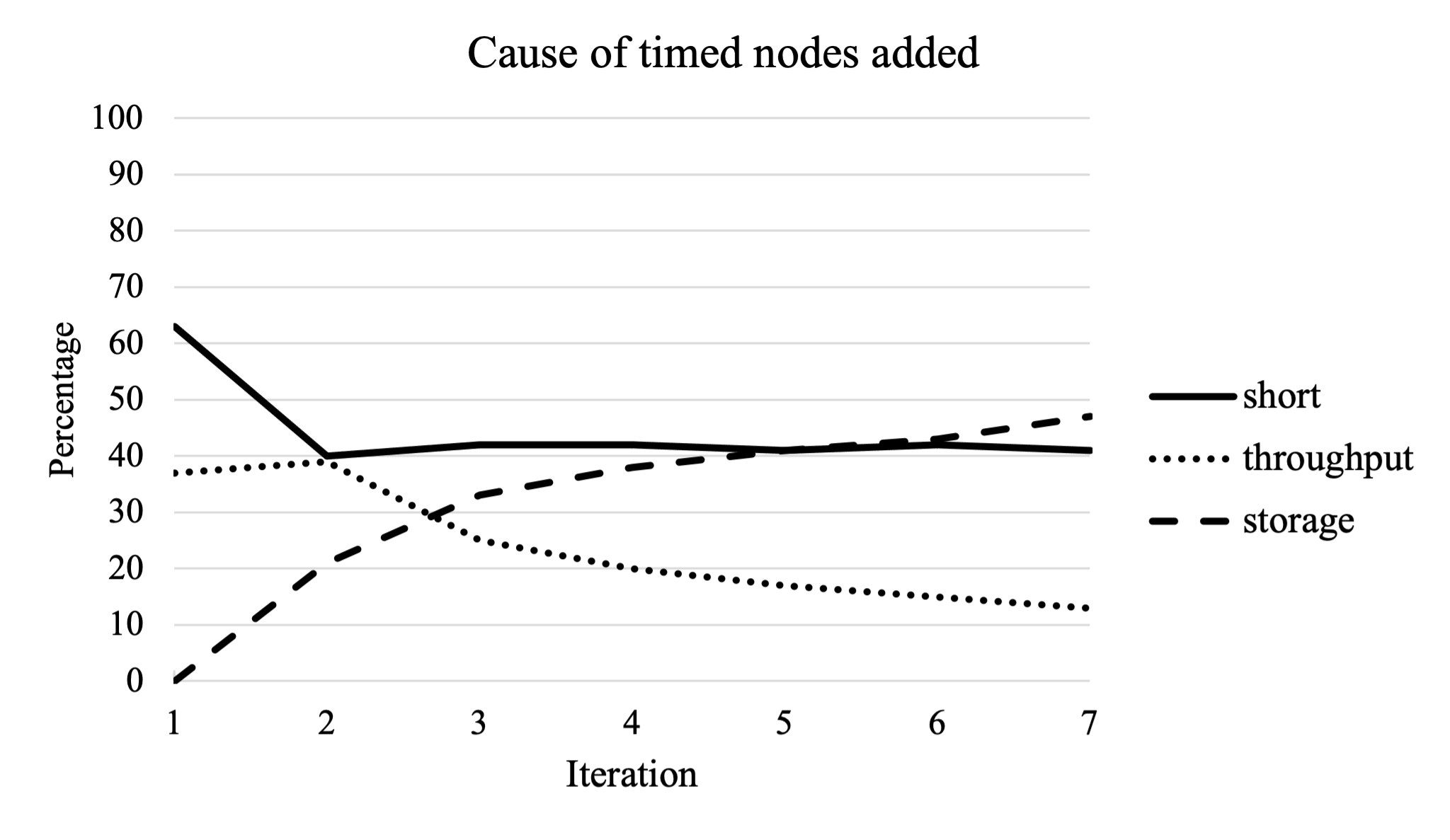}
        \caption{Cause of adding timed nodes}
        \label{fig:iter_added_nodes}
    \end{minipage}
\end{figure}
\FloatBarrier
\che

\section{Conclusion}

In this work, we study the universal packet routing problem and develop a novel DDD algorithm for solving this problem exactly. To the best of our knowledge, this work was the first to provide finite bounds on the relaxation required for storage levels in the DDD framework. To prove that our lower bound model is a relaxation, we present structural results of the map from the fully time-expanded network to the partially time-expanded network. These structural results would be helpful in extending the DDD framework to solve problems with inventory holding costs. 

We present an implementation along with illustrative computational results. In both our geographic and geometric instances, we observe that the runtime of the full IP increases more than the DDD algorithm as a function of the upper bound, $T$, provided. In many problems to which DDD has been previously applied, the support of an optimal solution is small compared to the fully time-expanded network. In UPR, this is no longer the case since the arc capacity and storage levels increase the size of the support of a solution. A significant contribution would be to alter the DDD approach so that different commodities have different partially time-expanded networks. The challenge would be to ensure we still obtain a lower bound on the optimal solution in each iteration. 

Finally, we observe that the performance of the full IP deteriorates as the upper bound provided increases relative to $T^*$. While continuous formulations are commonly said to perform more poorly than time-indexed formulations \cite{LagosDDD}, this will likely no longer be the case for some value of upper bound $T>> T^*$. Furthermore, it is possible that continuous formulations are more effective when paths are given in advance, such as is the case for the UPR-FP. \chs Potential future work would be to study the comparative runtime between continuous formulations and time-expanded formulations for problems when we add designated paths in the underlying static graph. \che Understanding this trade-off between continuous and time-indexed formulations may be of interest and could be used to improve upper bound and augmentation steps in DDD algorithms.

\subsection*{Acknowledgements}
We thank Cristiana L. Lara for valuable discussions and feedback.

\bibliographystyle{abbrv}
\bibliography{references}
\appendix
\chs \section{Importance of a tight storage bound}\label{app:improved_bound}
To demonstrate the importance of tightening the storage relaxation in the discrete-time setting we will look at the service network design (SND) problem \cite{Boland1} with the addition of hard node and arc capacities. In a discrete-time instance of SND, we are given a directed graph $D = (N,A)$, where each arc $a \in A$ has an associated transit time $\tau_a \in \mathbb{N}_{\geq 0}$, a per-unit-flow cost $c_a \in \mathbb{R}_{\geq 0}$, a fixed cost $f_a \in \mathbb{R}_{\geq 0}$, and a capacity $u_a \in \mathbb{N}_{> 0}$. In addition, we add hard node and arc capacities. Specifically, we are given a limit of $h_a$ trucks that can be sent along $a$ at any (integer) point in time, and each node $v \in N$ can store at most $b_v$ units at any time. 

Let $\mathcal{K}$ denote a set of commodities, each with a source $s_k \in N$ and sink $t_k \in N$, along with a demand $q_k$ that must be routed along a single trajectory from $s_k$ to $t_k$ (the flow is not splittable). Let $r_k$ and $d_k$ release time and deadline for commodity $k \in \mathcal{K}$ respectively. An $s_k,t_k$-trajectory is \emph{feasible} for commodity $k$ if it departs $s_k$ no earlier than $r_k$ and arrives at $t_k$ no later than $d_k$. The goal of SND is to determine a feasible trajectory for each commodity in order to minimize the total fixed and variable cost, ensuring that the hard node and arc capacities are satisfied.

We construct our instance of SND as follows. The base graph, $D$, is the directed graph depicted in Figure \ref{fig:better_relaxation}. Each arc in the figure is also labelled with its transit time and truck capacity. The nodes $v$ and $w$ are labelled with their storage capacity, and the nodes $s$ and $t$ have no storage limit. Let $T \geq 4$. 
\begin{itemize}[noitemsep]
    \item $f_a = 1, c_a = 0$, and $h_a = 1$ for all $a \in A$;
    \item $(s_1, t_1) = (s_2, t_2) = (s,t)$;
    \item $q_k = 50, r_k = 0$ and $d_k = T$ for $k \in \{1, 2\}$.
\end{itemize}
\begin{figure}[ht]
    \begin{center}
    \scalebox{0.4}{
    \includegraphics{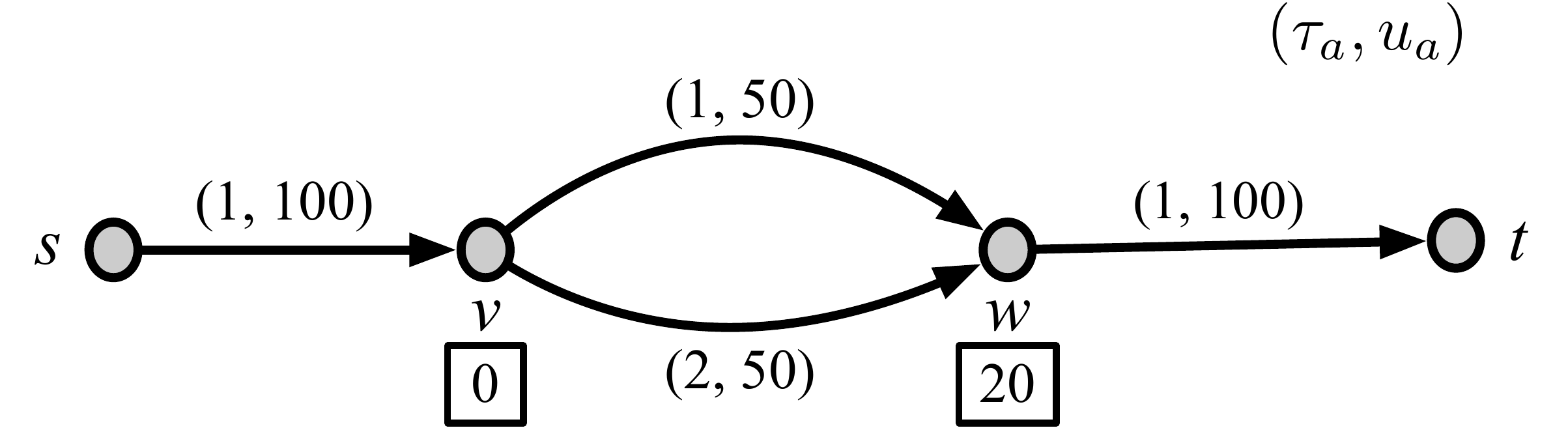}}
    \caption{}
    \label{fig:better_relaxation}
\end{center}\end{figure}
\FloatBarrier

\noindent Observe that while we would like to purchase only a single truck on each of arcs $sv$ and $wt$, this is not possible since the two $vw$ arcs have differing transit times and there is  insufficient storage capacity at $v$ and $w$ to hold the units travelling along the quicker route. Thus, the minimum cost of a solution is 5. 

Consider an iteration of DDD which includes all time copies of $s$, all copies of $v$, and copies of the node $w$ only at times 0, 1, 2, and $T$. Note that it is not hard to make this the minimal set of timed nodes permitted in a partially time-expanded network by adding commodities with restrictive release times and deadlines. 

When applied to this discrete setting, the storage relaxation of Lagos et al. \cite{LagosCIR} would dictate that the storage capacity of $(w,2)$ is set to infinity since $(w,3) \notin N_S$, leading to a solution in the partially time-expanded network with cost 4. In the refinement process of Lagos et al., when storage is violated at a timed node $(w,t)$, then $(w,t+1)$ is added to $N_S$. However, the DDD algorithm would only terminate once $\Omega(T)$ copies of $w$ are added to the solution since every feasible time that $w$ could be reached via commodities 1 and 2 must be included, as well as the following hour to prevent a relaxation of the storage capacity. This would result in $\Omega(T)$ iterations for the DDD algorithm. 

Without the proof of Lemma \ref{lemma:x1-c}, the combination of Lemmas \ref{lemma:x1-b} and \ref{lemma:x2} would result in the following weaker (in terms of strength of the relaxation) relaxed capacity property:

\noindent ($P^{\mathtt{storage-relaxed}}$): For any $e = ((v,t), (v,t')) \in H_S$, 
\[b'_e \geq \mathtt{m_S}(v,t) \cdot b_v + U_e,\]
where 
\[ U_e = \sum\limits_{(w,t') \in N^-_T(v,t) \cup N^-_S(v,t)} u_{wv} \cdot (\mathtt{m_S}(w,t') - 1).\]

However, this storage relaxation gives a similarly poor bound on the number of iterations in a DDD algorithm. Since each timed copy of the neighbouring node $v$ are included in $N_S$, $U_e = 0$ and so we can set the storage capacity of timed node $(w,t)$ to be $\mathtt{m_S}(w,t) \cdot b_w$. Whenever $\mathtt{m_S}(w,t) \geq 3$, the storage will be at least 50, allowing the partially time-expanded network to permit commodities 1 and 2 to overlap on $sv$ and $wt$. Thus, the DDD algorithm will only terminate once each
copy of $w$ that could be reached via commodities 1 and 2 has $\mathtt{m_S}(w,t) \leq 2$ and so the algorithm would still required $\Omega(T)$ iterations. 

With the tighter storage limit given by ($P^{\mathtt{storage}}$), the capacity at $(w,2)$ is only 40, and thus prevents the infeasible solution. As a result, with storage capacities selected according to ($P^{\mathtt{storage}}$) the DDD algorithm terminates after a single  iteration. \che

\chs \section{Improvement on the continuous setting}\label{app:improved_continuous}
In this section we generalize and tighten the relaxation of static storage constraints in the continuous setting. As previously mentioned, Lagos et al. \cite{LagosDDD} consider CIR with out-and-back routes where each client facility can have at most one waiting vehicle at any point in time. We will refer to this as the yard constraint. Let $v$ be a client node in the network and $w$ its incoming neighbour (which is the unique supply facility in the out-and-back network structure they consider). 

In the lower bound model presented in \cite{LagosDDD}, the yard constraint at a client location $v$ at time $t$ is only included when the timed nodes $(v,t+ \epsilon)$ and $(w, t - \tau_{wv} + \epsilon)$ are in the current partially time-expanded network. We will now demonstrate that the bounds proven in this paper generalize the relaxation of the storage constraint in \cite{LagosDDD} to arbitrary flat networks and tighten the previous result even in the continuous setting.

We restate our storage assignment here for clarity.\\
\noindent ($P^{\mathtt{storage}}$): For any $e = ((v,t), (v,t')) \in H_S$, 
\[b'_e \geq \begin{cases}
        b_v + U_e & \quad \mbox {if} ~(v,t+1) \in N_S\\
        2b_v + U_e & \quad \mbox {if} ~(v,t+1) \notin N_S.\\
        \end{cases}\]
In the work of Lagos et al., the value of $\epsilon$ defines a discretization of the time horizon. Furthermore, since there are no hard throughput constraints in CIR, we set $u_{wv} = \infty$ for all arcs $wv$. Finally, the yard constraint considered by the authors has a unit capacity, so $b_v = 1$ for all client vertices $v$. 

When $(v,t+\epsilon) \in N_S$, it follows that $N^{-}_S(v,t) = N^{-}_T(v,t)$. Then if $(w, t - \tau_{wv} + \epsilon) \in N_S$, we obtain that $U_e = 0$. Thus, for the instances considered in  \cite{LagosDDD}, we obtain the following bound.\\
\noindent ($P^{\mathtt{storage - CIR-OB}}$): For any $e = ((v,t), (v,t')) \in H_S$, where $v$ is a client node,
\[b'_e \geq \begin{cases}
        1 \quad & \mbox{if} ~(v,t+\epsilon) \in N_S \mbox{ and } (w, t - \tau_{wv} + \epsilon) \in N_S\\
        2 + U_e &\mbox{if} ~(v,t+\epsilon) \notin N_S \\
    \end{cases}\]
Note that when $(w, \bar{t} + \epsilon) \in N_S$ for all $(w, \bar{t}) \in N^-_S(v,t) \cup N^-_T(v,t)$, it follows that $U_e = 0$. In such a case, it is sufficient to set the yard capacity to 2 at $(v,t)$ instead of $\infty$ when $(v, t+\epsilon) \notin N_S$. Thus, the bounds presented here generalize and tighten those presented by Lagos et al. \cite{LagosDDD}. \che
\chs \section{Two-phase DDD for geographic setting}\label{app:2_stage}
The two-phase DDD approach was introduced recently by Hewitt \cite{enhanced} as a method to speed-up the initial iterations of DDD. In the first phase, the DDD paradigm solves the LP relaxation of the mixed-integer program. In each iteration in this phase, the LP relaxation of the MIP defined on the partially time-expanded network is solved rather than the MIP itself, in order to save on computation time. The first phase then terminates with an optimal solution to the LP relaxation and a final partially time-expanded network, $D_S^{LP}$. In the second phase, DDD solves the original MIP with the partially time-expanded network initialized to be $D_S^{LP}$. This two-phase method was demonstrated to produce optimal solutions more quickly than the single-phase DDD approach for variants of SND \cite{enhanced, DDDAuto}. 

While two-phase DDD is a promising speed-up strategy, one downside of this approach is that the support of an optimal solution to the LP relaxation may be larger than the support of an optimal solution to the MIP, resulting in iterations with larger partially time-expanded networks compared to the partially time-expanded networks of the single-phase approach. 

In this section we present computational results comparing a two-phase approach and the original DDD approach on the set of geographic instances presented in Section \ref{subsec:geographic}. For the lower bound in each iteration, we solve the LP relaxation induced by $D_S$ and obtain a solution $\hat{x}$ to UPR$(D_S)$ with value $\hat{T}$. Note that the final arrival time of a timed movement arc in the support of $\hat{x}$, denoted $\mathtt{final}(\hat{x})$, could exceed $\hat{T}$ when solving the LP relaxation due to the allowance of fractional variables. Thus, instead of setting $T' = \lceil (1+\alpha) \hat{T}\rceil$ as in the original DDD approach, we set $T' = \lceil (1+\alpha) \mathtt{final}(\hat{x}) \rceil$ in phase one. Furthermore, packets are no longer forced to travel along a single trajectory in the LP relaxation. Let $\hat{\mathcal{Q}}_k = \{\hat{Q}_k^1, \cdots, \hat{Q}_k^{r_k}\}$ be the set of trajectories for packet $k$ in the support of $\hat{x}$, and let $\mathcal{P}_k = \{P_k^1, \cdots, P_k^{r_k}\}$ be the corresponding set of paths in $D$. For each $k \in \mathcal{K}$ we redefine $A_T^k$ and $H_T^k$ as
\begin{align*}
    A_T^k & = \{ ((v, t), (w, t')) \in A_T: vw \in P_k^i \in \mathcal{P}_k\}, \mbox{ and}\\
    H_T^k & = \{((v, t), (v, t')) \in H_T: v \in N(P^i_k), P^i_k \in \mathcal{P}_k\}.
\end{align*}
The proof of correctness for the first phase now follows directly from the proof of correctness for the original DDD approach. However, we note that Corollary \ref{cor:added_nodes} no longer holds for a two-phase approach and instead in the first phase timed nodes $(v,t)$ may be added where $t > T^* + 1$. However, we still observe that for any added timed node $(v,t)$, it must be that $t \leq T$, the provided upper bound. It follows that the maximum number of iterations in phase 1 is at most $|N|T$ (a weakening of Theorem \ref{thm:termination} and Corollary \ref{cor:termination}). 

Overall we found that a two-phase DDD approach was on average \emph{slower} than solving the instance with the original single-phase DDD approach. Specifically, on average the two-phase DDD algorithm was 20.4\% slower than the original DDD algorithm. As previously mentioned, there are a few reasons this is not entirely surprising. While the initial iterations can be solved more quickly if we only solve the LP relaxation, these iterations have very sparse partially time-expanded networks, and so naturally the later iterations dominate the runtime of the algorithm. In the following tables, we see that on average the 2-phase approach requires a total of 8.72 iterations, whereas the single-phase DDD approach terminates after an average 7.91 iterations. Additionally, we find an increase of 10\% in the average value of $|N_S^{final}|/|N_T|$ (the final timed node set over the full timed node set) when using the two-phase approach. 

\subsubsection*{Results}
In Tables \ref{table:summary_time} and \ref{table:summary_iterations} we compare two-phase and single-phase DDD (original) in terms of their runtime, total number of iterations, and average value of $|N_S^{final}|/|N_T|$.
\begin{table}[h!]
\makebox[\textwidth][c]{
    \begin{tabular}{c|rrc|ccc}
\toprule
     & \multicolumn{3}{c}{\chs average runtime (s)} & \multicolumn{3}{c}{\chs average $|N_S^{final}|/|N_T|$} \\
\midrule
\chs UB & \chs original & \chs 2-phase & \chs $\frac{\mbox{2-phase}}{\mbox{original}}$ & \chs original & \chs 2-phase & \chs $\frac{\mbox{2-phase}}{\mbox{original}}$ \\ \hline
\chs $T^*$ & \chs 569  & \chs 616  & \chs 1.08 & \chs 0.58 & \chs 0.63 & \chs 1.08 \\
\chs $1.5T^*$ & \chs 1,019 & \chs 1,206 & \chs 1.18 & \chs 0.39 & \chs 0.44 & \chs 1.13 \\
\chs $2T^*$ & \chs 1,094 & \chs 1,422 & \chs 1.30 & \chs 0.30 & \chs 0.33 & \chs 1.10\\
\bottomrule
\end{tabular}%
}
\caption{Average runtime and average value of $|N_S^{final}|/|N_T|$.}
\label{table:summary_time}%
\end{table}%
\vspace{-.25cm}
\FloatBarrier

Across all upper bound factors ($T \in \{T^*, 1.5T^*, 2T^*\}$), on average the original DDD approach terminates more quickly than the two-phase approach. One contributing factor is the fact that on average, the size of the final timed node set is larger when running two-phase DDD, resulting in slower final iterations. Additionally, in Table \ref{table:summary_iterations}, we see that the total number of iterations increases when running the two-phase approach. After solving the LP relaxation with DDD, the two-phase approach required an average of approximately 2 iterations in the second phase. Since the later iterations dominate the runtime of DDD and the two-phase approach resulted in slower final iterations, this explains why the two-phase approach does not offer improvement over the original DDD approach for this particular problem and DDD implementation. 

\begin{table}[h!]
\makebox[\textwidth][c]{
    \begin{tabular}{c|ccc|c}
\toprule
     & \multicolumn{3}{c}{\chs two-phase DDD} & \multicolumn{1}{c}{\chs original DDD} \\
\midrule
\chs UB & \chs phase 1 &\chs  phase 2 & \chs total & \chs total \\ \hline
\chs $T^*$ & \chs 6.62 & \chs 2.04 & \chs 8.65 & \chs 7.86 \\
\chs $1.5T^*$ & \chs 6.63 & \chs 2.27 & \chs 8.90 & \chs 7.85 \\
\chs $2T^*$ & \chs 6.57 & \chs 2.04 & \chs 8.60 & \chs 8.01 \\
\bottomrule
\end{tabular}%
}
\caption{Average number of iterations in each phase.}
\label{table:summary_iterations}%
\end{table}%
\FloatBarrier

The following figures present the cumulative instances solved whe $T = 1.5T^*$  and $T = 2T^*$. 
\begin{figure}[h!]
    \centering
    \begin{minipage}{.5\textwidth}
        \centering
        \includegraphics[width=\textwidth]{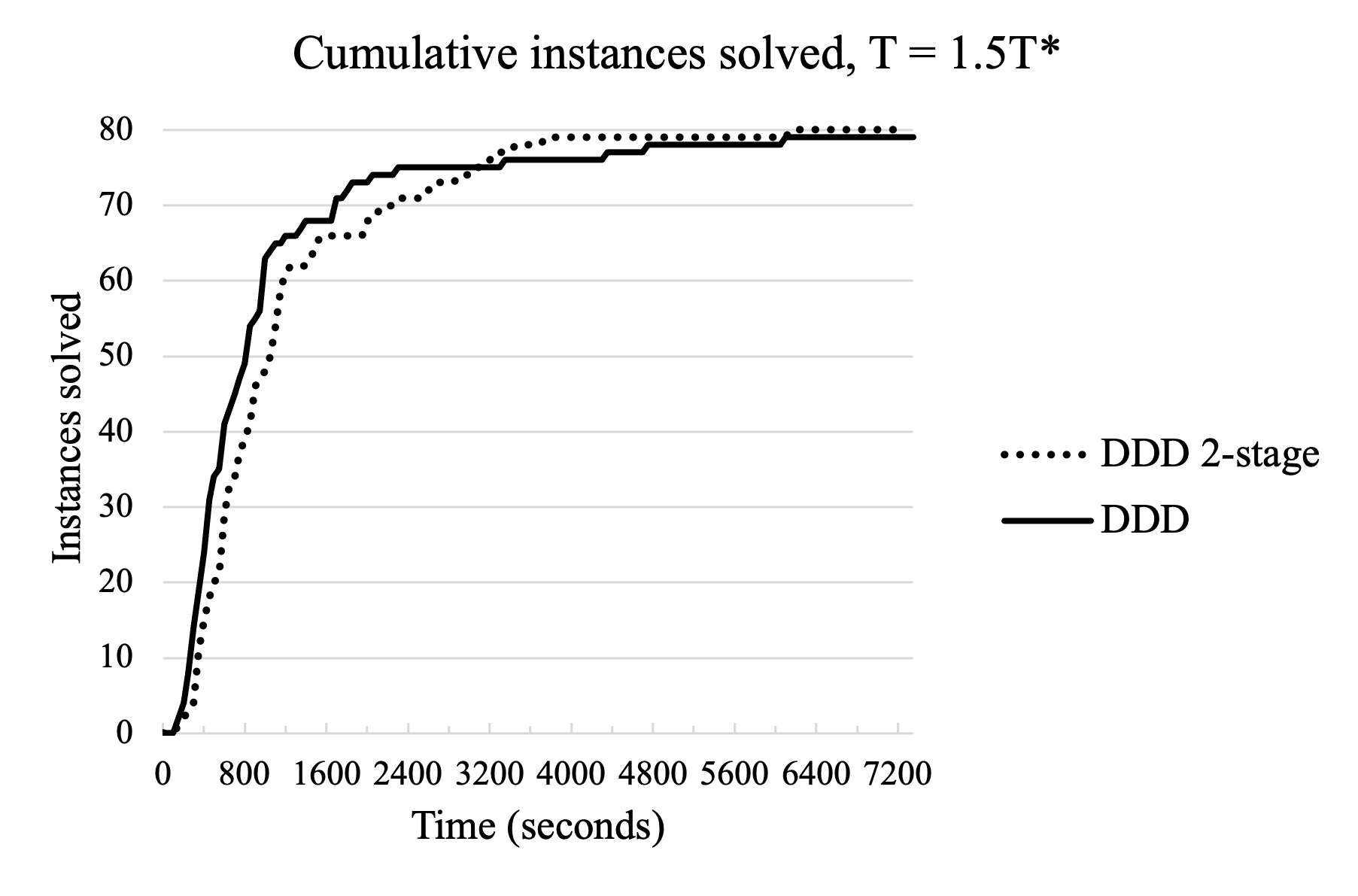}
        \caption{$T=1.5T^*$}
        \label{}
    \end{minipage}%
    \begin{minipage}{0.5\textwidth}
        \centering
        \includegraphics[width=\textwidth]{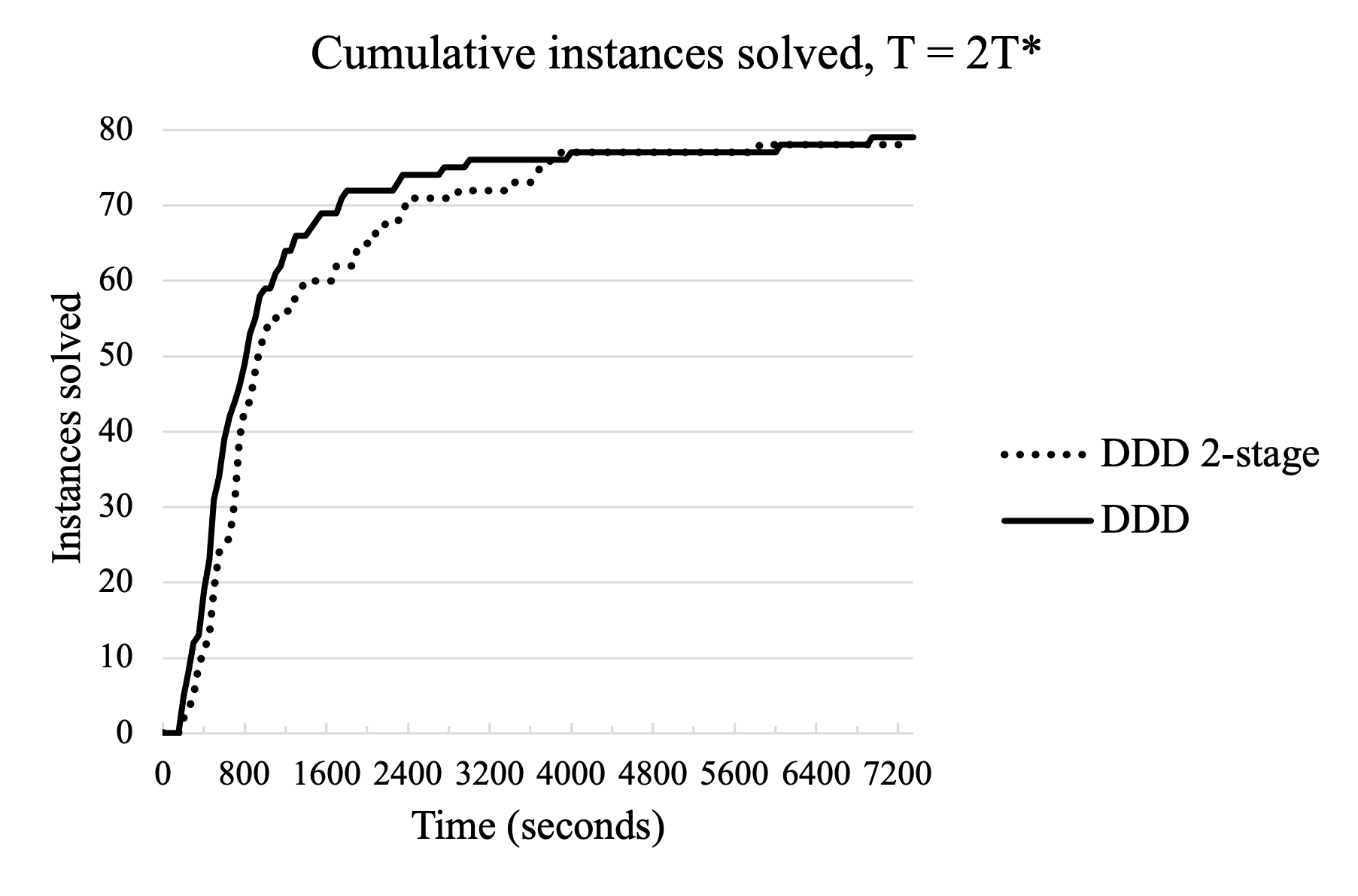}
        \caption{$T=2T^*$}
        \label{}
    \end{minipage}
\end{figure}
\FloatBarrier

The following tables provide a breakdown of the performance of the two-phase DDD for each value of $m$, the number of arcs in the base graph. Each entry in the table is equal to the average value of that feature for the two-phase approach divided by the average value of the feature for the original single-phase DDD approach. For example, among the instances where $m = 30, k = 200$, and $T = T^*$, the average number of iterations for two-phase DDD was 11.44 and the average number of iterations for the original DDD approach was 10.56, giving a ratio of 11.44/10.56 = 1.08.

\begin{table}[h!]
\makebox[\textwidth][c]{
    \begin{tabular}{c|ccc|ccc|ccc}
\toprule
      & \multicolumn{1}{c}{} & \multicolumn{1}{c}{\chs $k = 200$} & \multicolumn{5}{c}{\chs $k = 250$} & \multicolumn{1}{c}{\chs $k = 300$} \\
\midrule
\chs UB & \chs time & \chs iterations & \chs $\frac{|N_S^{final}|}{|N_T|}$ & \chs time & \chs iterations & \chs $\frac{|N_S^{final}|}{|N_T|}$ & \chs time & \chs iterations & \chs $\frac{|N_S^{final}|}{|N_T|}$ \\ \hline
\chs $T^*$ & \chs 0.69 & \chs 1.08 & \chs 1.06 & \chs 0.89 & \chs 1.00 & \chs 1.03 & \chs 1.12 & \chs 1.15 & \chs 1.38 \\
\chs $1.5T^*$ & \chs 0.88 & \chs 1.10 & \chs 1.05 & \chs 1.56 & \chs 1.15 & \chs 1.27 & \chs 1.20 & \chs 1.24 & \chs 1.14 \\
\chs $2T^*$ & \chs 1.89 & \chs 1.02 & \chs 1.10 & \chs 1.46 & \chs 1.09 & \chs 1.24 & \chs 1.27 & \chs 1.19 & \chs 1.06 \\
\bottomrule
\end{tabular}%
}
\caption{$m = 30$.}
\label{tab:addlabel}%
\end{table}%

\begin{table}[h!]
\makebox[\textwidth][c]{
    \begin{tabular}{c|ccc|ccc|ccc}
\toprule
      & \multicolumn{1}{c}{} & \multicolumn{1}{c}{\chs $k = 200$} & \multicolumn{5}{c}{\chs $k = 250$} & \multicolumn{1}{c}{\chs $k = 300$} \\
\midrule
\chs UB &\chs  time & \chs iterations & \chs $\frac{|N_S^{final}|}{|N_T|}$ & \chs time & \chs iterations & \chs $\frac{|N_S^{final}|}{|N_T|}$ & \chs time & \chs iterations & \chs $\frac{|N_S^{final}|}{|N_T|}$ \\ \hline
\chs $T^*$ & \chs 1.34 & \chs 1.12 & \chs 1.07 & \chs 1.36 & \chs 1.14 & \chs 1.08 & \chs 1.43 & \chs 1.19 & \chs 1.11 \\
\chs $1.5T^*$ & \chs 1.23 & \chs 1.07 & \chs 1.08 & \chs 1.58 & \chs 1.17 & \chs 1.13 & \chs 1.75 & \chs 1.20 & \chs 1.20 \\
\chs $2T^*$ & \chs 1.01 & \chs 1.03 & \chs 1.07 & \chs 2.03 & \chs 1.17 & \chs 1.15 & \chs 1.32 & \chs 1.10 & \chs 1.10\\
\bottomrule
\end{tabular}%
}
\caption{$m = 45$.}
\label{tab:addlabel}%
\end{table}%

\begin{table}[h!]
\makebox[\textwidth][c]{
    \begin{tabular}{c|ccc|ccc|ccc}
\toprule
      & \multicolumn{1}{c}{} & \multicolumn{1}{c}{\chs $k = 200$} & \multicolumn{5}{c}{\chs $k = 250$} & \multicolumn{1}{c}{\chs $k = 300$} \\
\midrule
\chs UB & \chs time & \chs iterations & \chs $\frac{|N_S^{final}|}{|N_T|}$ & \chs time & \chs iterations & \chs $\frac{|N_S^{final}|}{|N_T|}$ & \chs time &\chs  iterations & \chs $\frac{|N_S^{final}|}{|N_T|}$ \\ \hline
\chs $T^*$ & \chs 1.08 & \chs 1.02 & \chs 0.99 &\chs  1.51 & \chs 1.20 & \chs 1.10 & \chs 1.08 & \chs 1.05 & \chs 1.00 \\
\chs $1.5T^*$ & \chs 1.14 & \chs 1.06 & \chs 1.12 &\chs  1.52 & \chs 1.15 & \chs 1.12 & \chs 0.66 & \chs 1.05 & \chs 1.11 \\
\chs $2T^*$ &\chs  1.16 &\chs  0.98 & \chs 1.08 & \chs 1.15 & \chs 1.02 & \chs 1.06 & \chs 0.90 & \chs 1.05 & \chs 1.12 \\
\bottomrule
\end{tabular}%
}
\caption{$m = 60$.}
\label{tab:addlabel}%
\end{table}%
\FloatBarrier
Consistently we see that the average number of iterations and the average value of $|N_S^{final}|/|N_T|$ increases for the two-phase DDD approach. The same can be said for the number of iterations. 

Note that the two-phase approach is not necessarily well-defined even if the DDD algorithm is well-defined for a given MIP. Specifically, without the addition of constraint (\ref{constr:two_stage_IP_S}), adding timed nodes to $D_S$ could cause the value of the LP relaxation to \emph{decrease} for UPR. Furthermore, it could be the case that the refinement procedure impacts the effectiveness of a two-phase approach, since an aggressive refinement process that corrects all short arcs in the support of a solution may add more timed nodes when given a fractional solution with larger support. Therefore it would be interesting to understand how the formulation and the refinement process can impact the effectiveness of a two-phase DDD approach.\che

\end{document}